\documentclass[11pt]{article}
\usepackage[utf8]{inputenc} 
\usepackage[round]{natbib}
\usepackage{graphicx}
\makeatletter
\let\c@lofdepth\relax
\let\c@lotdepth\relax
\makeatother
\usepackage{subfigure}
\usepackage{float}
\usepackage{caption}
\usepackage{verbatim}
\usepackage{multirow}
\usepackage{amsmath}
\usepackage{adjustbox}
\usepackage{tabularx}
\usepackage{xcolor}
\usepackage{authblk}
\usepackage{geometry}
\usepackage{setspace}
\usepackage{hyperref}
\usepackage{array}
\usepackage{bm}
\usepackage{lipsum}
\usepackage{quoting}
\usepackage{indentfirst}
\usepackage{eqnarray}
\usepackage{rotating}
\usepackage{textcomp}
\usepackage{booktabs}
\usepackage{amssymb,amsfonts,amsthm}
\usepackage{enumerate}
\usepackage{graphics}
\usepackage[dvips]{epsfig}
\usepackage{multicol}
\usepackage{url}
\usepackage{icomma}

\newtheorem{theorem}{Theorem}
\newtheorem{corollary}{Corollary}
\DeclareMathOperator*{\argmax}{arg\,max}

\title{Unit-Modified Weibull Distribution and Quantile Regression Model}

\author{
João Inácio Scrimini\thanks{Graduate Program in Industrial Engineering, Federal University of Santa Maria, Roraima Avenue, 1000, 97105-900, Santa Maria, RS, Brazil. Email: joao.scrimini@acad.ufsm.br} ,
Cleber Bisognin\thanks{Department of Statistics, Federal University of Santa Maria, Roraima Avenue, 1000, 97105-900, Santa Maria, RS, Brazil. Email: cleber.bisognin@ufsm.br} ,
Renata Rojas Guerra\thanks{Department of Statistics, Federal University of Santa Maria, Roraima Avenue, 1000, 97105-900, Santa Maria, RS, Brazil. Email: renata.r.guerra@ufsm.br} \,\,and
Fábio M. Bayer\thanks{Department of Statistics and LACESM, Federal University of Santa Maria, Roraima Avenue, 1000, 97105-900, Santa Maria, RS, Brazil. Email: bayer@ufsm.br}
}

\date{} 

\makeindex

\begin{document}

\renewcommand{\thefootnote}{\arabic{footnote}} 
\maketitle

\begin{abstract}
 The Sustainable Development Goals (SDGs) of the United Nations consist of 17 general objectives, subdivided into 169 targets to be achieved by 2030. Several SDG indices and indicators require continuous analysis and evaluation, and most of these indices are supported in the unit interval \((0,1)\). To incorporate the flexibility of the modified Weibull (MW) distribution in doubly constrained datasets, the first objective of this work is to propose a new unit probability distribution based on the MW distribution. For this, a transformation of the MW distribution is applied, through which the unit modified Weibull (UMW) distribution is obtained.
The second objective of this work is to introduce a quantile regression model for random variables with UMW distribution, reparameterized in terms of the quantiles of the distribution. 
Maximum likelihood estimators (MLEs) are used to estimate the model parameters. 
Monte Carlo simulations are performed to evaluate the MLE properties of the model parameters in finite sample sizes. 
The introduced methods are used for modeling some 
sustainability indicators related to the SDGs, also considering the reading skills of dyslexic children, which are indirectly associated with SDG 4 (Quality Education) and SDG 3 (Health and Well-Being).

\vspace{1.2cm}
 
 \noindent
\noindent\textbf{Keywords:} Maximum Likelihood, Monte Carlo Simulation, Quantile Regression, Unit Distribution. 
\end{abstract}



\section*{Introduction}\label{introduction}

The United Nations Sustainable Development Goals (SDGs) are global plans to promote sustainable development, improve living conditions, and foster peace. 
The SDGs comprise 17 overarching objectives and 169 targets to be achieved by 2030. 
One challenge with the modeling indices used to monitor the progress of the SDGs is that many are limited to the range $(0,1)$, where usual statistical distributions are not appropriate.  
This limitation highlights the need for new, more flexible statistical distributions that can better accommodate such data, ensuring greater precision in assessing and monitoring outcomes. Additionally, "spin-off" indicators, while not included in the official list, play a significant role in supporting the monitoring and promotion of the SDGs. 
These indicators provide complementary insights that align with the established objectives. 
Consequently, developing more flexible probability distributions tailored to these data, along with predictive models capable of generating more reliable inferences with reduced uncertainty, becomes essential.

In statistical analysis, identifying an appropriate distribution for modeling data sets is crucial. 
By selecting the distribution that best fits or describes the behavior of a specific data set, more accurate inferences can be made. 
To this end, new techniques have been developed to modify existing statistical distributions, enhancing their flexibility for modeling data sets emerging across various fields of study. 
These more flexible statistical distributions offer greater adaptability to different patterns observed in the data, enabling their characteristics to be modeled more accurately. 
As a result, they become more suitable for a variety of situations and improve the quality of statistical inferences, increasing the accuracy of estimates and predictions.

Several distributions in the literature have garnered considerable attention in recent years. 
For example, the two-parameter Weibull distribution, proposed by \cite{weibull1951}, has a wide range of applications in different scientific fields. 
It is generally employed for lifetime analysis, hazard rates, reliability studies, and similar contexts \citep{Lai2014}. 
However, the Weibull model is inadequate for describing non-monotonic failure rates, such as those exhibiting bathtub-shaped or inverted bathtub patterns in their hazard functions \citep{shama2023}. 
Numerous modifications have been proposed to enhance the flexibility of the Weibull distribution. 
One such modification is the three-parameter modified Weibull distribution introduced by \cite{MW2003}, which is capable of modeling bathtub-shaped lifetime and hazard rate data. 
Despite its versatility, the modified Weibull distribution is limited by its support on positive real values, preventing the bathtub-shaped characteristic from manifesting in its density function. 
This limitation arises because the density approaches zero asymptotically as values diverge from zero to infinity. 
To address this, additional transformations of the Weibull distribution have been introduced, including the generalized modified Weibull distribution for lifetime modeling \citep{carrasco2008generalized}; 
the modified Weibull beta distribution \citep{silva2010beta}, 
primarily applied to survival data; 
the additive modified Weibull distribution \citep{he2016additive}; 
and the alpha power Weibull transformation distribution, used to describe the behavior of electronic devices under voltage stress profiles \citep{mendez2022alpha}, among others.

In the context of regression models, the normal linear regression model is the best known and most widely used, assuming normally distributed additive errors. 
Alternatively, generalized linear models (GLMs) \citep{nelder1972generalized} assume that the variable of interest follows a distribution from the canonical exponential family,  
which includes the normal, Poisson, negative binomial, gamma, and inverse normal distributions. However, in practical applications, response variables may not always conform to these distributions. 
To address this limitation, new regression models have been proposed as alternatives to both linear regression models and GLMs. 
For data constrained within a limited range, as several SDG indices are bounded in $(0,1)$, the most common distributions include: 
the beta distribution \citep{johnson1995}, for which the beta regression model was introduced by \cite{Ferrari2004}; 
the Kumaraswamy distribution \citep{kumaraswamy1980}, with the Kumaraswamy regression model incorporating the Aranda-Ordaz link function \citep{Pumi2020}; 
the simplex distribution \citep{barndorff1991}, along with its respective regression model proposed by \cite{song2004modelling}, among others. 
A common feature of these models is the reparameterization in terms of the mean or median of the distribution, enabling parameter interpretation in terms of position and/or precision metrics. Through these reparameterizations, a regression structure is introduced to model the mean or median, following a similar approach to GLMs. In general, the median is often more robust than the mean when the variable of interest exhibits asymmetric behavior or contains outliers. In such cases, modeling the median instead of the mean tends to yield better results \citep{john2015,lemonte2016}.

Recent studies have introduced quantile regression models tailored to various data structures. 
For positive continuous responses, notable examples include models based on the Birnbaum–Saunders distribution \citep{gallardo2024}, 
the Dagum and Singh–Maddala distributions \citep{saulo2023}, 
and gamma-based \citep{bourguignon2025}. 
In the unit context, examples include the beta \citep{bourguignon2024parametric}, 
unit log–log \citep{korkmaz2023unit}, 
unit power half-normal \citep{santoro2024unit}, 
unit generalized half-normal \citep{mazucheli2023unit}, 
and unit Burr-XII \citep{korkmaz2021unit,ribeiro2022another} distributions. 
More recently, unit Weibull-type distributions have also been proposed by \cite{de2024covid,abubakari2024unit,sapkota2025new}. 
Moreover, a comprehensive review of unit quantile regression models is presented in \citet{mazucheli2022overview}. 

To explore the flexibility of the modified Weibull distribution for modeling doubly constrained data sets, this work proposes a new distribution with support in the interval $(0,1)$, derived from the modified Weibull distribution.
By transforming it into a unitary distribution, the characteristic bathtub-shaped hazard function becomes representable in the density function.
This transformation also allows the new distribution to capture increasing-decreasing-increasing density patterns, enhancing its flexibility for modeling complex data behaviors.
Furthermore, we introduce a regression model based on the quantiles of this new unitary distribution, which can accommodate asymmetric behaviors, bathtub-shaped patterns, and increasing-decreasing-increasing shapes through the incorporation of exogenous variables.
Inference on the parameters of the proposed models is conducted using maximum likelihood estimation.
To evaluate the performance of the inference procedures, Monte Carlo simulations are performed, computing the bias and mean square error of the point estimators, as well as the 95\% coverage rates of the confidence intervals.
Finally, to assess the goodness of fit of the proposed models to real-world data, diagnostic tools based on quantile residuals are explored.

\section*{Proposed Models\label{sec:PropM}}

In this section, the Unit-Modified Weibull distribution will be presented, along with the quantile regression model based on this distribution.

\subsection*{Unit-Modified Weibull Distribution}

Let \(X\) be a random variable with the modified Weibull distribution, denoted by \(\text{MW}(\alpha, \gamma, \lambda)\), proposed by \cite{MW2003}. The cumulative distribution function (CDF) of the \(\text{MW}\) distribution is given by
\begin{equation}\label{eq1}
    F_X(x)=1-\exp\left(-\alpha x^\gamma \exp(\lambda x)\right),
\end{equation}
where $x, \alpha, \gamma >0$ {and} $\lambda \geq 0$. 
By deriving Equation \eqref{eq1} with respect to \(x\), we obtain the probability density function (PDF) of the \(\text{MW}(\alpha, \gamma, \lambda)\) distribution, which is given by
\begin{equation}\label{eq2}
    f_X(x)=\alpha x^{\gamma-1}\left(\lambda x+\gamma\right)\exp\left(\lambda x-\alpha x^\gamma \exp(\lambda x)\right),
\end{equation}
where \(\alpha\) is the scale parameter, \(\gamma\) is the shape parameter, and \(\lambda\) is the acceleration parameter, which acts as an accelerating factor in the time of imperfection and functions as a fragility factor in the survival of the individual as time increases. The MW distribution has some particular cases: when \(\lambda = 0\), we obtain the Weibull distribution \citep{weibull1951}; when \(\alpha = 1\) and \(\gamma = 0\), we obtain the extreme value distribution \citep{bain1974analysis}; and when \(\lambda = 0\) with \(\gamma = 1\) and \(\gamma = 2\), we obtain the Exponential and Rayleigh distributions \citep{bain1974analysis}, respectively. 
Additionally, the modified Rayleigh distribution can be defined for \(\gamma = 2\) and the modified exponential distribution for \(\gamma = 1\). While these distributions have not been extensively explored in the literature, they are recognized by \cite{silva2010beta} as special cases of the modified Weibull distribution.

Considering the transformation 
\(Y = \mathrm{e}^{-X}\),
where \(X \sim \text{MW}(\alpha, \gamma, \lambda)\), whose CDF and PDF are given by Equations \eqref{eq1} and \eqref{eq2}, respectively,  
we propose the Unit-Modified Weibull distribution, denoted by \(\text{UMW}(\alpha, \gamma, \lambda)\).  
The CDF and PDF of the new distribution are given by, respectively,  
\begin{equation}\label{eq3}
F_Y(y)=\exp\left(-\alpha\left[-\log(y)\right]^\gamma y^{-\lambda}\right)
\end{equation}
and
\begin{equation}\label{eq4}
 f_Y(y)=   \frac{\alpha\left[-\log(y)\right]^{\gamma}}{\log\left(y\right)y^{\lambda+1}}\left[\lambda\log(y)-\gamma\right]\exp\left(-\alpha\left[-\log(y)\right]^\gamma y^{-\lambda}\right),
\end{equation}
for \( y \in (0, 1) \), where \(\alpha, \gamma > 0\) and \(\lambda \geq 0\). 
The \(\alpha\) represents the scale parameter, while 
\(\gamma\) and \(\lambda\) are shape parameters. 
These characteristics are illustrated in Figure \ref{fig:density1}, where the flexibility of the distribution is also evident, particularly with respect to the parameter \(\gamma\). 
Notably, when \(\gamma <1\), the distribution exhibits an increasing-decreasing-increasing pattern, also known as a bathtub-shaped behavior. Therefore, Equation \eqref{eq3} extends at least four unit distributions, 
incorporating as submodels some distributions that, to the best of our knowledge, have not yet been explored in the literature. Specifically: when \(\alpha = 1\) and \(\gamma = 0\), the unit-extreme value; when \(\lambda = 0\), the unit-Weibull \citep{weibullunit}; when \(\lambda = 0\) and \(\gamma = 2\), the unit-Rayleigh distribution \citep{bantan2020some}; and when \(\gamma = 2\), the unit-modified Rayleigh distribution. Additionally, the UMW distribution is part of the unit extended Weibull family \citep{guerra2021unit}.

\begin{figure}
\centering
\subfigure[$\gamma=1.2$ 
and $\lambda=0.7$]{\includegraphics[width=0.45\textwidth]{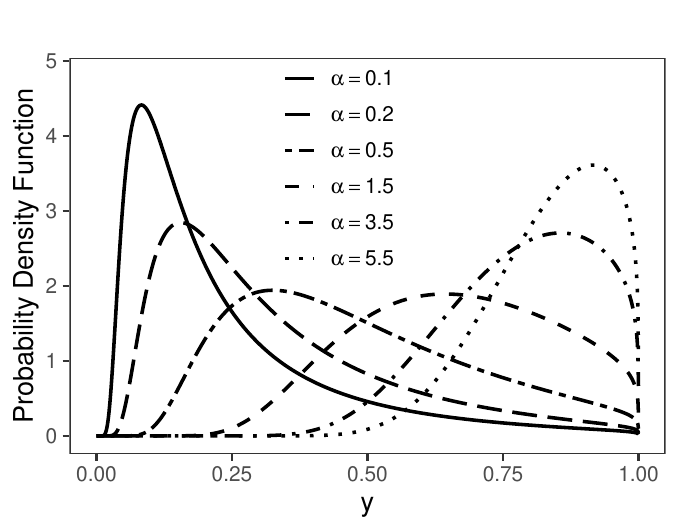}}
\subfigure[$\alpha=0.4$ 
and $\lambda=0.3$]{\includegraphics[width=0.45\textwidth]{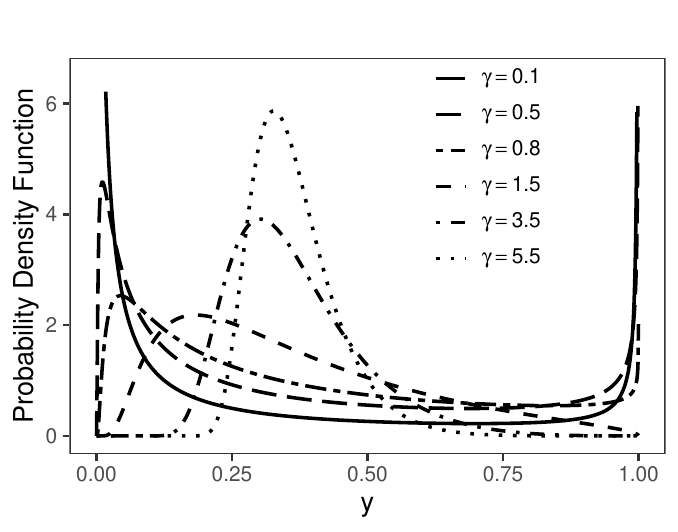}}\\[-0.05cm]
\subfigure[$\alpha=0.6$ {and} $\gamma=1.5$]{\includegraphics[width=0.45\textwidth]{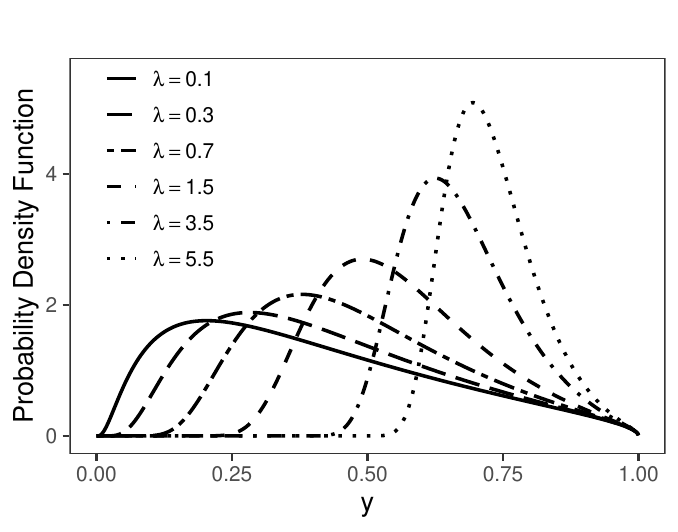}}
\caption{The probability density function of the \(\mbox{UMW}\) distribution for some values of the parameters $(\alpha,\gamma,\lambda)$.}\label{fig:density1}
\scriptsize{Source: Authors.}%
\end{figure}

The UMW distribution is identifiable in the classical sense, meaning that distinct parameter values correspond to distinct probability density functions (PDFs). 
This property can be established by analyzing the logarithm of the density function and its asymptotic behavior with respect to its argument.
Although the full proof is omitted for brevity, identifiability ensures that distinct parameter values yield distinct data distributions, allowing the corresponding estimators to be uniquely determined from the observed data and enabling valid inference.

The quantiles of the \(\mbox{UMW}(\alpha, \gamma, \lambda)\) distribution can be obtained by inverting the CDF, given by Equation \eqref{eq3}, represented by the quantile function \( Q_Y(\tau) = \mu_{\tau}\), 
which can be obtained by solving the following non-linear equation with respect to \(\mu_{\tau}\):
\begin{equation}\label{eq:quantilumw}
\mu_{\tau}^{-\lambda}\left[-\log(\mu_{\tau})\right]^\gamma+\frac{1}{\alpha}\log\left(\tau\right)=0, 
\end{equation}
for a defined quantile $\tau \in (0,1)$.
By evaluating the quantile function at a random variable following a uniform distribution \(U(0,1)\), random numbers from the \(\text{UMW}(\alpha, \gamma, \lambda)\) distribution can be generated.  
This process requires solving the following nonlinear equation for \(Y\):
\[Y^{-\lambda}\left[-\log(Y)\right]^\gamma+\frac{1}{\alpha}\log\left(U\right)=0,\]
where \(U\sim U(0,1)\). 

The hazard rate function of the UMW distribution is given by
\begin{equation}\label{hazard1}
    h_Y(y) = \dfrac{f_Y(y)}{1-F_Y(y)} = \dfrac{{\alpha}y^{-{\lambda}-1}\left[-\log\left(y\right)\right]^{\gamma}\left[{\lambda}\log\left(y\right)-{\gamma}\right]}{\log\left(y\right)\left\{\exp\left(\alpha\left[-\log(y)\right]^\gamma y^{-\lambda}\right)-1\right\}}.
\end{equation}
As shown in Figure \ref{fig:hazard}, it is possible to observe the behavior and flexibility of the hazard rate function, exhibiting bathtub-shaped and increasing-decreasing-increasing characteristics for the different values of the parameters \((\alpha, \gamma, \lambda)\).

\begin{figure}
\centering
{\includegraphics[width=0.49\textwidth]{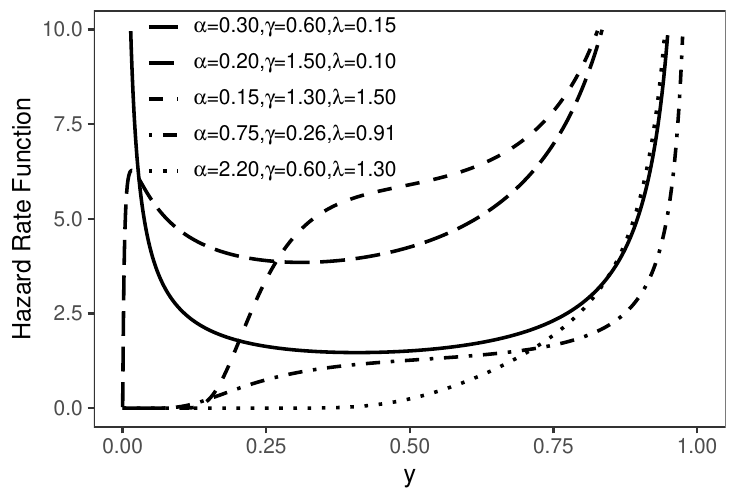}}
\caption{Hazard function of the \(\mbox{UMW}\) distribution for some values of the parameters $(\alpha,\gamma,\lambda)$.}\label{fig:hazard}
\scriptsize{Source: Authors.}%
\end{figure}

The versatility of the new distribution can be highlighted by the formulas for skewness from \cite{bowley1901} and kurtosis from \cite{moors1986}, respectively, given by:
\[\mbox{S} = \frac{{Q_Y(3/4) + Q_Y(1/4) - 2Q_Y(1/2)}}{{Q_Y(3/4) - Q_Y(1/4)}}\]
and
\[\mbox{K} = \frac{{Q_Y(7/8) + Q_Y(3/8) - Q_Y(5/8) - Q_Y(1/8)}}{{Q_Y(3/4) - Q_Y(1/4)}},\]
where \( Q_Y(\cdot) \) represents the calculation of the quantiles of the distribution given by Equation \eqref{eq:quantilumw}. 

We can observe in Figure~\ref{fig:kur} the flexibility of the UMW distribution with respect to skewness and kurtosis, especially when at least one of the parameters has a value less than 1. Under these conditions, the coefficients span a wider range of values, allowing the distribution to capture different shapes and tail structures. It is also noted that higher values of $\alpha$ tend to produce negative or only mild skewness, along with reduced kurtosis, whereas lower values of $\alpha$ lead to an increase in both skewness and kurtosis. This behavior is associated with the presence of heavier tails and a greater concentration of probability in extreme regions. These characteristics highlight the potential of the UMW distribution to model data with varying patterns of skewness and dispersion.

\begin{figure}
\centering
{\includegraphics[width=0.325\textwidth]{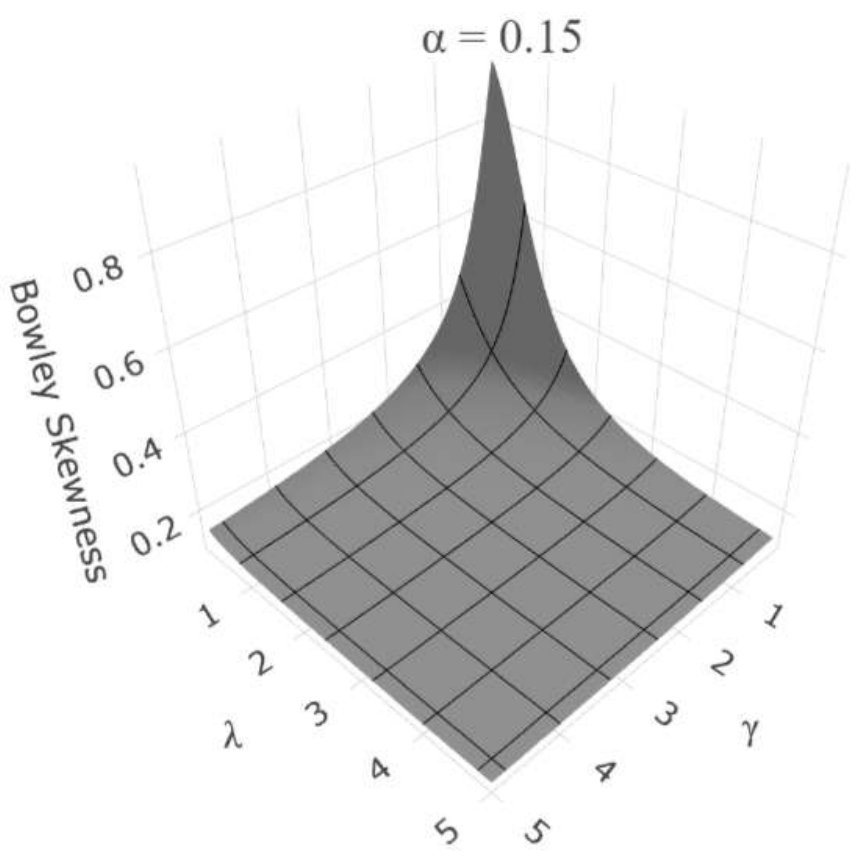}}
{\includegraphics[width=0.325\textwidth]{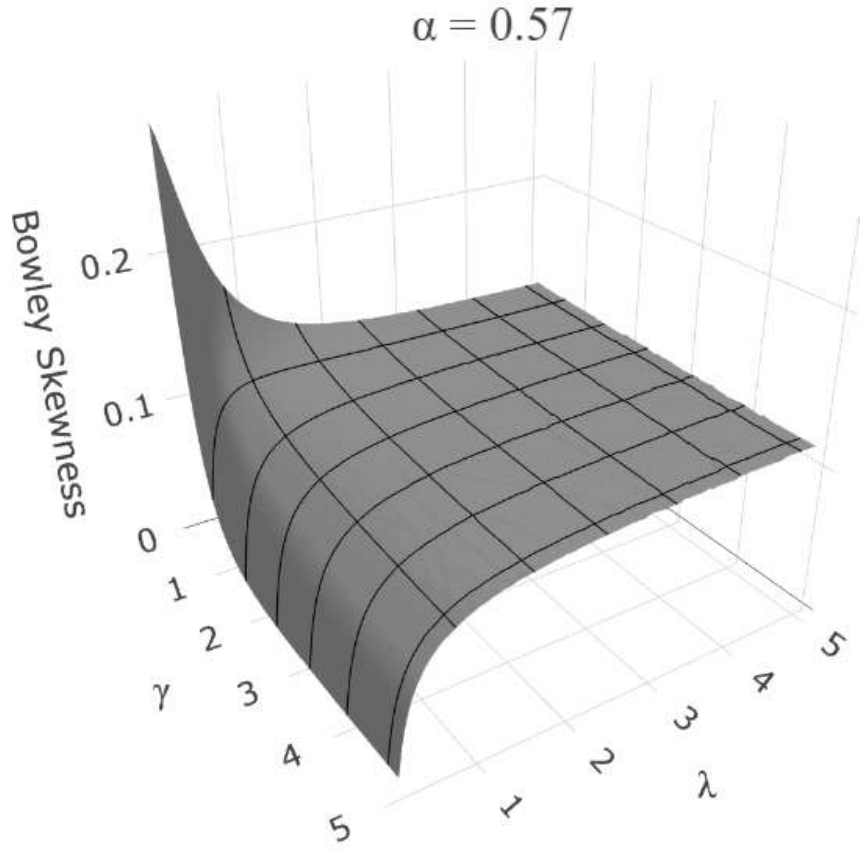}}
{\includegraphics[width=0.325\textwidth]{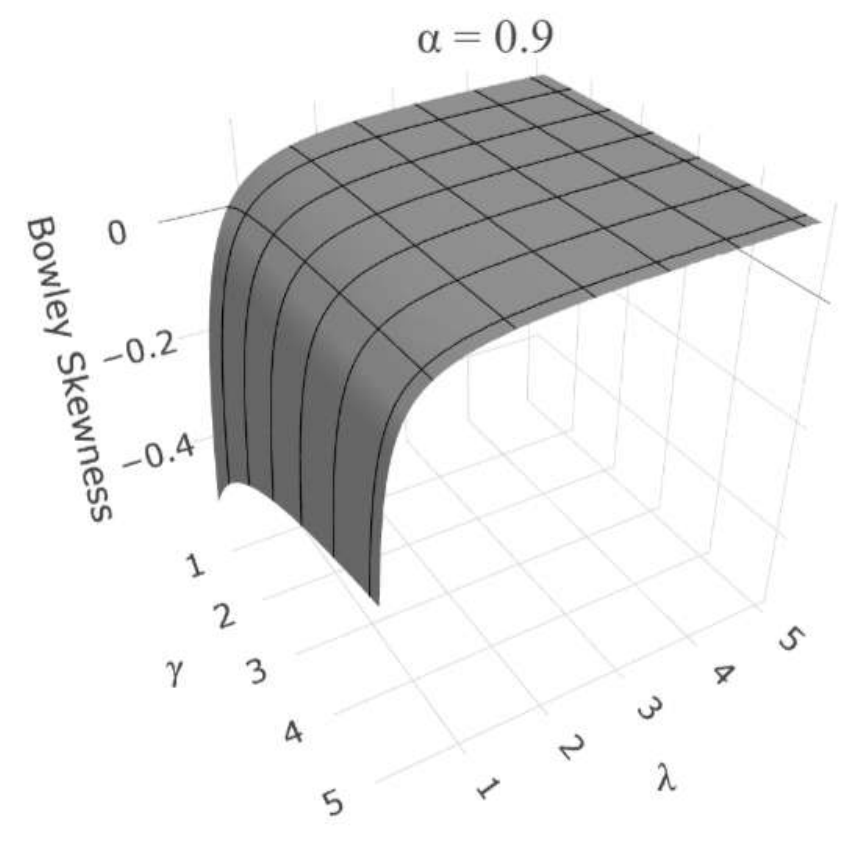}}\\[1em]
{\includegraphics[width=0.325\textwidth]{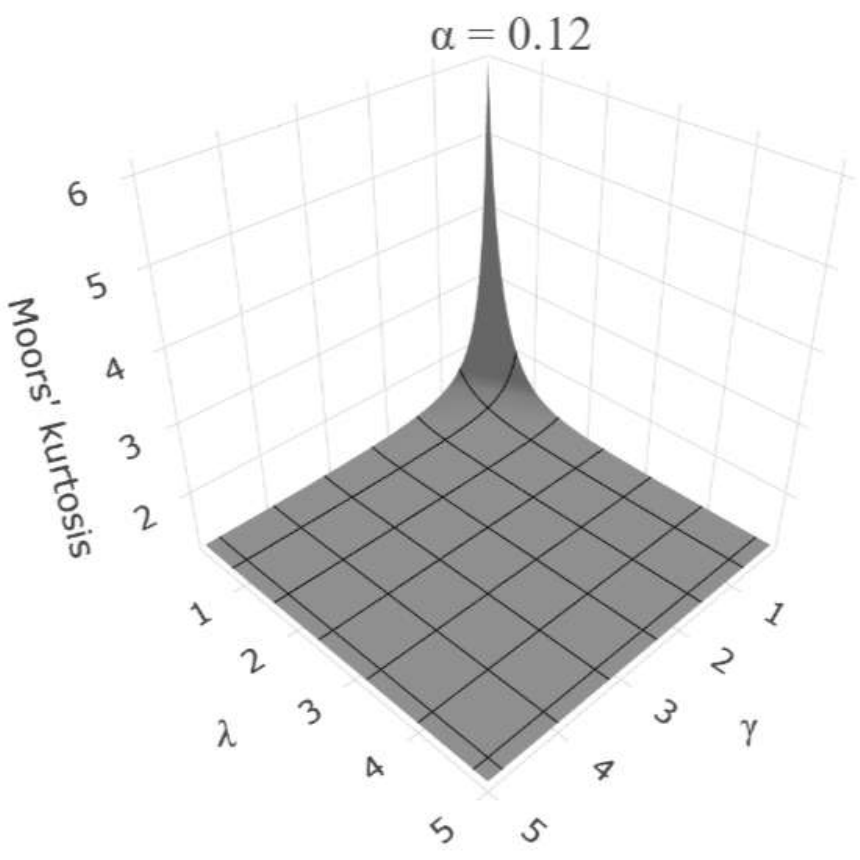}}
{\includegraphics[width=0.325\textwidth]{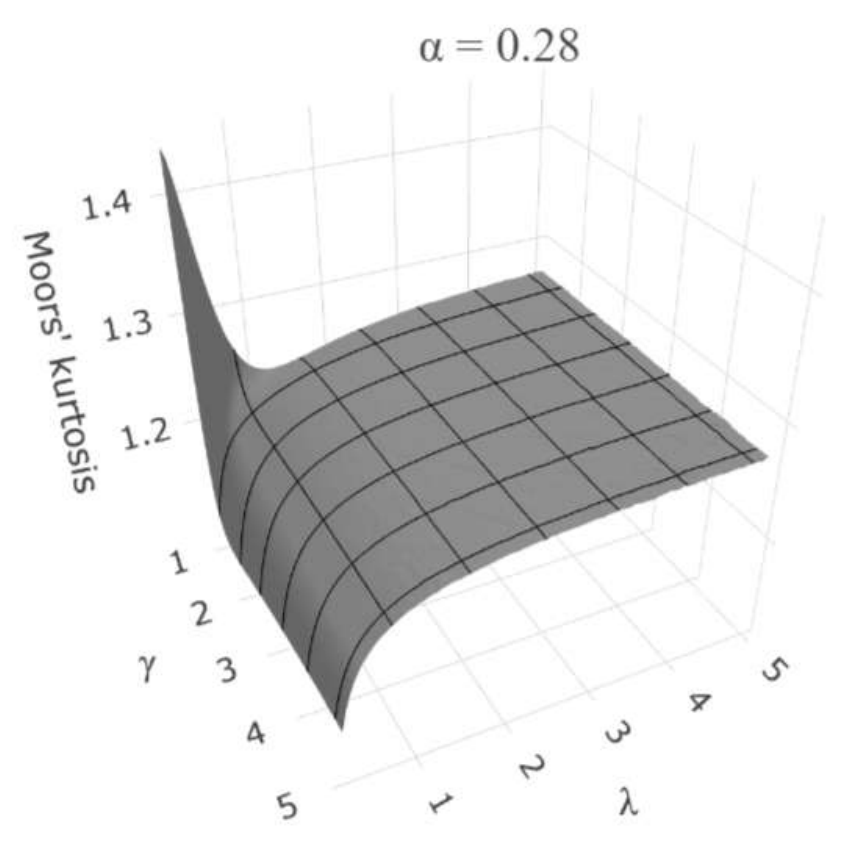}}
{\includegraphics[width=0.325\textwidth]{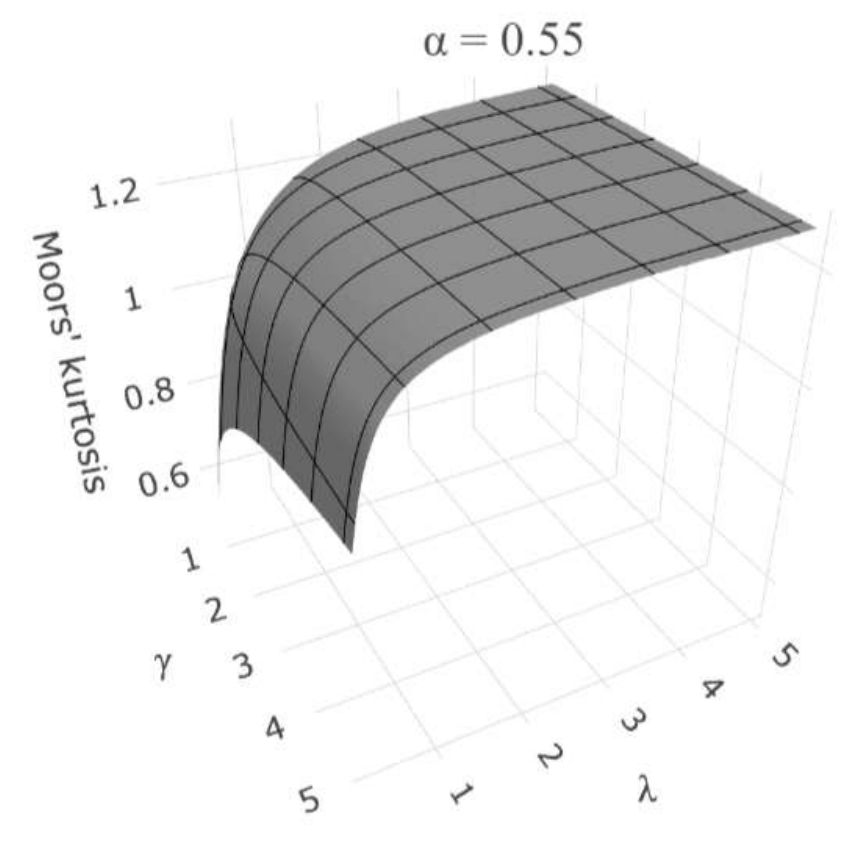}}
\caption{Skewness and kurtosis of the \(\mbox{UMW}\) distribution for some values of the parameters \((\alpha,\gamma,\lambda)\).}\label{fig:kur}
\scriptsize{Source: Authors.}%
\end{figure}

\subsubsection*{Order Statistics}

The evaluation of various life cycle systems with given component structures requires the consideration of ordered random variables, known as order statistics \citep{david2003order}. These variables play a fundamental role in understanding the reliability and failure behavior of systems. In this section, we will present the fundamental distributional properties of the order statistics of the \(\mbox{UMW}\) distribution.

The PDF of the \(r\)-th order statistic \(Y_{(r)}\), defined as the \(r\)-th smallest value in an ordered sample of size \(n\) from the UMW distribution, is given by the following theorem:

\begin{theorem}
   Let \( Y_1, Y_2, \ldots, Y_n \) be an independent and identically distributed (i.i.d.) random sample from the UMW distribution with sample size \( n \), and let \( Y_{(1)}, Y_{(2)}, \ldots, Y_{(n)} \) denote the order statistics of this sample, such that \( Y_{(1)} \leq Y_{(2)} \leq \dots \leq Y_{(n)} \). For any \( r = 1, 2, \dots, n \), the PDF of \( Y_{(r)} \) is given by
    \begin{align*}
        f_{Y_r}(y) = & \frac{n!}{(r-1)!(n-r)!} \dfrac{\alpha \left[-\log(y)\right]^\gamma}{\log(y) y^{\lambda+1}} \left[\lambda \log(y) - \gamma\right] \left\{ 1 - \exp\left( -\alpha \left[-\log(y)\right]^\gamma y^{-\lambda} \right) \right\}^{n-r} \\
        & \times \exp\left( -r \alpha \left[-\log(y)\right]^\gamma y^{-\lambda} \right),
    \end{align*}
    where \( r \in \{1, 2, \dots, n\} \) {and} \( y \in (0, 1) \).
\end{theorem}
\begin{proof}
    The PDF of \( Y_{(r)} \) can be obtained by applying Theorem 3 from \cite[p.~167]{rohatgi2015}. By substituting the CDF and PDF of the UMW distribution, as given in Equations \eqref{eq3} and \eqref{eq4}, into the general formula for the PDF of order statistics.
\end{proof}

The PDFs of the minimum, \(Y_{(1)}\), and the maximum, \(Y_{(n)}\), which represent the smallest and largest values in an ordered sample of size \(n\) from the UMW distribution, are specific cases of Theorem 1. Specifically, they are obtained when \(r = 1\) and \(r = n\), corresponding to the minimum and maximum, respectively, as presented in Corollary 1 and Corollary 2.

\begin{corollary}
   The PDF of the minimum, \( Y_{(1)} \), is given by:
    \begin{align*}
        f_{Y_1}(y) = & \dfrac{n \alpha \left[-\log(y)\right]^\gamma}{\log(y) y^{\lambda+1}} \left[\lambda \log(y) - \gamma\right] \left\{ 1 - \exp\left( -\alpha \left[-\log(y)\right]^\gamma y^{-\lambda} \right) \right\}^{n-1} \\
        & \times \exp\left( -\alpha \left[-\log(y)\right]^\gamma y^{-\lambda} \right).
    \end{align*}
\end{corollary}

\begin{corollary}
    The PDF of the maximum, \( Y_{(n)} \), is given by:
    \begin{align*}
        f_{Y_n}(y) = & \dfrac{n \alpha \left[-\log(y)\right]^\gamma}{\log(y) y^{\lambda+1}} \left[\lambda \log(y) - \gamma\right] \exp\left( -n \alpha \left[-\log(y)\right]^\gamma y^{-\lambda} \right).
    \end{align*}
\end{corollary}

Order statistics are important mathematical properties widely used in reliability analysis and service life modeling. 
In the context of unit distributions, for example, in a healthcare system composed of multiple municipalities, the vaccination coverage rate can be observed for each unit. The minimum observed value may indicate critical regions with low immunization, requiring priority attention from public health policies. Thus, order statistics, such as the minimum value, can provide crucial information for identifying areas of greater vulnerability, supporting the design of intervention strategies and the allocation of resources in public policies, as in the case of vaccination coverage.

\subsection*{Regression Model}
\label{sec:regressao}

One of the objectives of this work is to introduce the class of Unit-Modified Weibull quantile regression (RQ-UMW) models, in which the response variable has a UMW distribution. The reparameterization to be considered will be in terms of the quantile of the distribution, which, in general, presents advantages when modeling asymmetric random variables and with possible atypical observations compared to modeling in terms of the mean \citep{john2015,lemonte2016}.

Using Equation \eqref{eq:quantilumw} and performing some algebraic manipulations, we obtain:
\begin{equation}\label{eq5}
\alpha=-\frac{\log(\tau)\mu_{\tau}^\lambda}{\left[-\log(\mu_{\tau})\right]^\gamma}\cdot
\end{equation}
Based on Equation \eqref{eq5}, we derive a reparametrization of the $\mbox{UMW}(\alpha,\gamma,\lambda)$ distribution in terms of $\alpha$, 
yielding the reparameterized CDF and PDF of $\mbox{UMW}(\mu_{\tau},\gamma,\lambda)$, respectively:
\begin{equation}\label{eq2rep}
{F_Y}(y) = \exp\left(\frac{\mu_{\tau}^{\lambda}\log\left({\tau}\right)\left[-\log\left(y\right)\right]^{\gamma}}{y^{\lambda}\left[-\log\left({\mu_{\tau}}\right)\right]^{\gamma}}\right)
\end{equation}
and
\begin{equation}\label{eq1rep}
{f_Y}(y)= -\dfrac{\mu_{\tau}^{\lambda}\log\left({\tau}\right)\left[-\log\left(y\right)\right]^{\gamma}\left[{\lambda}\log\left(y\right)-{\gamma}\right]}{y^{{\lambda}+1}\left[-\log\left({\mu_{\tau}}\right)\right]^{\gamma}\log\left(y\right)}\exp\left(\frac{\mu_{\tau}^{\lambda}\log\left({\tau}\right)\left[-\log\left(y\right)\right]^{\gamma}}{y^{\lambda}\left[-\log\left({\mu_{\tau}}\right)\right]^{\gamma}}\right),
\end{equation}
with \(y\in (0, 1)\).

Let $\bm{Y}=(Y_1,\ldots , Y_n)^{\top}$ be a random sample, in which each $Y_t$, for $t=1,\ldots,n$, follows a reparameterized $\mbox{UMW}(\mu_{\tau},\gamma,\lambda)$ distribution. 
Considering de PDF in Equation \eqref{eq1rep}, we can include a regression structure for quantile modeling through the following structure:
\begin{equation*}\label{eq6}
 g(\mu_{\tau,t})=\bm{x}_t^{\top}\bm{\beta}=\zeta_t, \quad t=1, \ldots, n,
\end{equation*}
where \( g : (0,1) \rightarrow \mathbb{R} \) is monotonic and twice differentiable link function, 
such as logit, probit, cloglog, loglog, cauchit, among others, 
$\bm{\beta}=(\beta_0, \ldots, \beta_{k-1})^{\top}$ is the vector of unknown parameters ($\bm{\beta}\in\mathbb{R}^{k}$) and $\bm{x}_{t}=({x_{t0}},\ldots,{x_{tk-1}})^\top$ are observations of $k$ covariates ($k<n$), which are assumed to be fixed and known. Considering the intercept, we have that $x_{t0} = 1, \, \forall t$.
In practice, the parameters are unknown and need to be estimated. 
The next section explores likelihood inference. 

\section*{Likelihood Inference}
\label{sec:estim}

Let $\bm{y}=(y_1,\ldots , y_n)^{\top}$ be an observed sample, the parameter vectors of the UMW distribution and the RQ-UMW model are given by $\bm{\theta}_1=(\alpha,\gamma,\lambda)$ and $\bm{\theta}_2=(\gamma,\lambda, \bm{\beta}^\top)^\top$, respectively. 
The maximum likelihood estimators (MLE) \citep{pawitan2001}
of the parameter vectors $\bm{\theta}_m$, for $m=1,2$, are given by
\begin{equation*}
\widehat{\bm{\theta}}_m=\argmax_{{\bm{\theta}_m}\in \bm\Theta_m}(\ell_m(\bm{\theta}_m)), 
\end{equation*}
where ${\bm\Theta_1} \subseteq \mathbb{R}_+^3$ and ${\bm\Theta_2} \subseteq \{\mathbb{R}_+^2 \times \mathbb{R}^k\}$ are the parameter spaces of the UMW distribution and the RQ-UMW model, respectively, and $\ell_m(\bm{\theta}_m)$ are the log-likelihood functions, given by
\begin{equation}\label{LV}
\ell_m(\bm{\theta}_m)=\ell_m(\bm{\theta}_m;\bm{y})=\sum_{t=1}^{n}\ell_{_m,t}(\bm{\theta}_m,{{y_t}}),
\end{equation}
with
\begin{align*}
\ell_{1,t}(\bm{\theta_1},y_t) = & \ \log(\alpha)+\log\left(\gamma-\lambda\log(y_t)\right)-(\lambda+1)\log(y_t)+ (\gamma-1)  \log\left(-\log(y_t)\right) \\ & -\alpha y_t^{-\lambda}\left[-\log(y_t)\right]^{\gamma},                     
\\ 
\ell_{2,t}(\bm{\theta_2},y_t) = & \ \log\left({\gamma}-{\lambda}\log\left(y_\text{t}\right)\right)-\left({\lambda}+1\right)\log\left(y_\text{t}\right)+y_\text{t}^{-\lambda}\mu_{\tau,t}^{\lambda}\left[-\log\left(y_\text{t}\right)\right]^{\gamma}{\left[-\log\left({\mu_{\tau,t}}\right)\right]^{-\gamma}}\\ &\times \log\left({\tau}\right) +\left({\gamma}-1\right)\log\left(-\log\left(y_\text{t}\right)\right) + \log\left(-{\mu_{\tau,t}^{\lambda}\log\left({\tau}\right)\left[-\log\left({\mu_{\tau,t}}\right)\right]^{-\gamma}}\right),
\end{align*}
where the MLE are obtained by maximizing the log-likelihood functions $\ell_m(\bm{\theta}_m)$. However, the solution to this maximization does not have a closed form, requiring the use of numerical methods to obtain the estimates.

\subsection*{Score Vector}

The score vector is obtained by deriving the log-likelihood function $\ell_m(\bm{\theta}_m)$, given by Equation \eqref{LV}, with respect to each of the components of the parameter vector, $\bm{\theta}_m$. 
The elements referring to the components of the score vector related to the parameters $\alpha$, $\gamma$, and $\lambda$ of the UMW distribution are given, respectively, by:
\begin{align*}
\mbox{U}_\alpha(\bm{\theta}_1) =\dfrac{\partial \ell_{1}(\bm{\theta}_1)}{\partial \alpha}= \sum_{t=1}^{n}\frac{\partial\ell_{1,t}(\bm{\theta}_1,{{y_t}})}{\partial \alpha}, \quad 
\mbox{U}_\gamma(\bm{\theta}_1) =\dfrac{\partial \ell_{1}(\bm{\theta}_{1})}{\partial \gamma} = \sum_{t=1}^{n}\frac{\partial\ell_{1,t}(\bm{\theta}_1,{{y_t}})}{\partial \gamma},\\
\mbox{U}_\lambda(\bm{\theta}_1) =\dfrac{\partial \ell_{1}(\bm{\theta}_1)}{\partial \lambda}= \sum_{t=1}^{n}\frac{\partial\ell_{1,t}(\bm{\theta}_1,{{y_t}})}{\partial \lambda}, \quad\quad\quad\quad\quad\quad\quad\quad\,\,\,
\end{align*}
where
\begin{align*}
\frac{\partial\ell_{1,t}(\bm{\theta_1},{{y_t}})}{\partial \alpha} &=  
\dfrac{1}{{\alpha}}-\dfrac{\left[-\log\left(y_\text{t}\right)\right]^{\gamma}}{y_\text{t}^{\lambda}} := \texttt{r}_t, \\
\frac{\partial\ell_{1,t}(\bm{\theta_1},{{y_t}})}{\partial \gamma}  &= \dfrac{1}{{\gamma}-{\lambda}\log\left(y_\text{t}\right)}-\dfrac{{\alpha}\left[-\log\left(y_\text{t}\right)\right]^{\gamma}\log\left(-\log\left(y_\text{t}\right)\right)}{y_\text{t}^{\lambda}}+\log\left(-\log\left(y_\text{t}\right)\right) := \texttt{s}_t,
\\
\frac{\partial\ell_{1,t}(\bm{\theta_1},{{y_t}})}{\partial \lambda} &= -\dfrac{\log\left(y_\text{t}\right)}{{\gamma}-{\lambda}\log\left(y_\text{t}\right)}+\dfrac{{\alpha}\left[-\log\left(y_\text{t}\right)\right]^{\gamma}\log\left(y_\text{t}\right)}{y_\text{t}^{\lambda}}-\log\left(y_\text{t}\right) := \texttt{u}_t.
\end{align*}

The score vector of the UMW distribution can be written in matrix form as
\[\bm{\mbox{U}}(\bm{\theta}_1)=\Big[\mbox{U}_\alpha(\bm{\theta}_1), \mbox{U}_\gamma(\bm{\theta}_1), \mbox{U}_\lambda(\bm{\theta}_1)\Big],\]
\noindent where
\[
\mbox{U}_{\alpha}(\bm{\theta_1})  = \bm{\texttt{r}}\bm{1}_n^{\top},\quad
\mbox{U}_{\gamma}(\bm{\theta_1})  = \bm{\texttt{s}}\bm{1}_n^{\top},\quad
\mbox{U}_{\lambda}(\bm{\theta_1})  = \bm{\texttt{u}}\bm{1}_n^{\top},\]
with \(\bm{\texttt{r}}=(\texttt{r}_1,\ldots,\texttt{r}_n)\), \(\bm{\texttt{s}}=(\texttt{s}_1,\ldots,\texttt{s}_n)\), \(\bm{\texttt{u}}=(\texttt{u}_1,\ldots,\texttt{u}_n)\), \(\bm{1}_n^{\top}\)is a column vector of ones of dimension \(n\).

The elements referring to the components of the score vector related to the parameters \( \gamma \), \( \lambda \), and \( \beta_j \), where \( j = 1, \ldots, k \), of the RQ-UMW model are given, respectively, by:
\begin{align*}
\mbox{V}_\gamma(\bm{\theta}_2) =\dfrac{\partial \ell_{2}(\bm{\theta}_2)}{\partial \gamma}= \sum_{t=1}^{n}\frac{\partial\ell_{2,t}(\bm{\theta}_2,{{y_t}})}{\partial \gamma}, \quad 
\mbox{V}_\lambda(\bm{\theta}_2) =\dfrac{\partial \ell_{2}(\bm{\theta}_2)}{\partial \lambda} = \sum_{t=1}^{n}\frac{\partial\ell_{2,t}(\bm{\theta}_2,{{y_t}})}{\partial \lambda} , \\
\mbox{V}_{\beta_j}(\bm{\theta}_2)= \frac{\partial \ell_{2}(\bm{\theta}_2)}{\partial \beta_j}=\sum_{t=1}^{n}\frac{\partial \ell_{2,t}(\bm{\theta}_2,{{y_t}})}{\partial \beta_j}
=\sum_{t=1}^{n}\frac{\partial \ell_{2,t}(\bm{\theta}_2,{{y_t}})}{\partial \mu_{\tau,t}}\frac{d \mu_{\tau,t}}{d {\zeta_t}}\frac{\partial \zeta_t}{\partial \beta_j}, \quad\quad\quad
\end{align*}
where 
\begin{align*}
&\frac{\partial\ell_{2,t}(\bm{\theta}_2,{{y_t}})}{\partial \gamma} =   \dfrac{{\texttt{A}_{t}}{\texttt{B}_{t}}}{y_\text{t}^{\lambda}\left[-\log\left(\mu_{\tau,t}\right)\right]^{\gamma}} + \dfrac{1}{{\gamma}-{\lambda}\log\left(y_\text{t}\right)} + {\texttt{B}_{t}} := \texttt{v}_t, \\
& \frac{\partial\ell_{2,t}(\bm{\theta}_2,{{y_t}})}{\partial \lambda} =  \dfrac{{\texttt{A}_{t}}\left[\log\left(\mu_{\tau,t}\right)-\log\left(y_\text{t}\right)\right]}{y_\text{t}^{\lambda}\left[-\log\left(\mu_{\tau,t}\right)\right]^{\gamma}}- \dfrac{\log\left(y_\text{t}\right)}{{\gamma}-{\lambda}\log\left(y_\text{t}\right)}-\log\left(y_\text{t}\right) + \log\left(\mu_{\tau,t}\right) := \texttt{z}_t, \\
&\frac{\partial \ell_{2,t}(\bm{\theta}_2,{{y_t}})}{\partial \mu_{\tau,t}} =  \dfrac{\left[y_\text{t}^{\lambda}\left[-\log\left(\mu_{\tau,t}\right)\right]^{\gamma}+{\texttt{A}_{t}}\right]\left[{\lambda}\log\left(\mu_{\tau,t}\right)-{\gamma}\right]}{\mu_{\tau,t}\log\left(\mu_{\tau,t}\right)y_\text{t}^{\lambda}\left[-\log\left(\mu_{\tau,t}\right)\right]^{\gamma}} := \texttt{w}_t,
\end{align*}
with 
${\texttt{A}_{t}} = \mu_{\tau,t}^{\lambda}\log\left({\tau}\right)\left[-\log\left(y_\text{t}\right)\right]^{\gamma}$ 
and 
${\texttt{B}_{t}} = \log\left(-\log\left(y_\text{t}\right)\right)-\log\left(-\log\left(\mu_{\tau,t}\right)\right)$.
Note that \( \frac{\partial \zeta_t}{\partial \beta_j} = x_{tj} \), \( \frac{d \mu_{\tau,t}}{d {\zeta_t}} = \frac{1}{g^{\prime}(\mu_{\tau,t})} \), where \( g^{\prime}(\cdot) \) denotes the first derivative of the function \( g(\cdot) \). The elements corresponding to the coordinates of the score vector relative to \( \beta_j \) can be rewritten as:
\begin{equation*}
\mbox{V}_{\beta_j}(\bm{\theta}_2)= 
\sum_{t=1}^{n} \texttt{w}_t  \frac{1}{g^{\prime}(\mu_{\tau,t})} x_{tj}.
\end{equation*}

In matrix form, the score vector of the RQ-UMW model can be written as
\[\bm{\mbox{V}}(\bm{\theta}_2) = \Big[ \mbox{V}_\gamma(\bm{\theta}_2),  \mbox{V}_\lambda(\bm{\theta_2}), \bm{\mbox{V}}^{\top}_{\bm{\beta}}(\bm{\theta}_2)\Big]^{\top},\]
\noindent where
$$\mbox{V}_{\gamma}(\bm{\theta_2})  = \bm{\texttt{v}}\bm{1}_n^{\top},\quad
\mbox{V}_{\lambda}(\bm{\theta_2})  = \bm{\texttt{z}}\bm{1}_n^{\top},\quad
\bm{\mbox{V}}_{{\bm\beta}}(\bm{\theta_2})  = \boldsymbol{\mbox{X}}^{\top} \bm{\mbox{T}} \bm{\texttt{w}^\top},$$
with \(\bm{\texttt{v}} = (\texttt{v}_1, \ldots, \texttt{v}_n)\), \(\bm{\texttt{z}} = (\texttt{z}_1, \ldots, \texttt{z}_n)\), \(\bm{\mbox{X}}\) is an \(n \times k\) matrix whose \(t\)-th row is given by \(\bm{x}_t^{\top}\), \(\bm{\mbox{T}} = \mbox{diag}\{g'(\mu_{\tau,1})^{-1}, \ldots, g'(\mu_{\tau,n})^{-1}\}\), \(\bm{\texttt{w}} = (\texttt{w}_1, \ldots, \texttt{w}_n)\), and \(\bm{1}_n^{\top}\) is a column vector of ones of dimension \(n\).

The MLE can be obtained by solving the system of non-linear equations given by:
\[
\bm{\mbox{U}}(\bm\theta_1) \Big|_{\bm\theta_1 = \widehat{\bm\theta}_1} = \bm{0} \quad \text{and} \quad \bm{\mbox{V}}(\bm\theta_2) \Big|_{\bm\theta_2 = \widehat{\bm\theta}_2} = \bm{0},
\]
where \(\bm{0}\) are zero vectors with the appropriate dimensions. Since these systems of equations do not have an analytical solution, the use of non-linear optimization algorithms is necessary. In this case, we use the limited-memory 
Broyden-Fletcher-Goldfarb-Shanno with box constraints (L-BFGS-B) method \citep{byrd1995}, via the \textit{optim} package in the \texttt{R} environment \citep{R}. In the case of the UMW distribution, we can compute a semi-closed MLE for \(\alpha\).
The MLE of \(\alpha\) is obtained from \(\bm{\mbox{U}}_\alpha(\bm\theta_1) = 0\), and is given by
$$\hat{\alpha}(\hat{\gamma},\hat{\lambda}) = \dfrac{n}{\sum\limits_{t=1}^{n}\left[-\log\left(y_\text{t}\right)\right]^{\hat{\gamma}}{y_\text{t}^{-\hat{\lambda}}}}\cdot$$

\subsection*{Large Sample Inference}

Under some mild regularity conditions, according to \cite{pawitan2001}, the asymptotic properties of estimators can be derived using likelihood-based arguments. A more precise and classical description of these regularity conditions is given by \cite{lehmann1983}, who lists the following assumptions under which asymptotic results in point estimation can be established: the parameter space is an open set; the family of distributions share a common support independent of the parameter; the probability density function is twice continuously differentiable with respect to the parameter; differentiation under the integral sign is allowed; the Fisher information is positive and finite; there exists an integrable function that uniformly bounds the second derivative of the log-likelihood in a neighborhood of the true parameter; the expected value of the score function is zero; and the Fisher information equals both the variance of the score function and the negative expected value of the second derivative of the log-likelihood.

Under these suppositions, the MLE of \(\bm\theta_m\), denoted by \(\hat{\bm\theta}_m\), is consistent and approximately follows a multivariate normal distribution of dimensions \(q_1 = 3\) and \(q_2 = 2 + k\), respectively, given by
\begin{equation}\label{eq:normalidade}
\widehat{\bm{\theta}}_1 \stackrel{a}{\sim} N_{q_1}\left(\bm{\theta}_1,\bm{\mbox{J}}^{-1}({\bm{\theta}}_1)\right) \quad \text{{and}} \quad \widehat{\bm{\theta}}_2 \stackrel{a}{\sim} N_{q_2}\left(\bm{\theta}_2,\bm{\mbox{L}}^{-1}({\bm{\theta}}_2)\right),
\end{equation} 
where $\bm{\mbox{J}}^{-1}({\bm{\theta}}_1)$ {and} $\bm{\mbox{L}}^{-1}({\bm{\theta}}_2)$ are the inverses of the observed information matrix of the UMW distribution and the RQ-UMW model, respectively. Appendix \ref{ap:ApenOIM} provides the derivations and calculations of these matrices. Furthermore, \(\stackrel{a}{\sim}\) denotes approximately distributed as.

We can construct confidence intervals based on the approximate distribution of the MLE, with approximate confidence level \((1-\nu)\times100\%\), for the parameters of UMW distribution and RQ-UMW model, given by
\[\left[\widehat{{\theta}}_{m,i}-z_{\frac{\nu}{2}}\widehat{\text{se}}(\widehat{{\theta}}_{m,i});\widehat{{\theta}}_{m,i}+z_{\frac{\nu}{2}}\widehat{\text{se}}(\widehat{{\theta}}_{m,i})\right],\mbox{ {for} }i=1, \ldots, {q_m},\]
where \(\widehat{\text{se}}(\widehat{{\theta}}_{m,i}) = \sqrt{d_{ii}}\) with \(d_{ii}\) being the \(i\)-th element of the diagonal of \(\bm{\mbox{J}}^{-1}(\widehat{\bm{\theta}}_1)\) or \(\bm{\mbox{L}}^{-1}(\widehat{\bm{\theta}}_2)\), \(z_{\frac{\nu}{2}}\) is the standard normal quantile, such that \(\mathbb{P}(Z > z_{\frac{\nu}{2}}) = \frac{\nu}{2}\); \(\widehat{{\theta}}_{m,i}\) is the \(i\)-th coordinate of the estimated parameter vector \(\widehat{\bm{\theta}}_m\) and \(Z \sim {N}(0, 1)\).

Regarding hypothesis tests, consider the hypotheses \(\mbox{H}_{0}: {\theta}_{m,i} = {\theta}^{0}_{m,i}\) or \(\mbox{H}_{1}: {\theta}_{m,i} \neq {\theta}^{0}_{i}\), where \({\theta}_{m,i}^0\) is the specific value of an unknown parameter \({\theta}_{m,i}\). 
To test these hypotheses, we use the Wald test \citep{wald1943}, with the test statistic given by
\begin{eqnarray}\label{eq:wald}
W=\frac{\widehat{{\theta}}_{m,i}-{\theta}^{0}_{m,i}}{\widehat{\text{se}}(\widehat{{\theta}}_{m,i})} \overset{d}{\rightarrow} N(0,1),
\end{eqnarray}
where \(\overset{d}{\rightarrow}\) denotes the convergence in distribution, under the null hypothesis \(H_0\). 
For example, one can test the hypothesis \(\lambda = 0\), which results in the unit-Weibull distribution \citep{weibullunit} or the unit-Weibull regression model \citep{regweibullunit}.

\section*{Diagnostic Measures} \label{diagnostico} 

For the regression model, after estimating the parameters, it is essential to conduct a diagnostic analysis to identify observations that may disproportionately influence the parameter estimates and affect the accuracy of the fitted model. 
This section presents the diagnostic measures used to assess the goodness-of-fit of the RQ-UMW model.

The quantile residuals \citep{dunnsmyth1996} will be considered, given by
\begin{eqnarray*}
r_{t} = \Phi^{-1} \left( F_Y(y_t; \widehat{\bm{\theta}}_2,\tau) \right),\quad t=1,\ldots,n,
\end{eqnarray*}
where $\Phi^{-1}$ is the quantile function of the standard normal distribution, $F(y_t; \widehat{\bm{\theta}}_2,\tau)$ is the FDA given by \eqref{eq2rep}. 
If the model is correctly specified, the residuals should be uncorrelated and approximately normally distributed, with zero mean and unit variance.

In order to assess the quality of the fitted regression, we suggest using the generalized coefficient of determination $(R^{2}_{G})$ \citep{nagelkerke1991}, given by
\begin{eqnarray*}
R^{2}_{G} = 
1- \exp \left(-\frac{2}{n}\left[ \ell_2(\widehat{\bm{\theta}}_2) - \ell_1(\widehat{\bm{\theta}}_{1})\right]\right),
\end{eqnarray*}
where $\ell_1(\widehat{\bm{\theta}}_1)$ is the log-likelihood function evaluated at the maximum likelihood estimates of the model parameters without the regression structure (null model), and $\ell_2(\widehat{\bm{\theta}}_2)$ is the log-likelihood function evaluated at the maximum likelihood estimates of the parameters of the fitted regression model. 
The generalized coefficient of determination shows the proportion of the variability of $Y$ that can be explained by the fitted model. 
In this sense, note that ${0 \leqslant R^{2}_{G} \leqslant 1}$, that is, the closer $R^{2}_{G}$ is to one, the better the explanatory power of the model in relation to the variable of interest.

In practice, it is common to fit multiple candidate models to the available data and then use a selection criterion to choose the best model. For this purpose, we suggest considering the following traditional information criteria:
the Akaike Information Criterion (AIC) \citep{akaike1974}, 
the Bayesian Information Criterion (BIC) \citep{akaike1978bayesian,schwarz1978}, 
and the corrected AIC ($\mbox{AIC}_c$) \citep{hurvich1989}, 
where $\mbox{AIC} = 2{q_m} - 2 \ell_m (\widehat{\bm{\theta}}_m)$, $\mbox{BIC} = {q_m}\log(n) - 2 \ell_m (\widehat{\bm{\theta}}_m)$, and $\mbox{AICc} = \mbox{AIC} + \frac{2{q^2_m} + 2{q_m}}{n - {q_m} - 1}$. Among a set of adjusted candidate models, the best model will be the one that minimizes the chosen selection criterion.

\section*{Monte Carlo Simulations} \label{sec:SML}

In this section, the results of the Monte Carlo simulations for the UMW distribution and the RQ-UMW model are presented and discussed. 
To maximize the log-likelihood functions of the introduced models, the L-BFGS-B method is used, through the \textit{optim} package in the \texttt{R} environment \citep{R}. 
Tables with the results of the following metrics to evaluate the likelihood inference are presented: bias, mean squared error (MSE), and 95\% coverage rate (CR\%) of the parameter estimates. 
In addition, boxplot graphs are displayed for a more comprehensive analysis of the results. To evaluate the flexibility of the proposed models, the simulation study employed parameter values chosen arbitrarily across diverse scenarios, assessing their performance over a broad range of shapes and dispersion levels.

Estimation using the L-BFGS-B algorithm proved to be computationally efficient, exhibiting fast convergence even for moderate to large sample sizes. 
The procedure showed numerical stability, with convergence failures occurring only rarely; such cases were excluded from the analysis when present. 
The numerical procedures were robust to the chosen initial values, with distribution parameters initialized at 1 and regression coefficients initially estimated using ordinary least squares on the transformed response variable $g(y)$.
The implementation details are available in the source code at  \url{https://github.com/JoaoInacioS/UMW.git}.

\subsection*{Results For UMW Distribution}

The results of the Monte Carlo simulations for the UMW distribution are presented in Table \ref{tab:simUMW}, in which $R = 10,000$ Monte Carlo replicates were performed, varying in sample sizes $n \in \{40, 80, 120, 160, 200\}$. Four scenarios were analyzed: Scenario 1, with parameters $\alpha = 0.7$, $\gamma = 1.3$, and $\lambda = 0.5$; Scenario 2, with parameters $\alpha = 0.3$, $\gamma = 0.8$, and $\lambda = 1.2$; Scenario 3, with $\alpha = 1.3$, $\gamma = 1.1$, and $\lambda = 0.6$; and Scenario 4, with parameters $\alpha = 0.5$, $\gamma = 0.9$, and $\lambda = 0.8$.

\begin{figure}
\centering
{\includegraphics[width=0.315\textwidth]{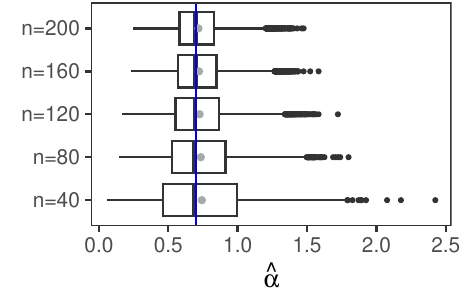}}
{\includegraphics[width=0.315\textwidth]{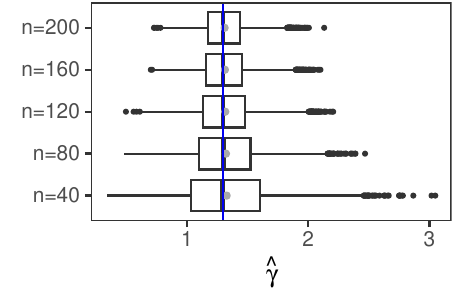}}
{\includegraphics[width=0.315\textwidth]{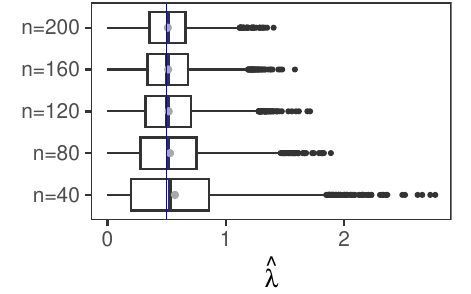}}\\[-0.05cm]
\footnotesize{(a) Scenario 1 ($\alpha = 0.7$, $\gamma = 1.3$ and $\lambda = 0.5$)}\\[0.4cm]
{\includegraphics[width=0.315\textwidth]{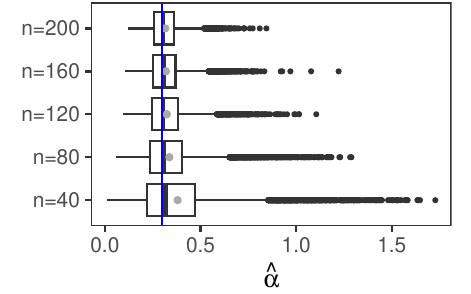}}
{\includegraphics[width=0.315\textwidth]{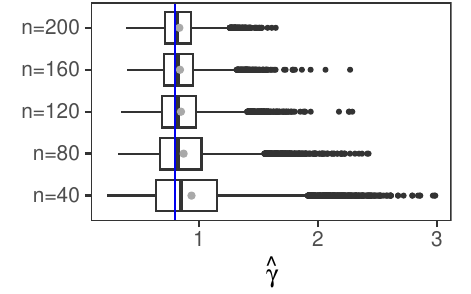}}
{\includegraphics[width=0.315\textwidth]{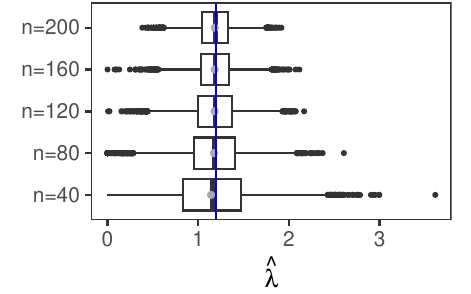}}\\[-0.05cm]
\footnotesize{(b) Scenario 2 ($\alpha = 0.3$, $\gamma = 0.8$ and $\lambda = 1.2$)}\\[0.4cm]
{\includegraphics[width=0.315\textwidth]{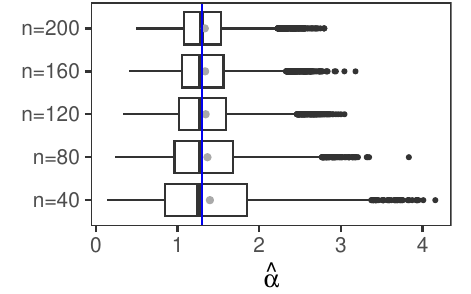}}
{\includegraphics[width=0.315\textwidth]{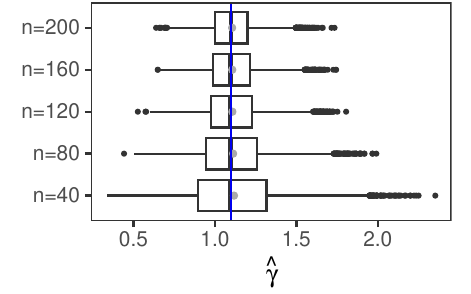}}
{\includegraphics[width=0.315\textwidth]{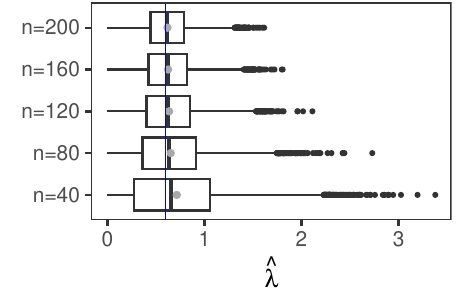}}\\[-0.05cm]
\footnotesize{(c) Scenario 3 ($\alpha = 1.3$, $\gamma = 1.1$ and $\lambda = 0.6$)}\\[0.4cm]
{\includegraphics[width=0.315\textwidth]{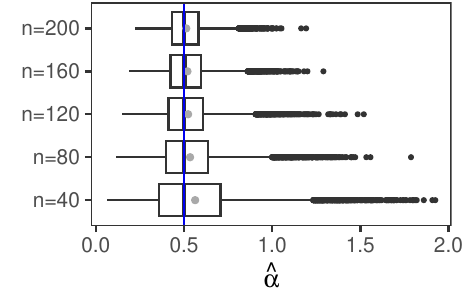}}
{\includegraphics[width=0.315\textwidth]{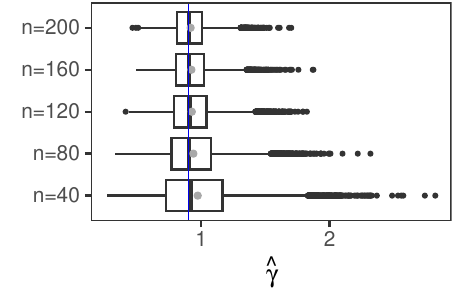}}
{\includegraphics[width=0.315\textwidth]{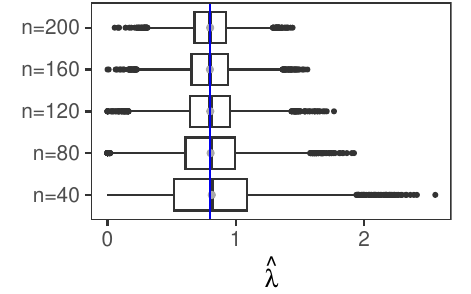}}\\[-0.05cm]
\footnotesize{(d) Scenario 4 ($\alpha = 0.5$, $\gamma = 0.9$ and $\lambda = 0.8$)}
\caption{Boxplots of the MLE from Monte Carlo simulations of the {UMW} distribution, with $R=10,000$ and $n \in \left\{40, 80, 120, 160, 200\right\}$.}\label{fig:CenUMW}
\scriptsize{Source: Authors.}%
\end{figure}

When analyzing the results in Table \ref{tab:simUMW} for all scenarios, we observe that as the sample size increases, both bias and MSE decrease, while CR\% approaches the nominal value of 95\%. 
These results evidence that MLE are asymptotically consistent, unbiased, and normally distributed. 
These characteristics are illustrated by the boxplots in Figure \ref{fig:CenUMW}, where the blue line represents the fixed value of the parameter, and the gray dot indicates the mean of the estimate. 
For all parameters, as the sample size increases, there is a decrease in the interquartile range, with the mean converging to the true value of the parameter. 
This highlights the consistency and unbiasedness of the MLE, with the mean and median approaching each other, indicating the symmetry of the estimator distributions.

The proposed estimators exhibit good performance, even in finite samples, with results aligning well with their expected asymptotic properties. 
These findings validate both the theoretical formulation and the developed implementation.

\begin{table}
\centering
\caption{\label{tab:simUMW}{Results of Monte Carlo simulations of the {UMW} distribution, with ${R=10,000}$ and ${n \in \left\{40, 80, 120, 160, 200\right\}}$.}}
\begin{tabular}{@{\hspace{0.1cm}}l@{\hspace{0.1cm}}c@{\hspace{0.17cm}}c@{\hspace{0.17cm}}c@{\hspace{0.17cm}}l@{\hspace{0.2cm}}c@{\hspace{0.25cm}}l@{\hspace{0.0cm}}r@{\hspace{0.25cm}}r@{\hspace{0.25cm}}rl@{\hspace{0.1cm}}r@{\hspace{0.25cm}}r@{\hspace{0.25cm}}rl@{\hspace{0.1cm}}r@{\hspace{0.25cm}}r@{\hspace{0.25cm}}r@{\hspace{0.1cm}}}
  \hline
&\multicolumn{3}{c}{Scenario} & &&&  \multicolumn{3}{c}{Bias} & & \multicolumn{3}{c}{MSE} & & \multicolumn{3}{c}{CR\%}\\
\cline{2-4} \cline{8-10} \cline{12-14} \cline{16-18}
& $\alpha$ & $\gamma$ & $\lambda$ && $n$& & \multicolumn{1}{c}{$\hat\alpha$} & \multicolumn{1}{c}{$\hat\gamma$} & \multicolumn{1}{c}{$\hat\lambda$} & & \multicolumn{1}{c}{$\hat\alpha$} & \multicolumn{1}{c}{$\hat\gamma$} & \multicolumn{1}{c}{$\hat\lambda$} && \multicolumn{1}{c}{$\hat\alpha$} & \multicolumn{1}{c}{$\hat\gamma$} & \multicolumn{1}{c}{$\hat\lambda$} \\
\hline
\multirow[l]{5}{*}{(1)}& \multirow[c]{5}{*}{0.7}&\multirow[c]{5}{*}{1.3}&\multirow[c]{5}{*}{0.5}&&40 && 0.042 & 0.028 & 0.071 && 0.123 & 0.153 & 0.206 && 0.903 & 0.964 & 0.971 \\ 
&&&&&80 && 0.034 & 0.022 & 0.029 && 0.078 & 0.089 & 0.113 && 0.923 & 0.967 & 0.971 \\
&&&&&120 && 0.026 & 0.016 & 0.018 && 0.056 & 0.063 & 0.079 && 0.931 & 0.964 & 0.966 \\
&&&&&160 && 0.021 & 0.013 & 0.014 && 0.044 & 0.049 & 0.062 && 0.935 & 0.952 & 0.955 \\ 
&&&&&200 && 0.017 & 0.011 & 0.011 && 0.035 & 0.038 & 0.049 && 0.941 & 0.954 & 0.955 \\
\hline
\multirow[l]{5}{*}{(2)}&\multirow[c]{5}{*}{0.3}&\multirow[c]{5}{*}{0.8}&\multirow[c]{5}{*}{1.2}&&40 && 0.080 & 0.137 & $-$0.057 && 0.062 & 0.191 & 0.243 && 0.927 & 0.967 & 0.962 \\
&&&&&80 && 0.037 & 0.070 & $-$0.025 && 0.024 & 0.084 & 0.120 && 0.935 & 0.951 & 0.955 \\
&&&&&120 && 0.024 & 0.048 & $-$0.018 && 0.013 & 0.052 & 0.077 && 0.945 & 0.953 & 0.957 \\ 
&&&&&160 && 0.019 & 0.040 & $-$0.017 && 0.009 & 0.038 & 0.056 && 0.950 & 0.958 & 0.957 \\ 
&&&&&200 && 0.015 & 0.035 & $-$0.016 && 0.007 & 0.029 & 0.044 && 0.953 & 0.954 & 0.956 \\
\hline
\multirow[l]{5}{*}{(3)}& \multirow[c]{5}{*}{1.3}&\multirow[c]{5}{*}{1.1}&\multirow[c]{5}{*}{0.6}&&40 && 0.096 & 0.017 & 0.115 && 0.492 & 0.088 & 0.309 && 0.899 & 0.965 & 0.973 \\ 
& &&&&80 && 0.066 & 0.012 & 0.052 && 0.292 & 0.051 & 0.161 && 0.919 & 0.962 & 0.969 \\
& &&&&120 && 0.044 & 0.007 & 0.036 && 0.198 & 0.034 & 0.107 && 0.931 & 0.962 & 0.965 \\
&&&&&160 && 0.038 & 0.007 & 0.026 && 0.155 & 0.027 & 0.083 && 0.932 & 0.953 & 0.953 \\
&&&&&200 && 0.034 & 0.006 & 0.020 && 0.126 & 0.022 & 0.066 && 0.939 & 0.949 & 0.950 \\ 
\hline
\multirow[l]{5}{*}{(4)}&\multirow[c]{5}{*}{0.5}&\multirow[c]{5}{*}{0.9}&\multirow[c]{5}{*}{0.8}&&40 && 0.063 & 0.072 & 0.014 && 0.085 & 0.122 & 0.174 && 0.919 & 0.966 & 0.967 \\
&&&&&80 && 0.035 & 0.038 & 0.004 && 0.041 & 0.060 & 0.088 && 0.931 & 0.951 & 0.955 \\
&&&&&120 && 0.024 & 0.026 & 0.001 && 0.026 & 0.039 & 0.058 && 0.942 & 0.951 & 0.951 \\
&&&&&160 && 0.020 & 0.023 & $-$0.001 && 0.019 & 0.029 & 0.044 && 0.944 & 0.955 & 0.953 \\ 
&&&&&200 && 0.013 & 0.017 & 0.001 && 0.014 & 0.022 & 0.034 && 0.945 & 0.953 & 0.953 \\    
\hline
\end{tabular}
\\\scriptsize{Source: Authors.}%
\end{table}

\subsection*{Results For The RQ-UMW Model}

The results of the Monte Carlo simulations for the RQ-UMW model are presented in Table \ref{tab:simRUMW}, where $R=10,000$ Monte Carlo replicates were performed, varying the quantiles $\tau \in \{0.1, 0.5, 0.9\}$ and the sample sizes $n \in \{50, 150, 300, 500\}$. 
Two scenarios were analyzed: Scenario 1, with parameters $\gamma = 2.7$, $\lambda = 1.8$, $\beta_0 = 0.2$, $\beta_1 = -0.4$, and $\beta_2 = 0.5$; and Scenario 2, with parameters $\gamma = 1.5$, $\lambda = 2.3$, $\beta_0 = 0.5$, $\beta_1 = -0.6$, and $\beta_2 = 0.2$. 
The covariates were randomly generated from the uniform distribution $U(0,1)$ and considered fixed during all replications.

When analyzing the results of the Monte Carlo simulations, as presented in Table \ref{tab:simRUMW} for both scenarios and quantiles analyzed, we observe that both the bias and the MSE decrease as the sample size increases, while the CR\% approaches 95\%, as expected. 
These figures provide numerical evidence that the MLEs are asymptotically consistent, unbiased, and approximately normally distributed. These characteristics can be observed in Figure \ref{fig:CenRUMW}, where the blue line represents the fixed value of the parameter and the gray dot indicates the mean of the estimate. 
As expected, in both scenarios, as the sample size increases, the interquartile range decreases, and the mean approaches the true value of the parameter. 
We can also observe that the mean and median are converging, which is a characteristic of the normality of the estimators' distributions.

\begin{figure}
\centering
{\includegraphics[width=0.325\textwidth]{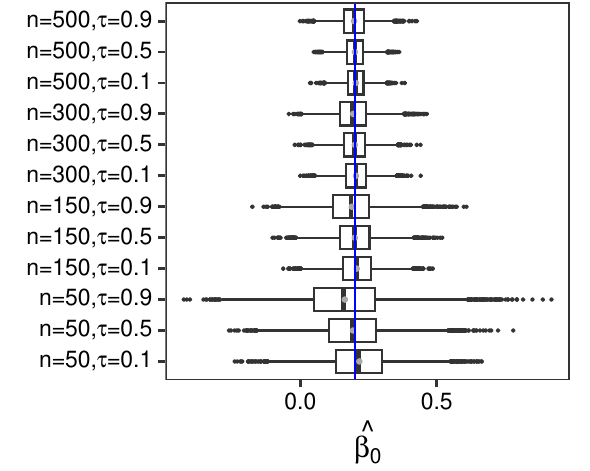}}
{\includegraphics[width=0.325\textwidth]{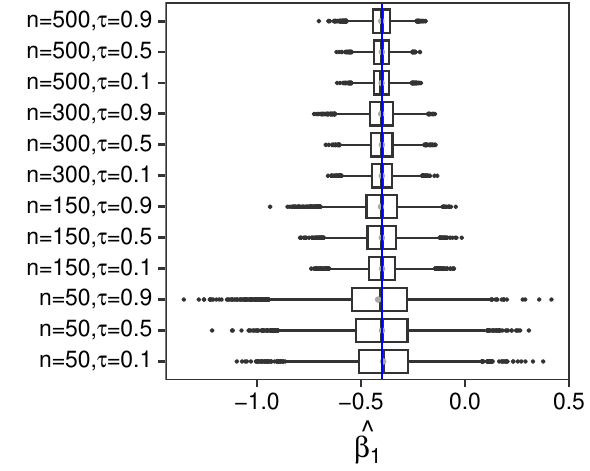}}
{\includegraphics[width=0.325\textwidth]{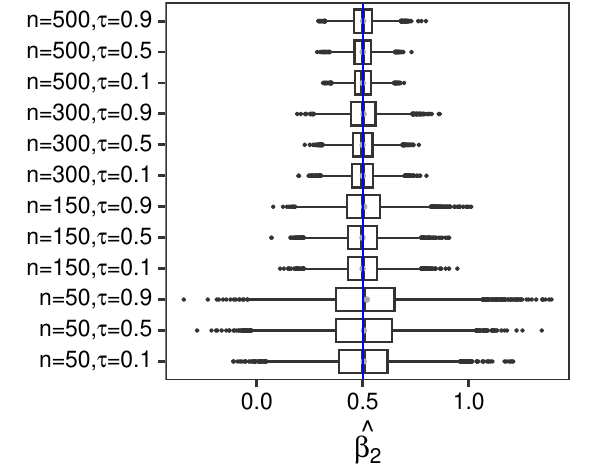}}\\
{\includegraphics[width=0.325\textwidth]{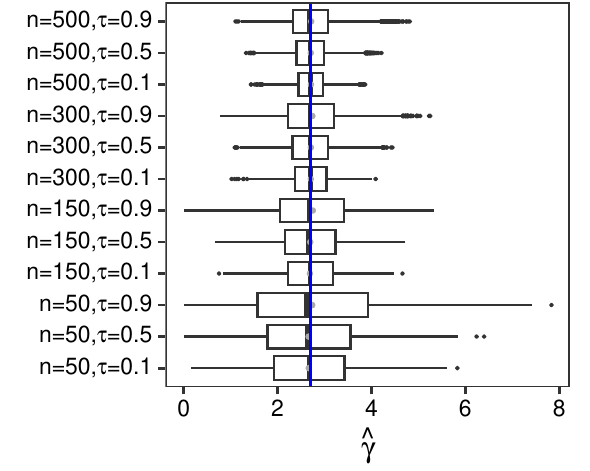}}
{\includegraphics[width=0.325\textwidth]{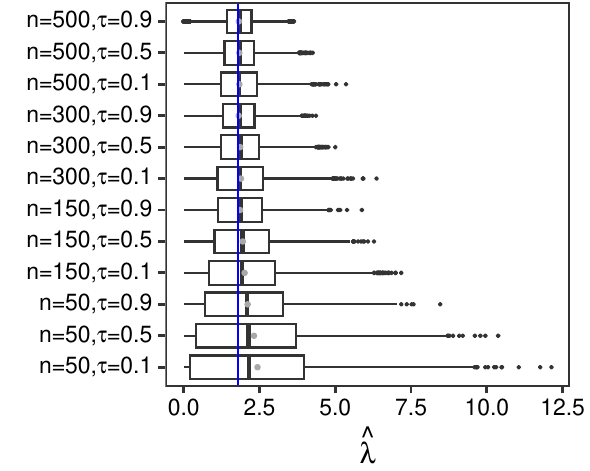}}\\[-0.05cm]
\footnotesize{(a) Scenario 1 ($\beta_0= 0.2, \beta_1 = -0.4, \beta_2 = 0.5, \gamma = 2.7$ and $\lambda = 1.8)$}\\[0.4cm]
{\includegraphics[width=0.325\textwidth]{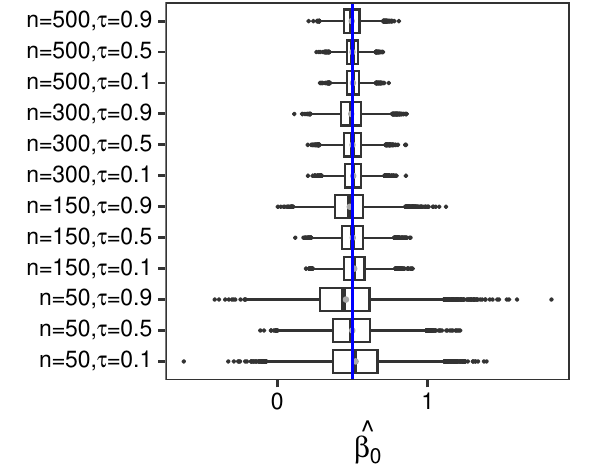}}
{\includegraphics[width=0.325\textwidth]{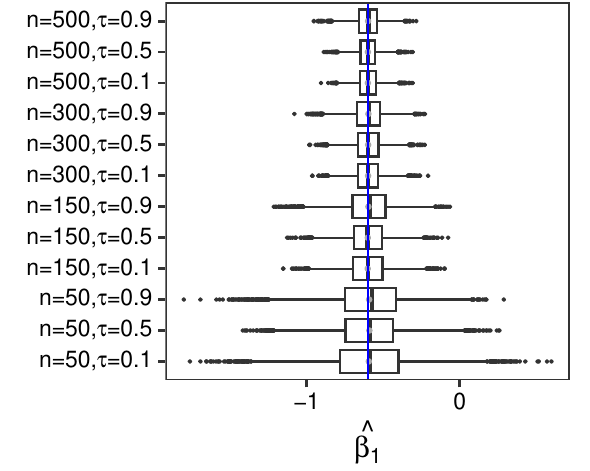}}
{\includegraphics[width=0.325\textwidth]{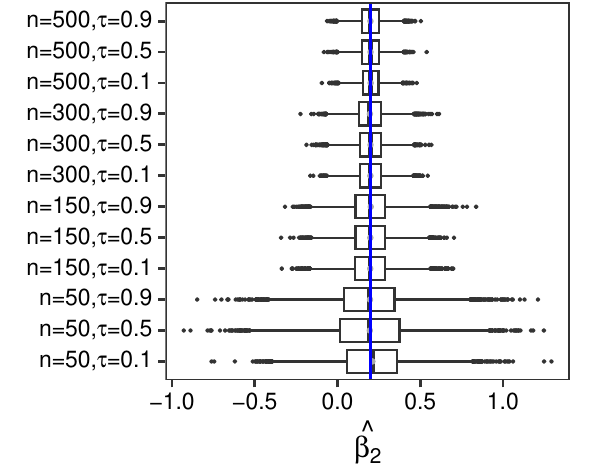}}\\
{\includegraphics[width=0.325\textwidth]{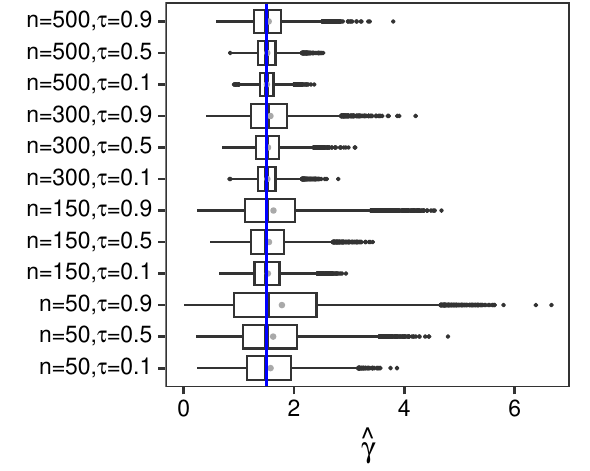}}
{\includegraphics[width=0.325\textwidth]{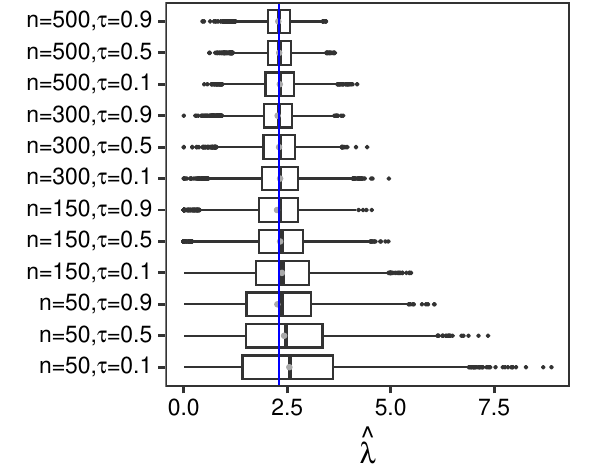}}\\[-0.05cm]
\footnotesize{(b) Scenario 2 ($\beta_0= 0.5, \beta_1 = -0.6$, $\beta_2 = 0.2, \gamma = 1.5$ and $\lambda = 2.3$)}
\caption{Boxplots of the MLE from Monte Carlo simulations of the RQ-UMW model with reparameterization by $\alpha$, with $R = 10,000$, $n \in \left\{50, 150, 300, 500\right\}$ and $\tau \in \left\{0.1, 0.5, 0.9\right\}$.
}\label{fig:CenRUMW}
\scriptsize{Source: Authors.}%
\end{figure}

\begin{table}
\centering
\caption{\label{tab:simRUMW}{Results of Monte Carlo simulations of the RQ-UMW model with reparameterization by ${\alpha}$, with ${R = 10{,}000}$, ${n \in \left\{50, 150, 300, 500\right\}}$ and ${\tau \in \left\{0.1, 0.5, 0.9\right\}}$.}}
\resizebox{0.88\textwidth}{!}{
\begin{tabular}{c@{\hspace{0.2cm}}l@{\hspace{0.1cm}}cl@{\hspace{0.0cm}}r@{\hspace{0.35cm}}r@{\hspace{0.35cm}}rl@{\hspace{0.1cm}}r@{\hspace{0.35cm}}r@{\hspace{0.35cm}}rl@{\hspace{0.1cm}}r@{\hspace{0.35cm}}r@{\hspace{0.35cm}}r}
\hline
 & &&&  \multicolumn{3}{c}{$\tau = 0.1$} & & \multicolumn{3}{c}{$\tau = 0.5$} & & \multicolumn{3}{c}{$\tau = 0.9$}\\
 \cline{5-7} \cline{9-11} \cline{13-15}
Par && $n$ && Bias & MSE & CR\% && Bias & MSE & CR\% && Bias & MSE & CR\% \\
\hline
\multicolumn{15}{c}{Scenario 1: $\gamma = 2.7, \lambda = 1.8, \beta_0= 0.2, \beta_1 = -0.4, \beta_2 = 0.5$} \\
\hline
\multirow[c]{4}{*}{$\hat\gamma$} && 50 && $-$0.035 & 0.927 & 0.964 && $-$0.039 & 1.220 & 0.959 && 0.039 & 2.061 & 0.963 \\
&& 150 && $-$0.011 & 0.393 & 0.967 && $-$0.004 & 0.539 & 0.965 && 0.049 & 0.893 & 0.965 \\
&& 300 && 0.008 & 0.225 & 0.969 && 0.009 & 0.301 & 0.960 && 0.040 & 0.504 & 0.953 \\
&& 500 && 0.010 & 0.141 & 0.956 && 0.006 & 0.185 & 0.952 && 0.015 & 0.308 & 0.949 \\
\hline
\multirow[c]{4}{*}{$\hat\lambda$} && 50 && 0.634 & 5.120 & 0.965 && 0.517 & 4.064 & 0.960 && 0.308 & 2.550 & 0.960 \\
&& 150 && 0.207 & 2.095 & 0.968 && 0.148 & 1.519 & 0.964 && 0.057 & 1.018 & 0.966 \\
&& 300 && 0.086 & 1.172 & 0.968 && 0.062 & 0.832 & 0.964 && 0.012 & 0.567 & 0.954 \\
&& 500 && 0.034 & 0.727 & 0.956 && 0.033 & 0.515 & 0.953 && 0.018 & 0.352 & 0.951 \\
\hline
\multirow[c]{4}{*}{$\hat\beta_0$} && 50 && 0.016 & 0.016 & 0.927 && $-$0.007 & 0.018 & 0.929 && $-$0.037 & 0.031 & 0.925 \\
&& 150 && 0.007 & 0.006 & 0.944 && $-$0.000 & 0.006 & 0.944 && $-$0.012 & 0.010 & 0.945 \\
&& 300 && 0.004 & 0.003 & 0.947 && $-$0.000 & 0.003 & 0.946 && $-$0.007 & 0.005 & 0.946 \\
&& 500 && 0.003 & 0.002 & 0.948 && $-$0.000 & 0.002 & 0.949 && $-$0.003 & 0.003 & 0.946 \\
\hline
\multirow[c]{4}{*}{$\hat\beta_1$} && 50 && 0.006 & 0.031 & 0.937 && 0.001 & 0.042 & 0.931 && $-$0.017 & 0.042 & 0.947 \\
&& 150 && 0.001 & 0.009 & 0.942 && 0.000 & 0.007 & 0.954 && $-$0.004 & 0.012 & 0.954 \\
&& 300 && 0.000 & 0.005 & 0.947 && 0.001 & 0.004 & 0.953 && $-$0.003 & 0.007 & 0.953 \\
&& 500 && 0.001 & 0.003 & 0.948 && 0.000 & 0.003 & 0.952 && $-$0.002 & 0.004 & 0.952 \\
\hline
\multirow[c]{4}{*}{$\hat\beta_2$} && 50 && 0.004 & 0.029 & 0.930 && 0.005 & 0.039 & 0.932 && 0.019 & 0.044 & 0.947 \\
&& 150 && $-$0.001 & 0.010 & 0.951 && $-$0.002 & 0.011 & 0.943 && 0.006 & 0.014 & 0.955 \\
&& 300 && $-$0.000 & 0.005 & 0.949 && $-$0.001 & 0.005 & 0.945 && 0.003 & 0.007 & 0.953 \\
&& 500 && $-$0.000 & 0.003 & 0.948 && $-$0.001 & 0.003 & 0.948 && 0.002 & 0.004 & 0.950 \\
\hline
\multicolumn{15}{c}{Scenario 2: $\gamma = 1.5, \lambda = 2.3, \beta_0= 0.5, \beta_1 = -0.6, \beta_2 = 0.2$} \\
\hline
\multirow[c]{4}{*}{$\hat\gamma$} && 50 && 0.074 & 0.326 & 0.971 && 0.121 & 0.538 & 0.960 && 0.280 & 1.301 & 0.967 \\
&& 150 && 0.025 & 0.116 & 0.952 && 0.046 & 0.194 & 0.948 && 0.127 & 0.511 & 0.942 \\
&& 300 && 0.017 & 0.058 & 0.954 && 0.028 & 0.098 & 0.952 && 0.072 & 0.240 & 0.946 \\
&& 500 && 0.007 & 0.035 & 0.951 && 0.015 & 0.056 & 0.952 && 0.040 & 0.144 & 0.945 \\
\hline
\multirow[c]{4}{*}{$\hat\lambda$} && 50 && 0.251 & 2.444 & 0.965 && 0.126 & 1.745 & 0.960 && $-$0.038 & 1.276 & 0.972 \\
&& 150 && 0.072 & 0.884 & 0.950 && 0.030 & 0.625 & 0.949 && $-$0.046 & 0.528 & 0.946 \\
&& 300 && 0.028 & 0.425 & 0.952 && 0.003 & 0.304 & 0.952 && $-$0.031 & 0.255 & 0.946 \\
&& 500 && 0.021 & 0.256 & 0.952 && 0.006 & 0.179 & 0.951 && $-$0.015 & 0.154 & 0.948 \\
\hline
\multirow[c]{4}{*}{$\hat\beta_0$} && 50 && 0.020 & 0.050 & 0.929 && $-$0.002 & 0.034 & 0.933 && $-$0.044 & 0.064 & 0.920 \\
&& 150 && 0.011 & 0.010 & 0.941 && 0.002 & 0.011 & 0.949 && $-$0.018 & 0.020 & 0.931 \\
&& 300 && 0.005 & 0.006 & 0.943 && 0.000 & 0.007 & 0.945 && $-$0.008 & 0.010 & 0.942 \\
&& 500 && 0.003 & 0.003 & 0.950 && 0.000 & 0.003 & 0.946 && $-$0.006 & 0.006 & 0.943 \\
\hline
\multirow[c]{4}{*}{$\hat\beta_1$} && 50 && 0.006 & 0.085 & 0.934 && 0.010 & 0.054 & 0.932 && 0.008 & 0.061 & 0.941 \\
&& 150 && $-$0.001 & 0.020 & 0.940 && 0.000 & 0.018 & 0.946 && 0.004 & 0.026 & 0.937 \\
&& 300 && $-$0.002 & 0.009 & 0.945 && 0.002 & 0.010 & 0.945 && 0.004 & 0.012 & 0.946 \\
&& 500 && $-$0.001 & 0.006 & 0.948 && 0.001 & 0.007 & 0.949 && 0.002 & 0.007 & 0.949 \\
\hline
\multirow[c]{4}{*}{$\hat\beta_2$} && 50 && 0.011 & 0.053 & 0.936 && $-$0.005 & 0.073 & 0.935 && $-$0.003 & 0.055 & 0.950 \\
&& 150 && $-$0.002 & 0.018 & 0.941 && $-$0.002 & 0.018 & 0.947 && 0.001 & 0.019 & 0.952 \\
&& 300 && $-$0.001 & 0.009 & 0.946 && 0.000 & 0.010 & 0.948 && $-$0.004 & 0.010 & 0.948 \\
&& 500 && 0.000 & 0.005 & 0.947 && 0.000 & 0.006 & 0.949 && $-$0.001 & 0.006 & 0.949 \\
\hline
\end{tabular}
 }
\\\scriptsize{Source: Authors.}
\end{table}

Regarding the different quantiles, we observed that in both scenarios, the estimates for the $\gamma$ parameter exhibited a smaller interquartile range at the $\tau = 0.1$ compared to the other quantiles. 
For the $\lambda$ parameter, the interquartile range was smaller at the $\tau = 0.9$ compared to the other quantiles. 
As for the $\beta$ coefficients, their estimates were similar across the three quantile values. These results indicate that the variability of the estimates for the $\gamma$ and $\lambda$ parameters differs according to the quantile, while the estimates of the $\beta$ coefficients show constant variability across the different quantiles.

\section*{Empirical Applications} \label{sec:apl}

In this section, we use the proposed methods for modeling different datasets. 
We started the empirical analyses with some descriptive measures of each variable. 
In this case, we consider the following sample measures: Minimum (Min.), Median (Median), Mean (Mean), Maximum (Max.), Standard Deviation (SD), Skewness (AC), and Kurtosis (K). 
To test the normality of the data, the Anderson-Darling normality test (AD($p$)) \citep{anderson1952asymptotic} is performed, where the p-value is used to assess the evidence against the null hypothesis.

Two applications were carried out considering the UMW distribution: the first with an indicator of Brazil's 17th SDG and the second with data on the useful volume of several reservoirs in Brazil, which is indirectly related to SDGs 6, 12, and 13.
To assess the performance of the UMW distribution, we compare it with the beta, Kumaraswamy (KW), Modified Kumaraswamy (MK) \citep{sagrillo2021modified}, and unit-Weibull (UW) distributions.
To determine which distribution best fits the data, we computed the maximized log-likelihood (Loglik), AIC, BIC, AIC$_c$, Kolmogorov-Smirnov (KS) statistic \citep{Kolmogorov}, Anderson-Darling (AD) statistic \citep{Stephens1974}, 
and Cramér-von Mises (CvM) criterion \citep{cramer1928}.

An analysis of reading skills in dyslexic children was conducted using the RQ-UMW model, which is indirectly associated with SDG 4 (Quality Education) and SDG 3 (Good Health and Well-Being).
SDG 4 aims to ensure inclusive and equitable education by providing support to students with learning disabilities, such as dyslexia.
SDG 3 focuses on improving well-being for all ages, including supporting children with dyslexia through mental health strategies, as learning disabilities can impact psychological well-being.
In addition to the model selection criteria, the Mean Squared Error (MSE), Root Mean Squared Error (RMSE), Mean Absolute Error (MAE), and Mean Absolute Percentage Error (MAPE) between observed and predicted values were calculated.
The RQ-UMW model was compared with the generalized beta quantile regression model (RQ-beta) \citep{bourguignon2024parametric}, the Kumaraswamy quantile regression model (RQ-KW), and the unit Weibull quantile regression model (RQ-UW) \citep{regweibullunit}. 
The codes used in applications of the UMW distribution and the RQ-UMW regression model are publicly available at \url{https://github.com/JoaoInacioS/UMW.git}.

\subsection*{SDG 17.3}

For the initial application of the UMW distribution, the indicator ``total municipal revenues collected" ($y$) of the municipalities of Rio Grande do Sul (RS) state in the year 2021 was considered. 
This indicator represents the proportion of revenues that a municipality actually managed to collect compared to the amount predicted or expected for the year 2021. 
In other words, $y$ is a measure of efficiency in collecting planned revenues. 
This is indicator 17.3 of SDG 17, which focuses on strengthening the mobilization of domestic resources, including through international support to developing countries, to improve national capacity to collect taxes and other revenues. 
This indicator is calculated by dividing the amount of municipal revenues collected (taxes, fees, and contributions) by the total amount of revenues of the municipality. 
The database is available at \href{https://www.cidadessustentaveis.org.br/paginas/idsc-br}{https://www.cidadessustentaveis.org.br/paginas/idsc-br}.

Table \ref{tab:descr_y} presents the descriptive statistics of total municipal revenues collected from municipalities in RS in 2021. 
The results indicate that the data exhibit right-skewness and heavy tails, suggesting that most municipalities have low revenue collection efficiency. 
This is characterized by an asymmetry coefficient greater than zero and a kurtosis greater than three, respectively.
Additionally, the AD($p$) test confirms that the data do not follow a normal distribution, as the $p$-value is below the 5\% significance level. 
This further suggests the need to consider more flexible distributions.

\begin{table}
\caption{{Descriptive measures of the total municipal revenue collected from municipalities in RS in 2021.}\label{tab:descr_y}}
\centering
\begin{tabular}{lrrrrrrrrr}
\hline
 & $n$ & Min. & Median & Mean & Max. & SD & AC & K & AD($p$) \\
\hline
$y$ & 497 & 0.018 & 0.068 & 0.090 & 0.563 & 0.068 &  2.499 & 12.558 & $<$0.001\\
\hline
\end{tabular}
\\\scriptsize{Source: Authors.}%
\end{table}

To compare the fit of the UMW distribution with its competitors, Table \ref{tab:apliq_umw_coef} presents the parameter estimates for the UMW, beta, KW, MK, and UW distributions, along with the corresponding goodness-of-fit measures.  
All parameters were statistically significant at least at the 5\% significance level. 
Among the seven goodness-of-fit measures considered, the UMW distribution achieved the best performance in six, outperforming the other models, except for the BIC, where the UW distribution performed better.  
Figure \ref{fig:umw_models} displays the estimated density functions for each distribution fitted to the data, as well as the QQ plots. 
It is evident that the UMW, MK, and UW distributions exhibit visually similar fits, which is consistent with the relatively close goodness-of-fit measures reported in Table \ref{tab:apliq_umw_coef}. 
On the other hand, the beta and KW distributions failed to accurately capture the data, particularly struggling to represent the observed peak.  
Although the KW and UW distributions showed similar fits, their overall performance was inferior to that of the UMW distribution, except in terms of the BIC. However, since the UW distribution is a particular case of the UMW distribution, we use the Wald statistic proposed in Equation \eqref{eq:wald} to test the hypothesis $\lambda=0$. With a $p$-value of 0.032, we conclude that the null hypothesis is rejected at the 5\% significance level, indicating that the UMW distribution is more suitable for describing the behavior of these data. 

\begin{table}
    \centering
    \caption{\label{tab:apliq_umw_coef}{Coefficients and goodness-of-fit measures of the fitted models for the total municipal revenues collected by municipalities in RS in 2021.}}
    \resizebox{1\textwidth}{!}{
    \begin{tabular}{l@{\hspace{0.2cm}}r@{\hspace{0.2cm}}r@{\hspace{0.3cm}}r@{\hspace{0.05cm}}r@{\hspace{0.2cm}}r@{\hspace{0.2cm}}r@{\hspace{0.05cm}}r@{\hspace{0.2cm}}r@{\hspace{0.2cm}}r@{\hspace{0.05cm}}r@{\hspace{0.2cm}}r@{\hspace{0.2cm}}r@{\hspace{0.05cm}}r@{\hspace{0.2cm}}r@{\hspace{0.2cm}}r}
    \hline
    &\multicolumn{3}{c}{UMW}&&\multicolumn{2}{c}{Beta}&&\multicolumn{2}{c}{KW}&&\multicolumn{2}{c}{MK}&&\multicolumn{2}{c}{UW}\\
    \cline{2-4}\cline{6-7}\cline{9-10}\cline{12-13}\cline{15-16}
    & \multicolumn{1}{c}{$\alpha$} & \multicolumn{1}{c}{$\gamma$} & \multicolumn{1}{c}{$\lambda$} && \multicolumn{1}{c}{$\alpha$} & \multicolumn{1}{c}{$\beta$} && \multicolumn{1}{c}{$\alpha$} & \multicolumn{1}{c}{$\beta$} && \multicolumn{1}{c}{$\alpha$} & \multicolumn{1}{c}{$\beta$}&& \multicolumn{1}{c}{$\alpha$} & \multicolumn{1}{c}{$\beta$}\\
    \hline
    Estimate & 0.006 & 3.213 & 0.622 &  & 24.317 & 0.091 &  & 1.473 & 0.080 &  & 0.113 & 2.722 &  & 0.006 & 4.878\\
    SE & 0.001 & 0.787 & 0.290 &  & 1.570 & 0.003 &  & 0.049 & 0.003 &  & 0.005 & 0.194 &  & 0.001 & 0.175\\
    $p$-value & $<$0.001 & $<$0.001 & 0.032 &  & $<$0.001 & $<$0.001 &  & $<$0.001 & $<$0.001 &  & $<$0.001 & $<$0.001 &  & $<$0.001 & $<$0.001\\
\hline
Loglik & \multicolumn{3}{c}{\textbf{837.317}}&& \multicolumn{2}{c}{781.248}&& \multicolumn{2}{c}{759.022}&& \multicolumn{2}{c}{835.131}&& \multicolumn{2}{c}{835.138}\\
AIC & \multicolumn{3}{c}{\textbf{$-$1668.634}} && \multicolumn{2}{c}{$-$1558.496} && \multicolumn{2}{c}{$-$1514.044}&& \multicolumn{2}{c}{$-$1666.263}&& \multicolumn{2}{c}{$-$1666.277}\\
BIC & \multicolumn{3}{c}{$-$1656.008} && \multicolumn{2}{c}{$-$1550.079 } && \multicolumn{2}{c}{$-$1505.627}&& \multicolumn{2}{c}{{$-$1657.846}}&& \multicolumn{2}{c}{\textbf{$-$1657.859}}\\
AIC$_c$ & \multicolumn{3}{c}{\textbf{$-$1668.585}} && \multicolumn{2}{c}{$-$1558.472} && \multicolumn{2}{c}{$-$1514.020}&& \multicolumn{2}{c}{$-$1666.239}&& \multicolumn{2}{c}{$-$1666.252}\\
KS & \multicolumn{3}{c}{\textbf{0.037}} && \multicolumn{2}{c}{0.098} && \multicolumn{2}{c}{0.097}&& \multicolumn{2}{c}{0.040}&& \multicolumn{2}{c}{0.040} \\
AD & \multicolumn{3}{c}{\textbf{0.831}} && \multicolumn{2}{c}{8.135} && \multicolumn{2}{c}{10.893} && \multicolumn{2}{c}{1.124}&& \multicolumn{2}{c}{1.134} \\
CvM & \multicolumn{3}{c}{\textbf{0.141}} && \multicolumn{2}{c}{1.385} && \multicolumn{2}{c}{1.638} && \multicolumn{2}{c}{0.194}&& \multicolumn{2}{c}{0.215}\\
\hline
\end{tabular}
}
\\\scriptsize{Source: Authors.}%
\end{table}

\begin{figure}
\centering
{\includegraphics[width=0.87\textwidth]{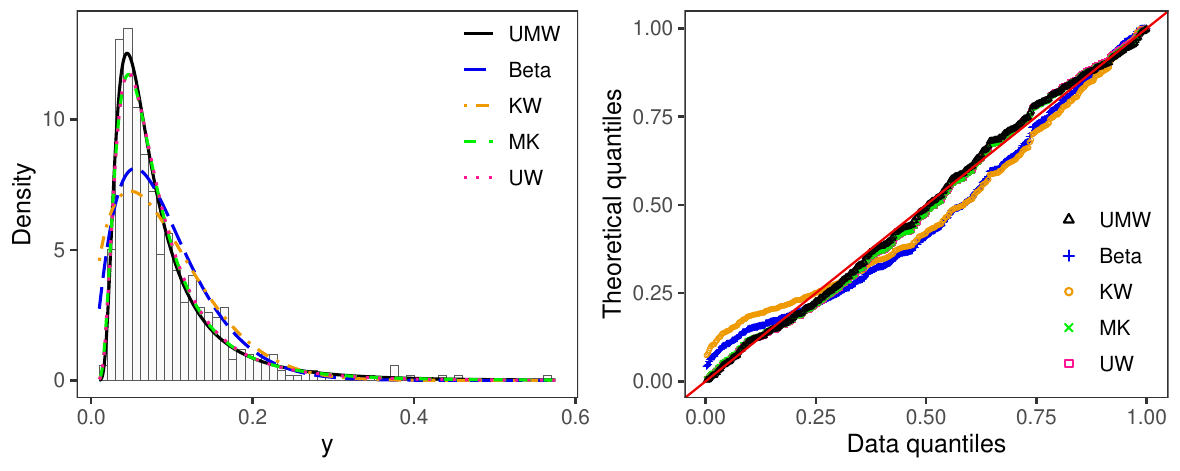}}
\caption{\label{fig:umw_models}Histograms, density plots, and QQ plots for the total municipal revenues collected by municipalities in RS in 2021.}
\scriptsize{Source: Authors.}%
\end{figure}

\subsection*{Useful Volume}

For the subsequent application of the UMW distribution, we consider the study conducted by \cite{sagrillo2021modified}, which utilized data on the relative useful volumes of several reservoirs in Brazil. 
These data can be obtained from the website of the National Electric System Operator (ONS), available at \href{https://www.ons.org.br/}{https://www.ons.org.br/}. 
Most of the reservoir data span from January 2011 to December 2019, with the exception of the Mauá reservoir, which starts in October 2012. 
Months with missing values, as well as values equal to zero or greater than or equal to one, were excluded from the sample.

The useful volume of reservoirs is a critical indicator for water resources management and has significant implications for water security and sustainability, aligning with several SDG targets. 
For example: Target 6.4 (SDG 6), which aims to increase water use efficiency to ensure the sustainability of water resources. Therefore, better management of the useful volume of reservoirs is essential for maintaining water availability and quality; 
Target 12.2 (SDG 12), which focuses on achieving sustainable management and efficient use of natural resources. In this context, the useful volume of reservoirs plays a vital role in efficient water resources management; 
Target 13.1 (SDG 13), which aims to improve education and capacity to mitigate and adapt to climate change. Efficient reservoir management is crucial in adapting to climate change, especially in regions vulnerable to droughts or floods.

A total of 26 reservoirs were analyzed, with only those located in the Southern region of Brazil included in this work for brevity.  
The remaining reservoirs are presented in Appendix \ref{ap:ApenA}.
Among the 26 reservoirs, the proposed UMW distribution provided the best fit for 15 (57.69\%), followed by the MK distribution for 9 (34.62\%) and the beta distribution for 2 (7.69\%). The KW and UW distributions did not provide the best fit for any of the reservoirs.  

Table \ref{tab:suldescrt} presents the descriptive measures of the relative useful volume of the reservoirs in the Southern region. 
In general, these volumes exhibit light tails, indicated by kurtosis values smaller than three and maximum values close to one, with similar standard deviations (SDs). 
The AD($p$) test revealed that none of the reservoirs exhibited normally distributed behavior, as all presented $p$-values smaller than the 5\% significance level.

\begin{table}
\caption{Descriptive measures of the relative useful volume of reservoirs in the South region. Omit(NA) represents the months that were discarded due to values being equal to zero and greater than or equal to one, and the number of months with missing values is shown in parentheses.\label{tab:suldescrt}}
\centering
\resizebox{1\textwidth}{!}{
\begin{tabular}{l@{\hspace{0.35cm}}r@{\hspace{0.3cm}}r@{\hspace{0.3cm}}r@{\hspace{0.3cm}}r@{\hspace{0.3cm}}r@{\hspace{0.3cm}}r@{\hspace{0.3cm}}r@{\hspace{0.3cm}}r@{\hspace{0.3cm}}r@{\hspace{0.3cm}}r}
\hline
Reservoir & $n$ & Omit(NA) & Min. & Median & Mean & Max. & SD & AC & K & AD($p$) \\
\hline
Barra Grande   & 104 & 4(0) & 0.028 & 0.545 & 0.586 & 0.999 & 0.265 & 0.127 & 1.788 & $<$0.001\\
Campos Novos   & 99  & 9(0) & 0.128 & 0.526 & 0.572 & 0.996 & 0.275 & 0.212 & 1.751 & $<$0.001\\
G. B. Munhoz   & 107 & 0(1) & 0.138 & 0.606 & 0.617 & 0.997 & 0.273 & $-$0.026 & 1.637 & $<$0.001\\
G. P. Souza    & 107 & 0(1) & 0.165 & 0.564 & 0.577 & 0.998 & 0.243 & 0.179 & 1.765 & $<$0.001\\
Machadinho     & 106 & 2(0) & 0.072 & 0.580 & 0.585 & 0.997 & 0.271 & 0.039 & 1.806 & $<$0.001\\
Mauá           & 80  & 7(21) & 0.108 & 0.665 & 0.650 & 0.999 & 0.282 & $-$0.223 & 1.735 & $<$0.001\\
Passo Fundo    & 106 & 2(0) & 0.287 & 0.785 & 0.739 & 0.998 & 0.217 & $-$0.360 & 1.832 & $<$0.001\\
Salto Santiago & 105 & 2(1) & 0.163 & 0.724 & 0.695 & 0.999 & 0.259 & $-$0.409 & 1.837 & $<$0.001\\
Santa Clara-PR & 89  & 18(1) & 0.133 & 0.580 & 0.598 & 0.999 & 0.260 & 0.030 & 1.684 & $<$0.001\\
Segredo        & 95  & 12(1) & 0.171 & 0.708 & 0.694 & 0.998 & 0.243 & $-$0.435 & 2.136 & $<$0.001\\
\hline
\end{tabular}
}
\\\scriptsize{Source: Authors.}%
\end{table}

The goodness-of-fit measures for the five distributions UMW, beta, KW, MK, and UW are presented in Table \ref{tab:sulajust}. 
In this table, below the name of each reservoir, the distribution that best fits is indicated, along with the number of criteria in which it was superior (in parentheses). 
At the end of the table, the number of reservoirs in which each distribution had the best fit is provided.
The results show that the UMW distribution was the best in all reservoirs in at least four criteria, particularly in the Barra Grande, Campos Novos, G. P. Souza, and Passo Fundo reservoirs, where the UMW distribution was superior in six of the criteria evaluated.
In general, the UMW distribution presented better results in the Loglik, KS, AD, and CvM criteria compared to the other distributions. 
However, in most reservoirs, it lost in the BIC criterion to the MK distribution. 
These distribution adjustments can be seen in Figure \ref{fig:g_south}, which shows the densities and QQ plots for each reservoir in the Southern region.

\begin{table} 
\caption{{Adjustment measures for the relative useful volume of reservoirs in the South region.}\label{tab:sulajust}}
\centering
\resizebox{0.89\textwidth}{!}{
\begin{tabular}{clrrrrrrr}
\hline
Reservoir & Dist & Loglik & AIC & BIC & AICc & KS & AD & CvM\\
\hline
\multirow[c]{5}{*}{Barra Grande} & UMW & \textbf{14.674} & \textbf{$-$23.348} & \textbf{$-$15.415} & \textbf{$-$23.108} & \textbf{0.093} & \textbf{1.402} & 0.213\\
 & Beta & 9.693 & $-$15.385 & $-$10.096 & $-$15.266 & 0.125 & 3.049 & 0.537\\
 & KW & 9.792 & $-$15.584 & $-$10.295 & $-$15.465 & 0.126 & 3.077 & 0.547\\
 & MK & 10.273 & $-$16.546 & $-$11.257 & $-$16.427 & 0.123 & 1.506 & \textbf{0.200}\\
\scriptsize{Best: UMW (6)} & UW & 8.033 & $-$12.066 & $-$6.778 & $-$11.947 & 0.117 & 3.169 & 0.470\\
\hline
\multirow[c]{5}{*}{Campos Novos} & UMW & \textbf{19.381} & \textbf{$-$32.761} & $-$24.976 & \textbf{$-$32.509} & \textbf{0.076} & \textbf{0.508} & \textbf{0.056}\\
 & Beta & 7.353 & $-$10.705 & $-$5.515 & $-$10.580 & 0.148 & 3.116 & 0.553\\
 & KW & 7.441 & $-$10.882 & $-$5.692 & $-$10.757 & 0.149 & 3.144 & 0.564\\
 & MK & 18.303 & $-$32.607 & \textbf{$-$27.416} & $-$32.482 & 0.082 & 0.911 & 0.140\\
 \scriptsize{Best: UMW (6)} & UW & 6.058 & $-$8.117 & $-$2.926 & $-$7.992 & 0.139 & 3.449 & 0.505\\
\hline
\multirow[c]{5}{*}{G. B. Munhoz} & UMW & \textbf{24.282} & $-$42.564 & $-$34.546 & $-$42.331 & \textbf{0.043} & \textbf{0.413} & \textbf{0.051}\\
 & Beta & 14.980 & $-$25.961 & $-$20.615 & $-$25.846 & 0.124 & 2.080 & 0.361\\
 & KW & 15.120 & $-$26.240 & $-$20.894 & $-$26.124 & 0.125 & 2.097 & 0.369\\
 & MK & 24.010 & \textbf{$-$44.020} & \textbf{$-$38.674} & \textbf{$-$43.905} & 0.062 & 0.525 & 0.077\\
 \scriptsize{Best: UMW (4)} & UW & 13.094 & $-$22.188 & $-$16.842 & $-$22.072 & 0.123 & 2.539 & 0.384\\
\hline
\multirow[c]{5}{*}{G. P. Souza}  & UMW & \textbf{20.607} & \textbf{$-$35.214} & $-$27.195 & \textbf{$-$34.981} & \textbf{0.074} & \textbf{0.349} & \textbf{0.061}\\
 & Beta & 9.352 & $-$14.705 & $-$9.359 & $-$14.589 & 0.125 & 2.028 & 0.330\\
 & KW & 9.225 & $-$14.450 & $-$9.104 & $-$14.334 & 0.124 & 2.053 & 0.338\\
 & MK & 19.277 & $-$34.554 & \textbf{$-$29.209} & $-$34.439 & 0.096 & 0.556 & 0.088\\
 \scriptsize{Best: UMW (6)} & UW & 11.037 & $-$18.075 & $-$12.729 & $-$17.959 & 0.123 & 1.740 & 0.240\\
\hline
\multirow[c]{5}{*}{Machadinho} & UMW & \textbf{16.038} & $-$26.076 & $-$18.085 & $-$25.840 & \textbf{0.044} & \textbf{0.291} & \textbf{0.037}\\
 & Beta & 8.974 & $-$13.949 & $-$8.622 & $-$13.832 & 0.106 & 1.962 & 0.346\\
 & KW & 9.062 & $-$14.124 & $-$8.797 & $-$14.007 & 0.107 & 1.978 & 0.354\\
 & MK & 15.503 & \textbf{$-$27.006} & \textbf{$-$21.679} & \textbf{$-$26.889} & 0.045 & 0.340 & 0.042\\
 \scriptsize{Best: UMW (4)} & UW & 7.670 & $-$11.341 & $-$6.014 & $-$11.224 & 0.103 & 2.136 & 0.302\\
\hline
\multirow[c]{5}{*}{Mauá} & UMW & \textbf{33.167} & $-$60.335 & $-$53.189 & $-$60.019 & \textbf{0.060} & \textbf{0.433} & \textbf{0.044}\\
 & Beta & 26.481 & $-$48.961 & $-$44.197 & $-$48.806 & 0.130 & 1.933 & 0.322\\
 & KW & 26.568 & $-$49.135 & $-$44.371 & $-$48.979 & 0.132 & 1.996 & 0.337\\
 & MK & 32.583 & \textbf{$-$61.167} & \textbf{$-$56.403} & \textbf{$-$61.011} & 0.066 & 0.585 & 0.072\\
 \scriptsize{Best: UMW (4)} & UW & 22.304 & $-$40.607 & $-$35.843 & $-$40.451 & 0.139 & 2.832 & 0.403\\
\hline
\multirow[c]{5}{*}{Passo Fundo} & UMW & \textbf{54.846} & \textbf{$-$103.692} & $-$95.701 & \textbf{$-$103.456} & \textbf{0.058} & \textbf{0.685} & \textbf{0.091}\\
 & Beta & 47.903 & $-$91.806 & $-$86.479 & $-$91.689 & 0.113 & 1.904 & 0.305\\
 & KW & 48.180 & $-$92.361 & $-$87.034 & $-$92.244 & 0.114 & 1.882 & 0.306\\
 & MK & 53.608 & $-$103.216 & \textbf{$-$97.890} & $-$103.100 & 0.074 & 0.962 & 0.143\\
 \scriptsize{Best: UMW (6)} & UW & 45.853 & $-$87.706 & $-$82.380 & $-$87.590 & 0.118 & 2.464 & 0.366\\
\hline
\multirow[c]{5}{*}{Salto Santiago} & UMW & \textbf{40.806} & $-$75.613 & $-$67.651 & $-$75.375 & \textbf{0.052} & \textbf{0.478} & \textbf{0.063}\\
 & Beta & 35.817 & $-$67.635 & $-$62.327 & $-$67.517 & 0.094 & 1.209 & 0.196\\
 & KW & 36.015 & $-$68.031 & $-$62.723 & $-$67.913 & 0.095 & 1.205 & 0.197\\
 & MK & 40.758 & \textbf{$-$77.515} & \textbf{$-$72.207} & \textbf{$-$77.397} & 0.055 & 0.490 & 0.066\\
 \scriptsize{Best: UMW (4)} & UW & 33.439 & $-$62.878 & $-$57.570 & $-$62.760 & 0.104 & 1.758 & 0.269\\
\hline
\multirow[c]{5}{*}{Santa Clara-PR} & UMW & \textbf{15.505} & $-$25.011 & $-$17.545 & $-$24.729 & \textbf{0.055} & \textbf{0.217} & \textbf{0.035}\\
 & Beta & 8.533 & $-$13.066 & $-$8.089 & $-$12.926 & 0.108 & 1.371 & 0.249\\
 & KW & 8.580 & $-$13.160 & $-$8.183 & $-$13.021 & 0.109 & 1.381 & 0.253\\
 & MK & 15.386 & \textbf{$-$26.773} & \textbf{$-$21.796} & \textbf{$-$26.633} & 0.058 & 0.260 & 0.045\\
 \scriptsize{Best: UMW (4)} & UW & 8.343 & $-$12.687 & $-$7.709 & $-$12.547 & 0.105 & 1.379 & 0.223\\
\hline
\multirow[c]{5}{*}{Segredo} & UMW & \textbf{34.127} & $-$62.253 & $-$54.592 & $-$61.990 & \textbf{0.068} & \textbf{0.698} & \textbf{0.088}\\
 & Beta & 30.181 & $-$56.363 & $-$51.255 & $-$56.232 & 0.095 & 1.618 & 0.259\\
 & KW & 30.360 & $-$56.720 & $-$51.612 & $-$56.589 & 0.096 & 1.631 & 0.264\\
 & MK & 33.519 & \textbf{$-$64.037} & \textbf{$-$58.929} & \textbf{$-$63.906} & 0.070 & 0.725 & 0.092\\
 \scriptsize{Best: UMW (4)} & UW & 28.146 & $-$52.293 & $-$47.185 & $-$52.162 & 0.101 & 2.173 & 0.303\\
\hline
\multicolumn{9}{c}{Win: \quad UMW = 10 \quad\quad Beta = 0 \quad\quad KW = 0 \quad\quad MK = 0 \quad\quad UW = 0}\\
\hline
\end{tabular}
}
\\\scriptsize{Source: Authors.}%
\end{table}

\begin{figure}
\centering
\subfigure[Barra Grande]{\includegraphics[width=0.45\textwidth]{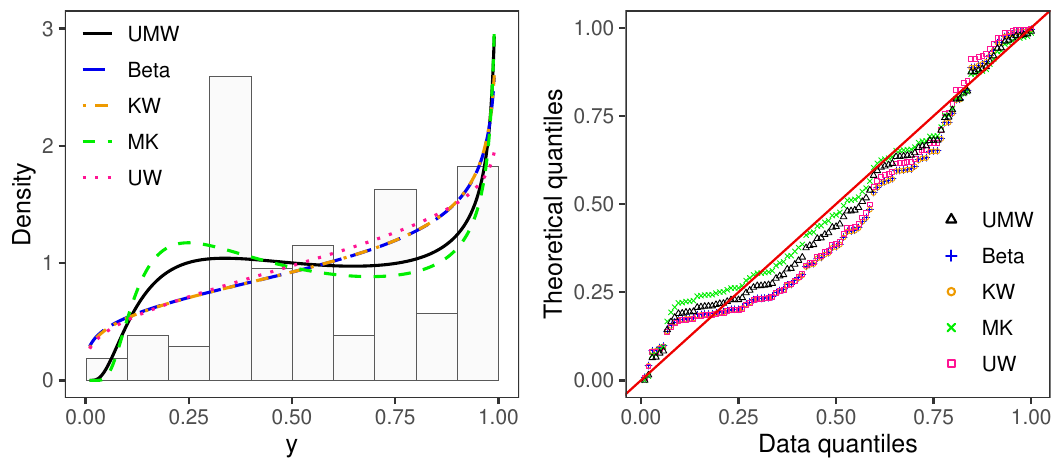}}
\subfigure[Campos Novos]{\includegraphics[width=0.45\textwidth]{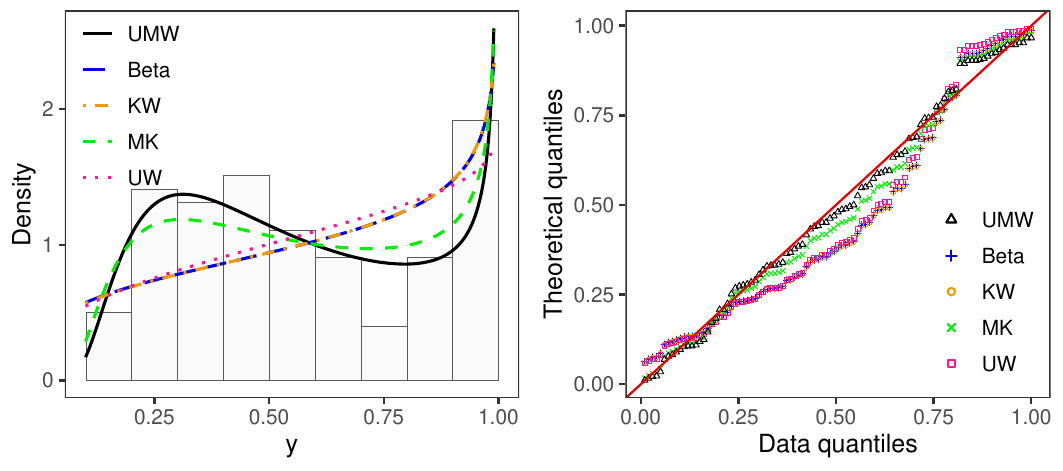}}
\subfigure[G. B. Munhoz]{\includegraphics[width=0.45\textwidth]{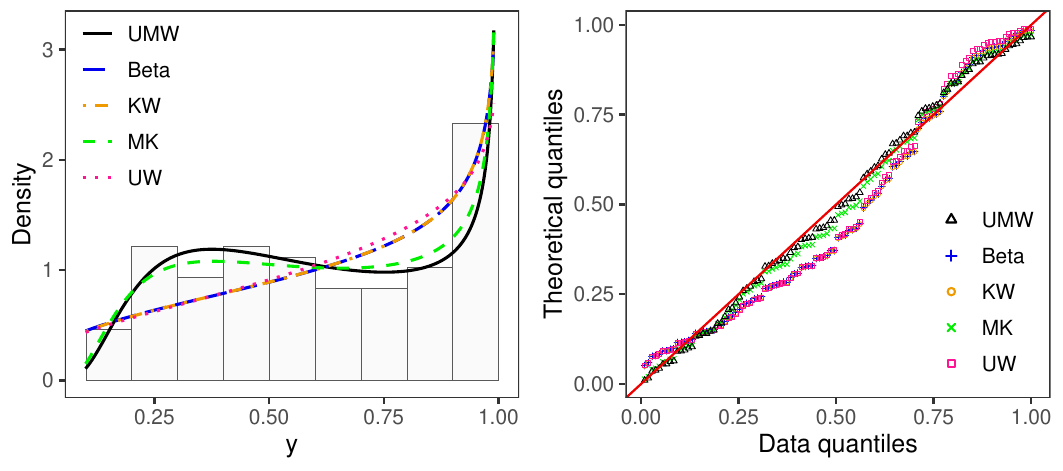}}
\subfigure[G. P. Souza]{\includegraphics[width=0.45\textwidth]{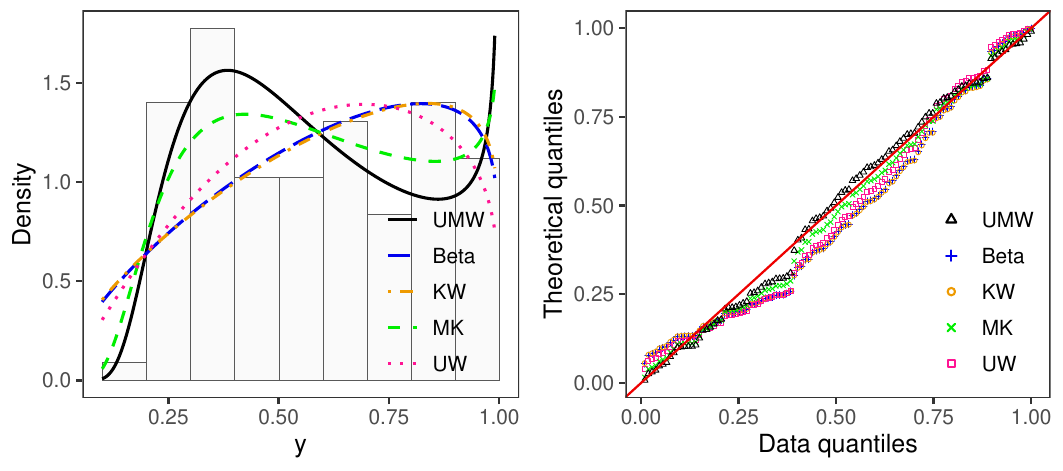}}\\
\subfigure[Machadinho]{\includegraphics[width=0.45\textwidth]{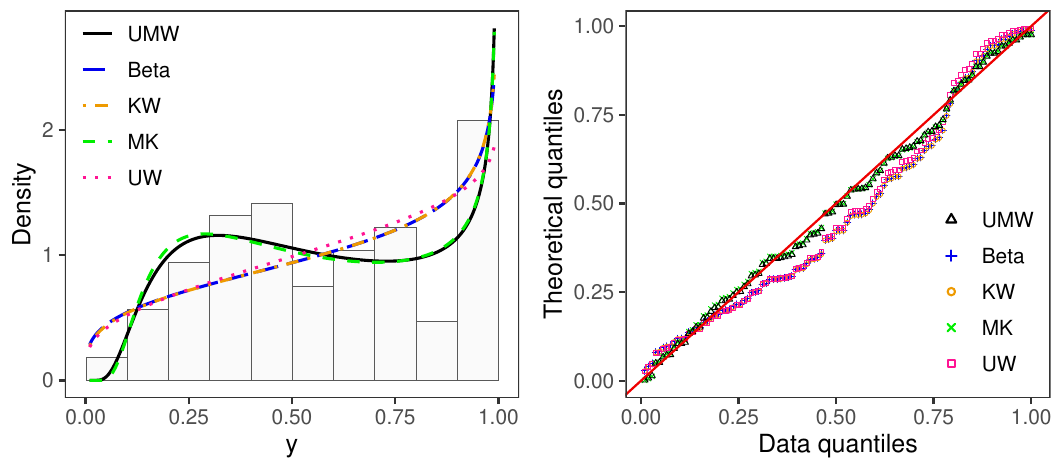}}
\subfigure[Mauá]{\includegraphics[width=0.45\textwidth]{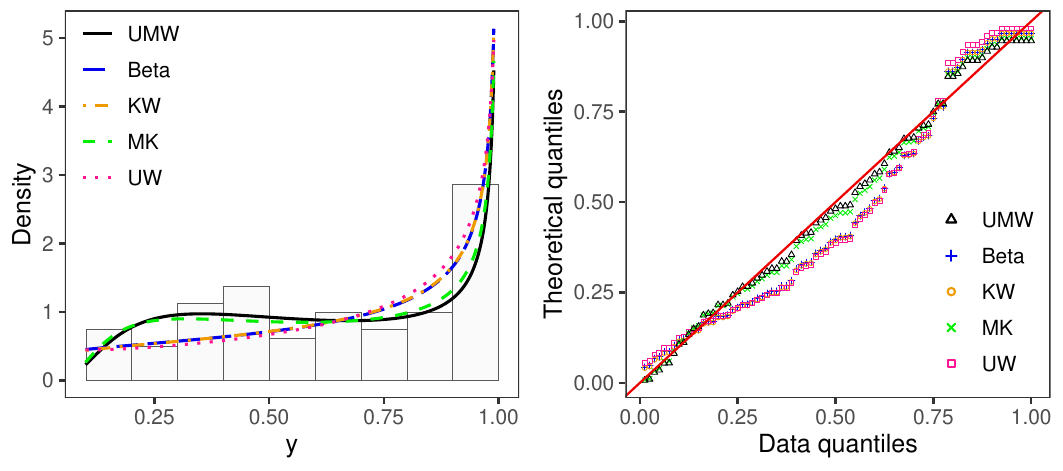}}
\subfigure[Passo Fundo]{\includegraphics[width=0.45\textwidth]{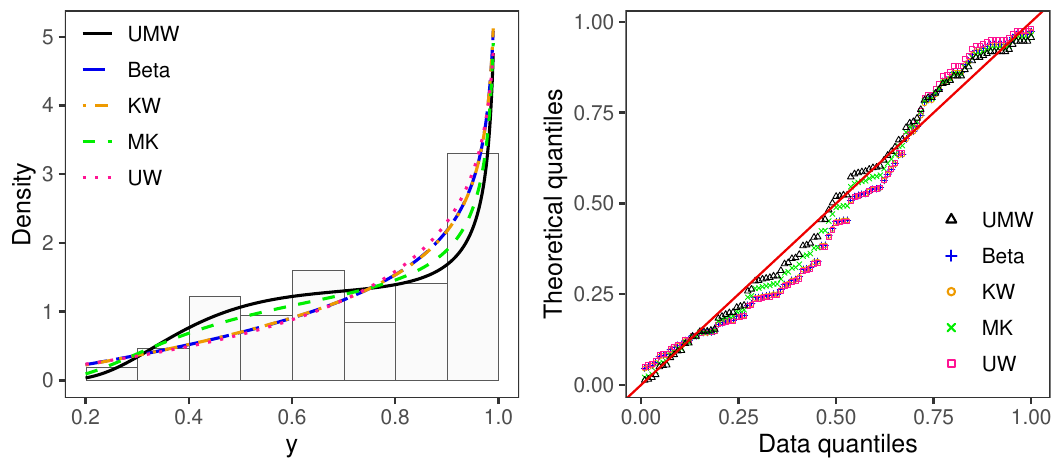}}
\subfigure[Salto Santiago]{\includegraphics[width=0.45\textwidth]{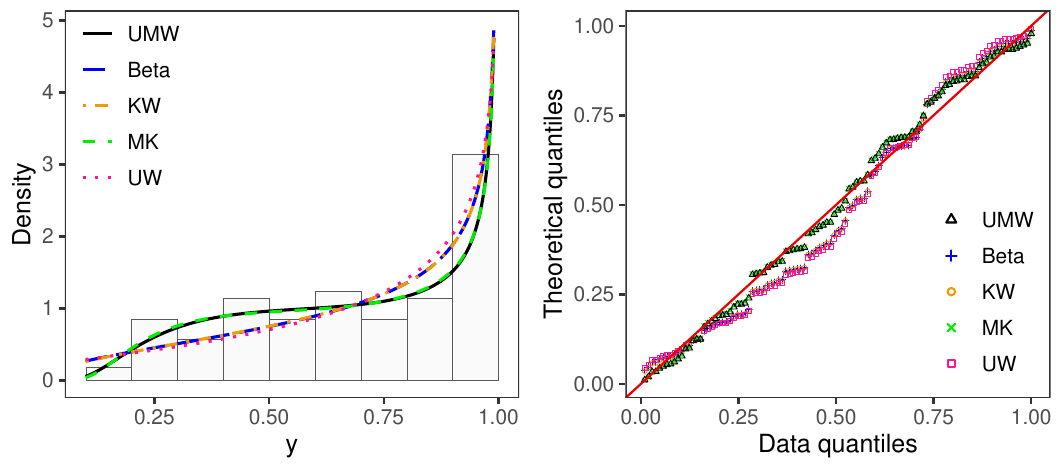}}
\subfigure[Santa Clara-PR]{\includegraphics[width=0.45\textwidth]{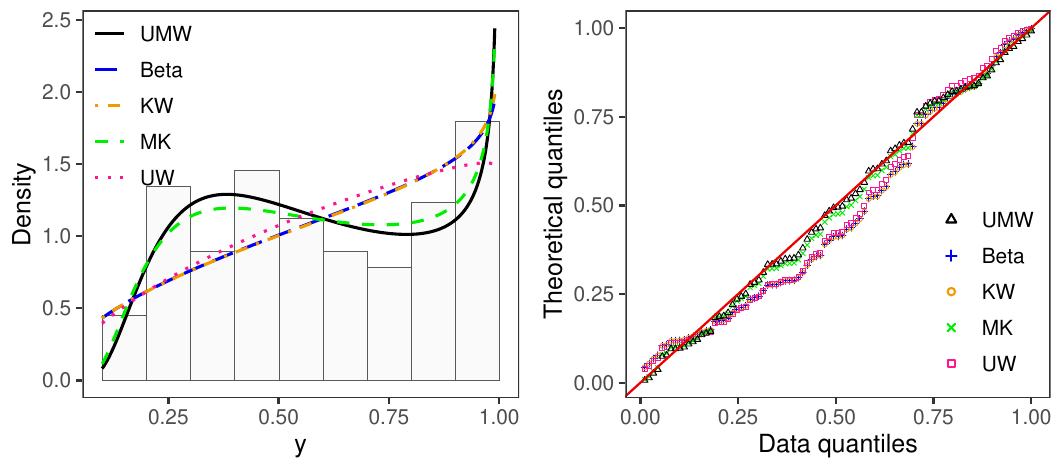}}
\subfigure[Segredo]{\includegraphics[width=0.45\textwidth]{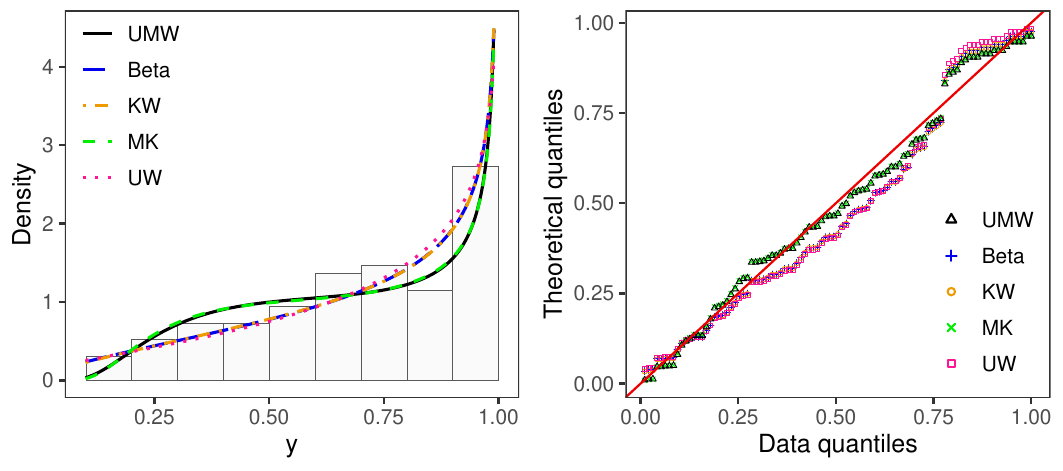}}
\caption{Histograms, density plots, and QQ plots for the relative useful volume of reservoirs in the South region.}\label{fig:g_south}
\scriptsize{Source: Authors.}%
\end{figure}

Analyzing all 26 reservoirs, it was possible to identify scenarios in which the UMW distribution stands out in comparison to the other distributions. These scenarios include behaviors where density exhibits an increasing-decreasing-increasing behavior, as illustrated in Figure \ref{fig:g_south} by the densities shown. 
Another behavior in which the UMW distribution adjusted well was when the distribution of the data presented a peak followed by a high frequency of observations close to one. This can be observed in the Serra do Facão and Itaparica reservoirs, presented in Figures \ref{fig:g_SeMw} and \ref{fig:g_Ne}, respectively, available in the Appendix \ref{ap:ApenA}. 

\subsection*{Reading Skills}

Now, we present an application of the regression model introduced. 
This study aims to analyze whether dyslexia has a significant impact on reading accuracy, even after adjusting the results for the intelligence quotient (IQ) score, 
as initially discussed by \cite{smithson2006better}. 
This data set was also used by \cite{cribari2010beta} in the context of the beta regression model.

The variable of interest ($y$) is accuracy, measured by scores on a reading test administered to 44 children. 
The two covariates are dyslexia, a factor with summed contrasts that differentiates the dyslexic group from the control group, and non-verbal IQ. 
Some summary statistics of these variables are presented in Table \ref{tab:descdislex}. 
The dependent variable has light tails, with kurtosis of \( 1.437 \) (less than 3), indicating fewer extreme values and a greater concentration of data around the mean, which is equal to \( 0.773 \). 
Based on the $p$-value of the AD($p$) test (less than $0.05$), it is concluded that accuracy does not follow a normal distribution. 
The study presents 19 children with dyslexia and 25 without the presence of dyslexia (control group).

\begin{table}
\caption{\label{tab:descdislex}{Descriptive statistics of reading skills in dyslexic children.}}
\centering
\begin{tabular}{lr@{\hspace{0.25cm}}r@{\hspace{0.25cm}}r@{\hspace{0.25cm}}r@{\hspace{0.25cm}}r@{\hspace{0.25cm}}r@{\hspace{0.25cm}}r@{\hspace{0.25cm}}r@{\hspace{0.25cm}}r}
\hline
 & $n$ & Min. & Median & Mean & Max. & SD & AC & K & AD($p$) \\
\hline
Accuracy & 44 & 0.459 & 0.706 & 0.773 & 0.990 & 0.179 & 0.099 & 1.437 & $<$0.001\\
IQ & 44 & $-$1.745 & $-$0.123 & 0 & 1.856 & 1.000 & 0.062 & 2.130 & 0.564\\
Dyslexia & 44 & \multicolumn{7}{c}{Non(-1) = 25 \quad \quad Yes(1) = 19 }\\
\hline
\end{tabular}
\\\scriptsize{Source: Authors.}%
\end{table}

As suggested by previous studies \citep{canterle2019variable}, the interaction between the covariates and the square of the IQ variable was considered as a regressor.  
Thus, the structure of the RQ-UMW model assumed for the median $\mu_t$, i.e., for $\tau = 0.5$, is given by  
\begin{equation}
\log\left(\frac{\mu_t}{1-\mu_t}\right) = \beta_{0} + \beta_{1} \text{IQ}_t^2 + \beta_{2} (\text{Dyslexia}_t \times \text{IQ}_t^2),
\end{equation}
where $t = 1, \dots, 44$.

Tables \ref{tab:coeficientes} and \ref{tab:coeficientes2} present the coefficients and goodness-of-fit measures for the regression models analyzed, respectively. 
For comparison purposes, the RQ-beta, RQ-KW, and RQ-UW regression models were fitted with the same covariates as in the RQ-UMW model. 
For all models, the final specification was selected based on the lowest AIC and the preservation of quantile residual normality, including only covariates that showed at least 10\% significance.
The RQ-UMW model demonstrated superior performance compared to the other models in all considered metrics, presenting the lowest error statistics (MSE, RMSE, MAE, and MAPE), indicating the best fit to the observed reading accuracy data. 
Furthermore, the RQ-UMW model obtained the lowest AIC and BIC values. 
The explainability of the model, measured by \( R^2_G \), was 0.572, indicating that the model is capable of explaining 57.2\% the observed variation in the reading accuracy. 
These results demonstrate the reliability of the RQ-UMW model in terms of fit and explanatory capacity.
Analysis revealed a differential effect of IQ$^2$ on median reading accuracy across groups. Specifically, IQ$^2$ exhibited a positive association with median accuracy among children without dyslexia, suggesting that higher IQ levels correspond to better reading performance. In contrast, this effect was negligible or slightly negative for children with dyslexia, implying that increased IQ does not translate into improvements in median accuracy for this group. This group-specific pattern is captured by the interaction term between dyslexia and IQ$^2$ included in the model.

The graphs of residuals versus indices and simulated envelopes of the RQ-UMW model for reading skills in dyslexic children are presented in Figure \ref{fig:dislex_g} for the regression models under study.
For the RQ-UMW model, we observed the best fit, evidenced by the graphs, with no atypical values (outside the range -3 to 3) and with practically all observations within the 95\% confidence bands. 
In the other models, although the graphs of residuals versus indices do not present atypical values and the observations do not reveal a defined pattern, being distributed randomly, the simulated envelope indicates a worse fit in relation to the RQ-UMW model, with the observations further away from the expected lines.

\begin{table}
\centering
\caption{\label{tab:coeficientes}{Coefficients of the regression models analyzed for reading skills in dyslexic children.}}
 \resizebox{1\textwidth}{!}{
\begin{tabular}{lr@{\hspace{0.25cm}}r@{\hspace{0.25cm}}r@{\hspace{0.1cm}}cr@{\hspace{0.25cm}}r@{\hspace{0.25cm}}r@{\hspace{0.1cm}}cr@{\hspace{0.25cm}}r@{\hspace{0.25cm}}r@{\hspace{0.1cm}}cr@{\hspace{0.25cm}}r@{\hspace{0.25cm}}r@{\hspace{0.1cm}}}
\hline
& \multicolumn{3}{c}{RQ-UMW} && \multicolumn{3}{c}{RQ-beta} && \multicolumn{3}{c}{RQ-KW} && \multicolumn{3}{c}{RQ-UW}\\
\cline{2-4}  \cline{6-8} \cline{10-12} \cline{14-16}
Par  & Estim. & SE & $p$-value&  & Estim. & SE & $p$-value &  & Estim. & SE & $p$-value &  & Estim. & SE & $p$-value \\
\hline
$\beta_0$ & 0.792 & 0.166 & $<$0.001 && 1.706 & 0.173 & $<$0.001 && 1.768 & 0.218 & $<$0.001 &&  1.688 & 1.178 & $<$0.001\\
$\beta_1$ & 0.760 & 0.155 & $<$0.001 && \multicolumn{1}{c}{-}  & \multicolumn{1}{c}{-}  & \multicolumn{1}{c}{-}  && \multicolumn{1}{c}{-}  & \multicolumn{1}{c}{-}  & \multicolumn{1}{c}{-}  && \multicolumn{1}{c}{-}  & \multicolumn{1}{c}{-}  & \multicolumn{1}{c}{-} \\
$\beta_2$ & $-$0.890 & 0.124 & $<$0.001 && $-$0.613 & 0.210 & 0.004  && \multicolumn{1}{c}{-}  & \multicolumn{1}{c}{-}  & \multicolumn{1}{c}{-}  && $-$0.783 & 0.242 & 0.001\\
$\beta_3$ & \multicolumn{1}{c}{-}  & \multicolumn{1}{c}{-}  & \multicolumn{1}{c}{-}  && $-$0.422 & 0.195 & 0.030 && \multicolumn{1}{c}{-}  & \multicolumn{1}{c}{-}  & \multicolumn{1}{c}{-}  &&  $-$0.350 & 0.186 & 0.060\\
$\beta_4$ & \multicolumn{1}{c}{-}  & \multicolumn{1}{c}{-} & \multicolumn{1}{c}{-}  && \multicolumn{1}{c}{-}  & \multicolumn{1}{c}{-}  & \multicolumn{1}{c}{-}  && 0.845 & 0.188 & $<$0.001 && \multicolumn{1}{c}{-}  & \multicolumn{1}{c}{-}  & \multicolumn{1}{c}{-} \\
$\gamma$ & 0.578 & 0.164 & $<$0.001 && 8748.051 & 182.125 & $<$0.001 && 3.610 & 0.692 & $<$0.001 && 1.091 & 0.135 & $<$0.001\\
$\lambda$ & 4.144 & 1.006 & $<$0.001 && 0.964 & 0.180 & $<$0.001 && \multicolumn{1}{c}{-} & \multicolumn{1}{c}{-}& \multicolumn{1}{c}{-} && \multicolumn{1}{c}{-} &\multicolumn{1}{c}{-} & \multicolumn{1}{c}{-}\\
\hline 
\multicolumn{16}{c}{$\beta_0$ - intercept \quad\, $\beta_1$ - IQ$^2$ \quad\, $\beta_2$ - Dyslexia $\times$ IQ$^2$ \quad\, $\beta_3$ - Dyslexia \quad\, $\beta_4$ - IQ }\\
\hline 
\end{tabular}
}
\\\scriptsize{Source: Authors.}%
\end{table}

\begin{table}
\centering
\caption{\label{tab:coeficientes2}{Goodness-of-fit measures for reading skills in dyslexic children from the regression models analyzed.}}
\resizebox{1\textwidth}{!}{
\begin{tabular}{lrrrrrrrrr}
\hline
  & Loglik & MSE & RMSE & MAE & MAPE & AIC & BIC & R$^2_G$ & AD($p$) resid.\\
\hline
RQ-UMW & \textbf{52.547} & \textbf{0.014} & \textbf{0.117} & \textbf{0.086} & \textbf{11.524} & \textbf{$-$95.094} & \textbf{$-$86.173} & \textbf{0.572} & 0.230\\
RQ-beta & 42.781 & 0.019 & 0.139 & 0.112 & 16.953 & $-$75.563 & $-$59.073 & 0.454 & 0.329\\
RQ-KW & 35.367 & 0.025 & 0.160 & 0.124 & 18.528 & $-$64.733 & $-$59.381 & 0.328 & 0.459\\
RQ-UW & 40.589 & 0.024 & 0.155 & 0.122 & 18.518 & $-$73.177 & $-$66.040 & 0.501 & 0.401\\
\hline
\end{tabular}
}
\\\scriptsize{Source: Authors.}%
\end{table}

\begin{figure}
\centering
\includegraphics[width=0.48\textwidth]{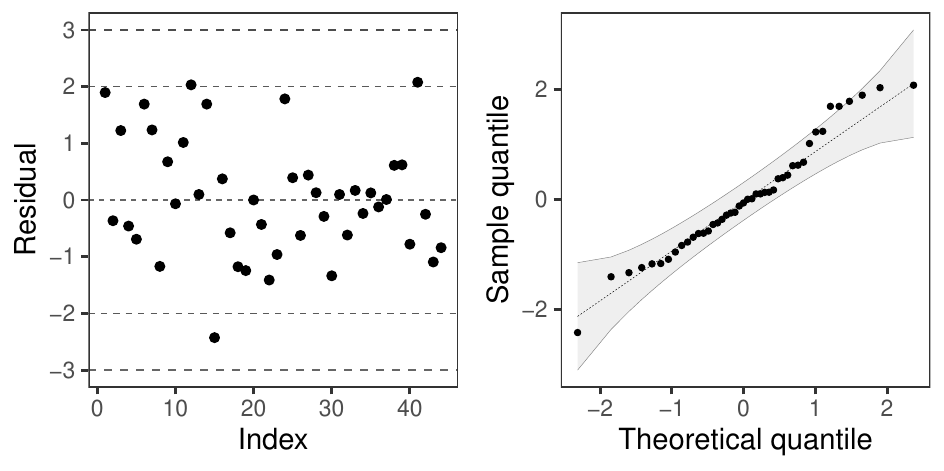}
\includegraphics[width=0.48\textwidth]{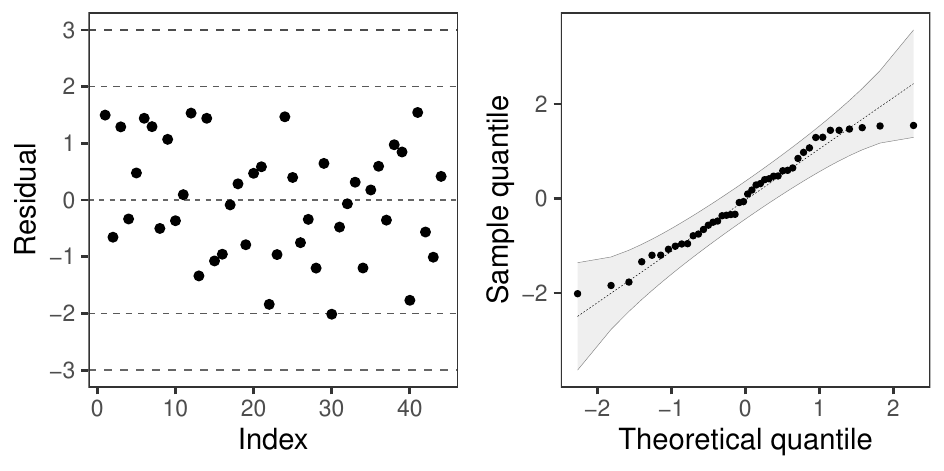}\\[-0.1em]
\footnotesize{(a) RQ-UMW \hspace{6cm} (b) RQ-beta}\\[1em]
\includegraphics[width=0.48\textwidth]{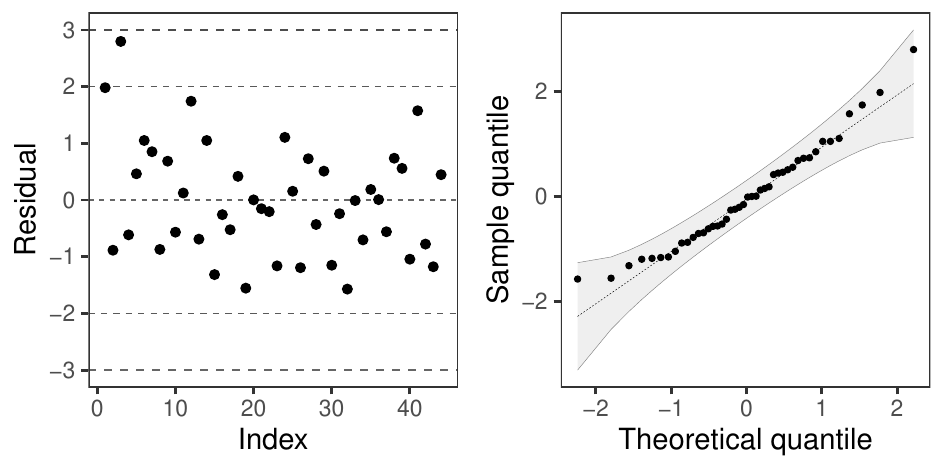}
\includegraphics[width=0.48\textwidth]{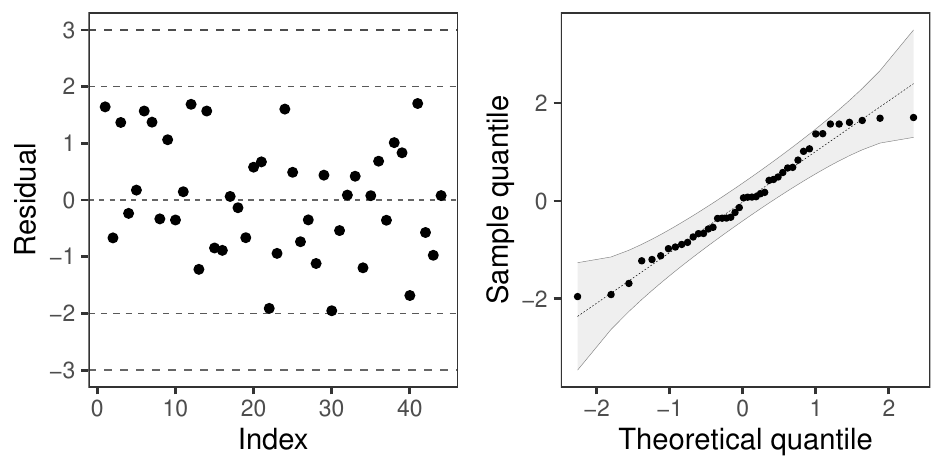}\\[-0.1em]
\footnotesize{(c) RQ-KW \hspace{6.4cm} (d) RQ-UW}
\caption{Residuals versus observation indices and simulated envelopes for reading skills in dyslexic children from the analyzed regression models.}\label{fig:dislex_g}
\scriptsize{Source: Authors.}%
\end{figure}

Figure \ref{fig:quant} presents an analysis of the parameter estimates of the RQ-UMW model for different quantiles, with \( \tau \in \{0.1, 0.2, \dots, 0.9\} \). 
The 95\% confidence intervals and point values were calculated, allowing us to observe how the parameters vary as a function of the quantiles considered.
This analysis reveals that while some exhibit varying behaviors with respect to the quantiles, others remain stable. 
The intercept (\( \hat{\beta}_0 \)) increases as \( \tau \) increases, indicating that higher quantiles are associated with higher values of the intercept. 
The coefficient \( \hat{\beta}_1 \) also shows an increase with \( \tau \), suggesting a greater influence at higher quantiles. In contrast, the coefficient \( \hat{\beta}_2 \) shows a negative estimate that intensifies with increasing \( \tau \), reflecting a stronger and more negative relationship at higher quantiles. 
On the other hand, the parameters \( \hat{\gamma} \) and \( \hat{\lambda} \) remain practically constant across the values of \( \tau \), indicating that their estimates are not affected by variations in the quantiles, except the quantile \( 0.9 \).

\begin{figure}
\centering
\includegraphics[width=0.3\textwidth]{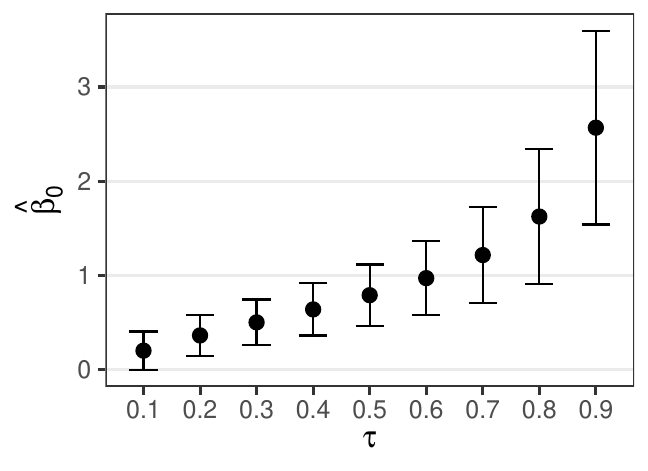}
{\includegraphics[width=0.3\textwidth]{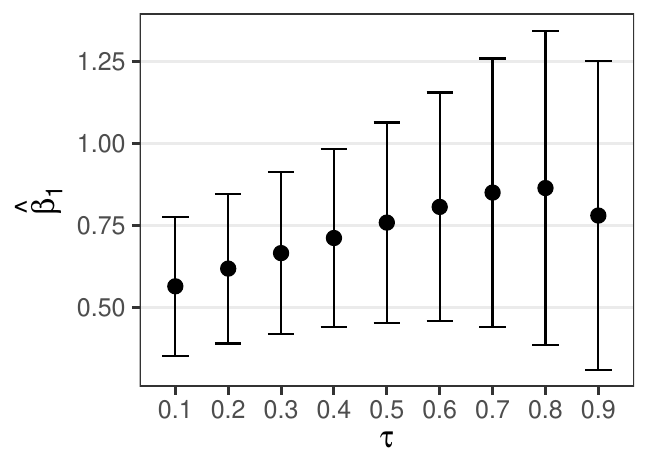}}
{\includegraphics[width=0.3\textwidth]{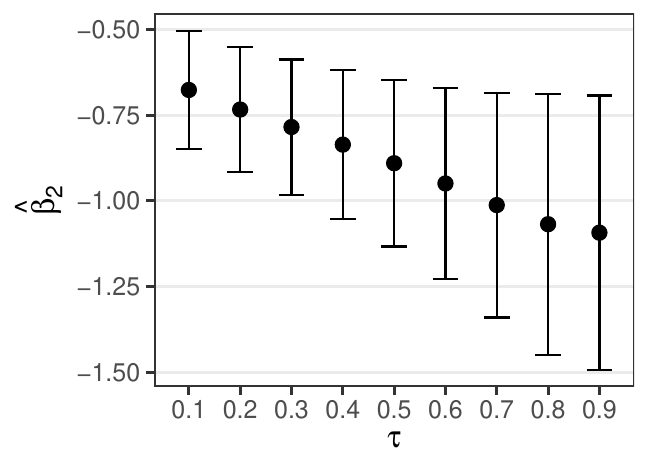}}\\
\includegraphics[width=0.3\textwidth]{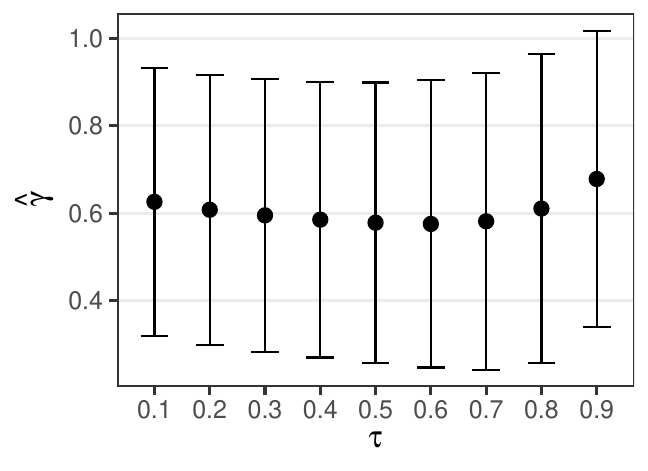}
\includegraphics[width=0.3\textwidth]{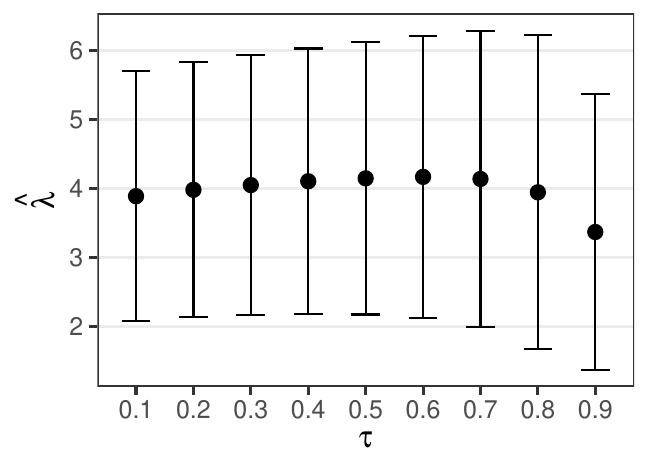}
\caption{Parameter estimates and 95\% confidence intervals for the RQ-UMW model considering \( \tau \in \{0.1, 0.2, \dots, 0.9\} \).}\label{fig:quant}
\scriptsize{Source: Authors.}%
\end{figure}

\section*{Final Considerations}

The primary goal of this work was to propose a new unit distribution based on the MW distribution, to explore its flexibility within the unit interval $(0,1)$.  
In addition, a new quantile regression model was developed based on this distribution.  
Monte Carlo simulations demonstrated that the maximum likelihood estimators of the parameters from the UMW distribution and the RQ-UMW model exhibit desirable properties, such as asymptotic unbiasedness and consistency. 
Furthermore, the results provided evidence of the asymptotic normality of the MLE.  

The empirical analyses conducted on different datasets demonstrated the flexibility and desirable properties of the proposed models, as they effectively captured increasing-decreasing-increasing data behaviors, high frequencies near one, and/or the presence of a peak in the distribution.  
In particular, for the application involving the useful volume of reservoirs, the UMW distribution achieved the best fit in 57.69\% of the cases, outperforming the beta, KW, MK, and UW distributions. 
This result reinforces its suitability and superiority for this type of modeling.  
When applying the RQ-UMW model to reading skills in dyslexic children, a good fit to the data was observed, with quantile residuals following a normal distribution and an explainability of 57.2\%. Furthermore, the RQ-UMW model outperformed key competitors, including the beta, KW, and UW regression models.  

In summary, the results obtained demonstrate the effectiveness and flexibility of the UMW distribution and the RQ-UMW model in different applied contexts, highlighting their ability to capture complex patterns, such as the increasing-decreasing-increasing behavior of the dependent variable, and offering new approaches for statistical analysis in several areas.

\section*{Acknowledgements}
This research was partially funded by the Serrapilheira Institute (grant number 2211-41692), FAPERGS (Fundação de Amparo à Pesquisa do Estado do Rio Grande do Sul - grant number 23/2551-0001595-1), CNPq (Conselho Nacional de Desenvolvimento Científico e Tecnológico - grant numbers 306274/2022-1 and 308578/2023-6), and CAPES (Coordenação de Aperfeiçoamento de Pessoal de Nível Superior - Finance Code 001). 
The content of this work is solely the responsibility of the authors and does not necessarily represent the official views of the funding agencies.

\bibliography{Article}

 \newpage

\section*{Appendix}
\appendix 

\section{Observed Information Matrix}
\label{ap:ApenOIM}

The UMW distribution is a model for which closed-form expressions for the moments cannot be established. 
Therefore, the use of the Fisher information matrix becomes challenging.
However, according to \cite{pawitan2001}, using the observed information matrix is a common strategy and provides approximate estimates for the Fisher information matrix, which is useful in defining the asymptotic distribution of the MLE.

The observed information matrix of the UMW distribution can be written as
\[
\bm{\mbox{J}}(\bm{{\theta}_1}) = -
\begin{bmatrix}
\mbox{J}_{\alpha\alpha}(\bm{\theta}_1) & \mbox{J}_{\alpha\gamma}(\bm{\theta}_1) & \mbox{J}_{\alpha\lambda}(\bm{\theta}_1)\\[1.5em]
\mbox{J}_{\gamma\alpha}(\bm{\theta}_1) & \mbox{J}_{\gamma\gamma}(\bm{\theta}_1) & \mbox{J}_{\gamma\lambda}(\bm{\theta}_1)\\[1.5em]
\mbox{J}_{\lambda\alpha}(\bm{\theta}_1) & \mbox{J}_{\lambda\gamma}(\bm{\theta}_1) & \mbox{J}_{\lambda\lambda}(\bm{\theta}_1)
\end{bmatrix}.
\]
The elements referring to the coordinates of the observed information matrix relative to the parameters \(\alpha\), \(\gamma\), and \(\lambda\) of the UMW distribution are given, respectively, by:
\begin{align*}
\mbox{J}_{\alpha\alpha}(\bm{\theta}_1) = \sum_{t=1}^{n}\frac{\partial^2\ell_{1,t}(\bm{\theta}_1,{{y_t}})}{\partial \alpha^2}, \quad \mbox{J}_{\alpha\gamma}(\bm{\theta}_1) = \sum_{t=1}^{n}\dfrac{\partial^{2}\ell_{1,t}(\bm{\theta}_1,{{y_t}})}{\partial\alpha\partial\gamma} = \mbox{J}_{\gamma\alpha}(\bm{\theta}_1), \\
\mbox{J}_{\gamma\gamma}(\bm{\theta}_1) = \sum_{t=1}^{n}\frac{\partial^2\ell_{1,t}(\bm{\theta}_1,{{y_t}})}{\partial \gamma^2}, \quad
\mbox{J}_{\alpha\lambda}(\bm{\theta}_1) = \sum_{t=1}^{n}\dfrac{\partial^{2}\ell_{1,t}(\bm{\theta}_1,{{y_t}})}{\partial\alpha\partial\lambda} = \mbox{J}_{\lambda\alpha}(\bm{\theta}_1),
 \\
\mbox{J}_{\lambda\lambda}(\bm{\theta}_1) = \sum_{t=1}^{n}\frac{\partial^2\ell_{1,t}(\bm{\theta}_1,{{y_t}})}{\partial \lambda^2}, \quad \mbox{J}_{\gamma\lambda}(\bm{\theta}_1) = \sum_{t=1}^{n}\dfrac{\partial^{2}\ell_{1,t}(\bm{\theta}_1,{{y_t}})}{\partial\gamma\partial\lambda} = \mbox{J}_{\lambda\gamma}(\bm{\theta}_1), 
\end{align*}
where
\begin{align*}
    &\dfrac{\partial^{2}\ell_{1,t}(\bm{\theta}_1,{{y_t}})}{\partial\alpha^2} =  -\dfrac{1}{{\alpha}^2} := \texttt{r}_t^\star, \\
    &\dfrac{\partial^{2}\ell_{1,t}(\bm{\theta}_1,{{y_t}})}{\partial\gamma^2} =   -\dfrac{1}{\left[{\gamma}-{\lambda}\log\left(y_\text{t}\right)\right]^2}-\dfrac{{\alpha}\left[-\log\left(y_\text{t}\right)\right]^{\gamma}\log^2\left(-\log\left(y_\text{t}\right)\right)}{y_\text{t}^{\lambda}} := \texttt{s}_t^\star,\\
    &\dfrac{\partial^{2}\ell_{1,t}(\bm{\theta}_1,{{y_t}})}{\partial\lambda^2} = -\dfrac{\log^2\left(y_\text{t}\right)}{\left[{\gamma}-{\lambda}\log\left(y_\text{t}\right)\right]^2}-\dfrac{{\alpha}\left[-\log\left(y_\text{t}\right)\right]^{\gamma}\log^2\left(y_\text{t}\right)}{y_\text{t}^{\lambda}} := \texttt{u}_t^\star,\\
    &\dfrac{\partial^{2}\ell_{1,t}(\bm{\theta}_1,{{y_t}})}{\partial\alpha\partial\gamma} = -\dfrac{\left[-\log\left(y_\text{t}\right)\right]^{\gamma}\log\left(-\log\left(y_\text{t}\right)\right)}{y_\text{t}^{\lambda}} := \texttt{v}_t^\star,\\
    &\dfrac{\partial^{2}\ell_{1,t}(\bm{\theta}_1,{{y_t}})}{\partial\alpha\partial\lambda} = \dfrac{\left[-\log\left(y_\text{t}\right)\right]^{\gamma}\log\left(y_\text{t}\right)}{y_\text{t}^{\lambda}} := \texttt{z}_t^\star,\\
    &\dfrac{\partial^{2}\ell_{1,t}(\bm{\theta}_1,{{y_t}})}{\partial\gamma\partial\lambda} = \dfrac{\log\left(y_\text{t}\right)}{\left[{\gamma}-{\lambda}\log\left(y_\text{t}\right)\right]^2}+\dfrac{{\alpha}\left[-\log\left(y_\text{t}\right)\right]^{\gamma}\log\left(y_\text{t}\right)\log\left(-\log\left(y_\text{t}\right)\right)}{y_\text{t}^{\lambda}} := \texttt{d}_t^\star,
\end{align*}
note that
$$\mbox{J}_{\alpha\alpha}(\bm{\theta}_1) = \texttt{r}^\star \bm{1}_n^{\top}, \quad \mbox{J}_{\gamma\gamma}(\bm{\theta}_1) = \texttt{s}^\star\bm{1}_n^{\top}, \quad \mbox{J}_{\lambda\lambda}(\bm{\theta}_1) = \texttt{u}^\star\bm{1}_n^{\top}, $$
$$\mbox{J}_{\alpha\gamma}(\bm{\theta}_1) = \texttt{v}^\star\bm{1}_n^{\top}, \quad \mbox{J}_{\alpha\lambda}(\bm{\theta}_1) = \texttt{z}^\star\bm{1}_n^{\top}, \quad \mbox{J}_{\gamma\lambda}(\bm{\theta}_1) = \texttt{d}^\star\bm{1}_n^{\top},$$
with \(\bm{\texttt{r}}^\star=(\texttt{r}^\star_1,\ldots,\texttt{r}^\star_n)\), \(\bm{\texttt{s}}^\star=(\texttt{s}^\star_1,\ldots,\texttt{s}^\star_n)\), \(\bm{\texttt{u}}^\star=(\texttt{u}^\star_1,\ldots,\texttt{u}^\star_n)\), \(\bm{\texttt{v}}^\star=(\texttt{v}^\star_1,\ldots,\texttt{v}^\star_n)\), \(\bm{\texttt{z}}^\star=(\texttt{z}^\star_1,\ldots,\texttt{z}^\star_n)\), \(\bm{\texttt{d}}^\star=(\texttt{d}^\star_1,\ldots,\texttt{d}^\star_n)\) and \(\bm{1}_n^{\top}\) is a column vector of ones of dimension \(n\).

The observed information matrix of the RQ-UMW model can be written as
\[
\bm{\mbox{L}}(\bm{{\theta}_2}) = -
\begin{bmatrix}
\mbox{L}_{\gamma\gamma}(\bm{\theta}_2) & \mbox{L}_{\gamma\lambda}(\bm{\theta}_2) &
\mbox{L}_{\gamma\bm{\beta}_j}(\bm{\theta}_2) \\[1.5em]
\mbox{L}_{\lambda\gamma}(\bm{\theta}_2) & \mbox{L}_{\lambda\lambda}(\bm{\theta}_2) &
\mbox{L}_{\lambda\bm{\beta}_j}(\bm{\theta}_2) \\[1.5em]
\mbox{L}_{\bm{\beta}_j\gamma}(\bm{\theta}_2) & \mbox{L}_{\bm{\beta}_j\lambda}(\bm{\theta}_2) & \mbox{L}_{\bm{\beta}_j\bm{\beta}_l}(\bm{\theta}_2)
\end{bmatrix} 
, \quad l = 1, \ldots, k.
\]
The elements referring to the coordinates of the observed information matrix relative to the parameters \(\gamma\), \(\lambda\), and \(\bm{\beta}_j\) of the RQ-UMW model are given, respectively, by:
\begin{align*}
\mbox{L}_{\gamma\gamma}(\bm{\theta}_2) = \sum_{t=1}^{n}\frac{\partial^2\ell_{2,t}(\bm{\theta}_2,{{y_t}})}{\partial \gamma^2}, \quad \mbox{L}_{\gamma\bm{\beta}_j}(\bm{\theta}_2) = \sum_{t=1}^{n}\frac{\partial^2 \ell_{2,t}(\bm{\theta}_2,{{y_t}})}{\partial {\mu_t}\partial {\gamma}}\frac{d \mu_t}{d {\zeta_t}} \frac{\partial \zeta_t}{\partial \bm{\beta}_j} = \mbox{L}_{\bm{\beta}_j\gamma}(\bm{\theta}_2), \\
\mbox{L}_{\lambda\lambda}(\bm{\theta}_2) = \sum_{t=1}^{n}\frac{\partial\ell_{2,t}(\bm{\theta}_2,{{y_t}})}{\partial \lambda^2}, \quad \mbox{L}_{\lambda\bm{\beta}_j}(\bm{\theta}_2) = \sum_{t=1}^{n}\frac{\partial^2 \ell_{2,t}(\bm{\theta}_2,{{y_t}})}{\partial {\mu_t}\partial {\lambda}}\frac{d \mu_t}{d {\zeta_t}}\frac{\partial \zeta_t}{\partial \bm{\beta}_j} = \mbox{L}_{\bm{\beta}_j\lambda}(\bm{\theta}_2), \\
\mbox{L}_{\bm{\beta}_j\bm{\beta}_l}(\bm{\theta}_2) =  \sum_{t=1}^{n}\left[\frac{\partial^2 \ell_{2,t}(\bm{\theta}_2,{{y_t}})}{\partial \mu_{\tau,t}^2}\frac{d \mu_t}{d {\zeta_t}} +\frac{\partial \ell_{2,t}(\bm{\theta}_2,{{y_t}})}{\partial \mu_{\tau,t}}\frac{\partial }{\partial {\mu_t}}\left(\frac{d \mu_t}{d {\zeta_t}}\right)\right]\frac{d \mu_t}{d {\zeta_t}} \frac{\partial \zeta_t}{\partial \bm{\beta}_l}\frac{\partial \zeta_t}{\partial \bm{\beta}_j},\quad  \\
\mbox{L}_{\gamma\lambda}(\bm{\theta}_2) = \sum_{t=1}^{n}\frac{\partial^2\ell_{2,t}(\bm{\theta}_2,{{y_t}})}{\partial \gamma \partial \lambda} = \mbox{L}_{\lambda\gamma}(\bm{\theta}_2), \quad\quad\quad\quad\quad\quad\quad\quad\quad
\end{align*}
where
\begin{align*}
    &\dfrac{\partial^{2}\ell_{2,t}(\bm{\theta}_2,{{y_t}})}{\partial\gamma^2} =   \dfrac{{\texttt{A}_{t}}\left({\texttt{B}_{t}}\right)^2}{y_\text{t}^{\lambda}\left[-\log\left(\mu_{\tau,t}\right)\right]^{\gamma}} -\dfrac{1}{\left[{\gamma}-{\lambda}\log\left(y_\text{t}\right)\right]^2} := \texttt{r}_t^\diamond, \\
   &\dfrac{\partial^{2}\ell_{2,t}(\bm{\theta}_2,{{y_t}})}{\partial\lambda^2} =  
    \dfrac{{\texttt{A}_{t}}[\log\left(\mu_{\tau,t}\right)-\log\left(y_\text{t}\right)]^2}{y_\text{t}^{\lambda}\left[-\log\left(\mu_{\tau,t}\right)\right]^{\gamma}}  -\dfrac{\log^2\left(y_\text{t}\right)}{\left[{\gamma}-{\lambda}\log\left(y_\text{t}\right)\right]^2} := \texttt{s}_t^\diamond,\\
    &\dfrac{\partial^{2}\ell_{2,t}(\bm{\theta}_2,{{y_t}})}{\partial\gamma\partial\lambda} = \dfrac{\log\left(y_\text{t}\right)}{\left[{\gamma} -{\lambda}\log\left(y_\text{t}\right)\right]^2}+ \dfrac{{\texttt{A}_{t}}{\texttt{B}_{t}}\left[\log\left(\mu_{\tau,t}\right) - \log\left(y_\text{t}\right)\right]}{y_\text{t}^{\lambda}\left[-\log\left(\mu_{\tau,t}\right)\right]^{\gamma}} := \texttt{u}_t^\diamond,\\
    &\frac{\partial^2 \ell_{2,t}(\bm{\theta}_2,{{y_t}})}{\partial {\mu_{\tau,t}}\partial {\gamma}} =  \dfrac{1}{\mu_{\tau,t}\log\left(\mu_{\tau,t}\right)}  \left( -\frac{{\texttt{A}_{t}}\left\{{\texttt{B}_{t}}\left[{\gamma}-{\lambda}\log\left(\mu_{\tau,t}\right)\right]+1\right\}}{y_\text{t}^{\lambda}\left[-\log\left(\mu_{\tau,t}\right)\right]^{\gamma}}-1\right) := \texttt{v}_t^\diamond,\\
    &\frac{\partial^2 \ell_{2,t}(\bm{\theta}_2,{{y_t}})}{\partial {\mu_t}\partial {\lambda}} = -\dfrac{{\texttt{A}_{t}}\left[\log\left(y_\text{t}\right)-\log\left(\mu_{\tau,t}\right)\right][{\lambda}\log\left(\mu_{\tau,t}\right)-{\gamma}]}{y_\text{t}^{\lambda}\mu_{\tau,t}\left[-\log\left(\mu_{\tau,t}\right)\right]^{\gamma}\log\left(\mu_{\tau,t}\right)} + \frac{y_\text{t}^{\lambda}\left[-\log\left(\mu_{\tau,t}\right)\right]^{\gamma}+{\texttt{A}_{t}}}{y_\text{t}^{\lambda}\mu_{\tau,t}\left[-\log\left(\mu_{\tau,t}\right)\right]^{\gamma}} := \texttt{z}_t^\diamond, \\
    & \frac{\partial^2 \ell_{2,t}(\bm{\theta}_2,{{y_t}})}{\partial \mu_{\tau,t}^2} = \dfrac{{\gamma}}{\mu_{\tau,t}^2\log^2\left(\mu_{\tau,t}\right)} -\dfrac{{\gamma}\log\left(\mu_{\tau,t}\right)\left\{\left(2{\lambda}-1\right){\texttt{A}_{t}}-y_\text{t}^{\lambda}\left[-\log\left(\mu_{\tau,t}\right)\right]^{\gamma}\right\}-{\texttt{A}_{t}}{\gamma}\left({\gamma}+1\right)}{y_\text{t}^{\lambda}\mu_{\tau,t}^2\left[-\log\left(\mu_{\tau,t}\right)\right]^{\gamma}\log^2\left(\mu_{\tau,t}\right)} \\ & \quad\quad\quad\quad\quad\quad -\dfrac{{\lambda}\left\{y_\text{t}^{\lambda}\left[-\log\left(\mu_{\tau,t}\right)\right]^{\gamma}+\left(1-{\lambda}\right){\texttt{A}_{t}}\right\}}{y_\text{t}^{\lambda}\mu_{\tau,t}^2\left[-\log\left(\mu_{\tau,t}\right)\right]^{\gamma}} := \texttt{w}_t^\diamond, 
\end{align*}   
with \(\frac{\partial \zeta_t}{\partial \beta_l}=x_{tl}\), $\frac{\partial }{\partial {\mu_t}}\left(\frac{d \mu_t}{d {\zeta_t}}\right) = - {g^{\prime\prime}({\mu_t})}{\left[g^{\prime}({\mu_t})\right]^{-2}}$ {and} \(g^{\prime\prime}(\cdot)\) denotes the second derivative of the function \(g(\cdot)\). Note that
$$\mbox{L}_{\gamma\gamma}(\bm{\theta}_2) = \texttt{r}^\diamond \bm{1}_n^{\top} , \quad \mbox{L}_{\lambda\lambda}(\bm{\theta}_2) = \texttt{s}^\diamond\bm{1}_n^{\top}, \quad \mbox{L}_{\gamma\lambda}(\bm{\theta}_2) = \texttt{u}^\diamond\bm{1}_n^{\top}, $$
$$\mbox{L}_{\bm{\beta}_j\gamma}(\bm{\theta}_2) = \boldsymbol{\mbox{X}}^{\top} \bm{\mbox{T}} \bm{\texttt{v}}^\diamond , \quad \mbox{L}_{\bm{\beta}_j\lambda}(\bm{\theta}_2) = \boldsymbol{X}^{\top} \bm{\mbox{T}} \bm{\texttt{z}}^\diamond, \quad \mbox{L}_{\bm{\beta}_j\bm{\beta}_l}(\bm{\theta}_2) = \boldsymbol{\mbox{X}}^{\top} \boldsymbol{\mbox{M}}\bm{\mbox{T}} \boldsymbol{\mbox{X
}},$$
where \(\boldsymbol{\mbox{M}} = \mbox{diag}\left\{ \texttt{w}^\diamond \bm{\mbox{T}} + \texttt{w} \bm{\mbox{T}}^\diamond \right\}\), \(\bm{\texttt{r}}^\diamond=(\texttt{r}^\diamond_1,\ldots,\texttt{r}^\diamond_n)\), \(\bm{\texttt{s}}^\diamond=(\texttt{s}^\diamond_1,\ldots,\texttt{s}^\diamond_n)\), \(\bm{\texttt{u}}^\diamond=(\texttt{u}^\diamond_1,\ldots,\texttt{u}^\diamond_n)\), \(\bm{\texttt{v}}^\diamond=(\texttt{v}^\diamond_1,\ldots,\texttt{v}^\diamond_n)^\top\), \(\bm{\texttt{z}}^\diamond=(\texttt{z}^\diamond_1,\ldots,\texttt{z}^\diamond_n)^\top\), \(\bm{\mbox{T}}^\diamond = \mbox{diag} \{ - g^{\prime\prime}({\mu_1}) \left[g^{\prime}({\mu_1})\right]^{-2}, \ldots, - {g^{\prime\prime}({\mu_n})}{\left[g^{\prime}({\mu_n})\right]^{-2}}\}\), \(\bm{\texttt{w}}^\diamond=(\texttt{w}^\diamond_1,\ldots,\texttt{w}^\diamond_n)\) and \(\bm{1}_n^{\top}\) is a column vector of ones of dimension \(n\).

\section{Useful Volume Application}
\label{ap:ApenA}

This appendix presents the histograms, density plots, and QQ plots of the 16 reservoirs from the Northeast (NE)/North (N) and Southeast (SE)/Center-West (CO) regions, shown respectively in Figures~\ref{fig:g_Ne} and~\ref{fig:g_SeMw}, fitted to the UMW, beta, KW, MK, and UW distributions.

\begin{figure}[H]
\centering
\subfigure[Irapé]{\includegraphics[width=0.45\textwidth]{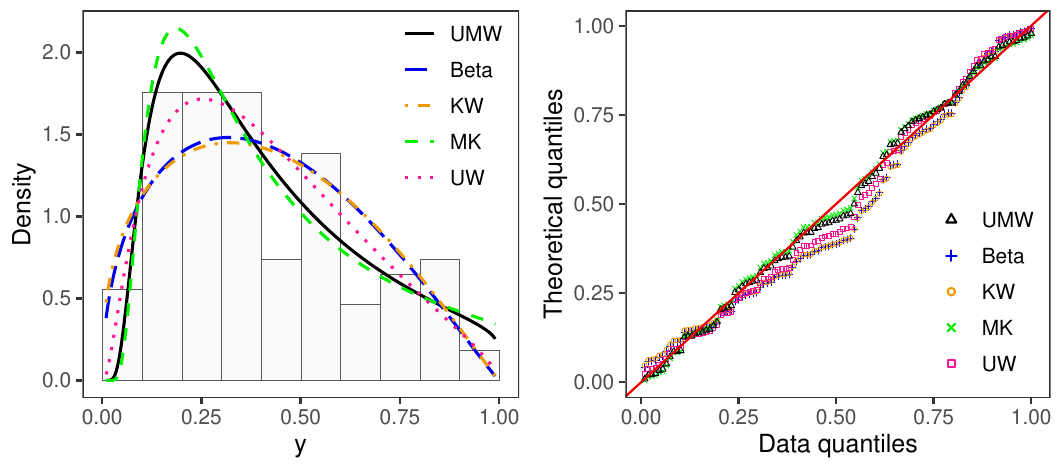}}
\subfigure[Itaparica]{\includegraphics[width=0.45\textwidth]{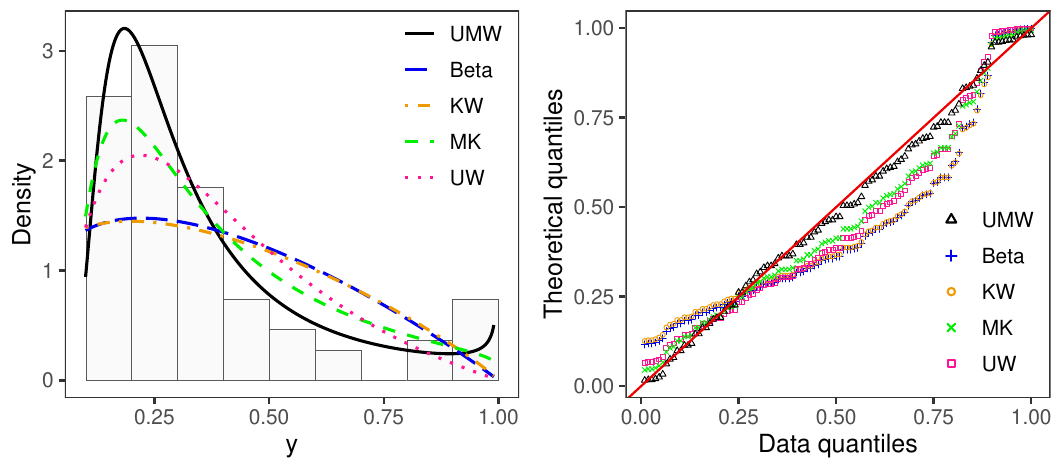}}
\subfigure[Sobradinho]{\includegraphics[width=0.45\textwidth]{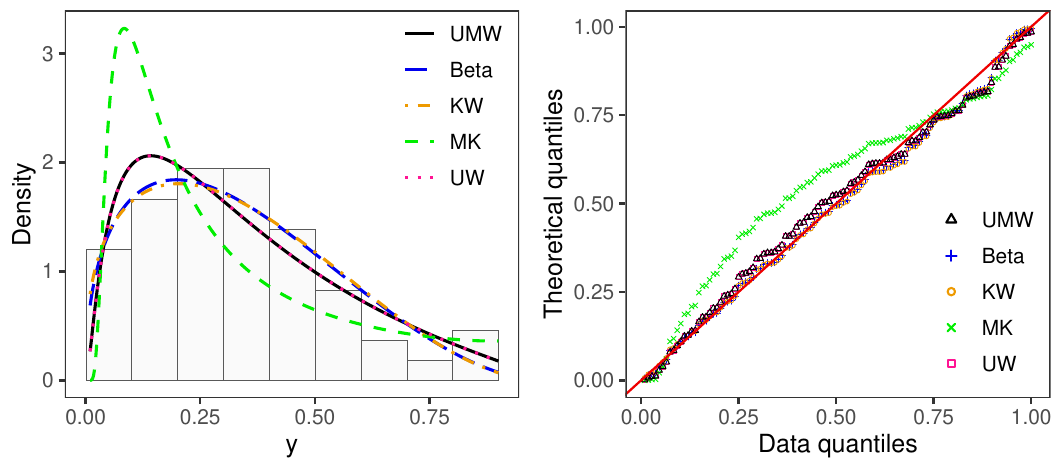}}
\subfigure[Três Marias]{\includegraphics[width=0.45\textwidth]{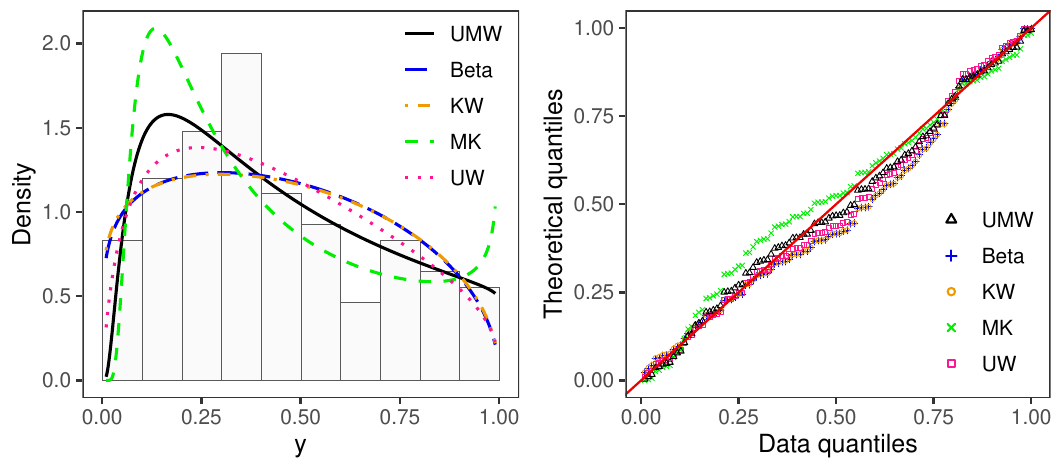}}
\subfigure[Serra da Mesa]{\includegraphics[width=0.45\textwidth]{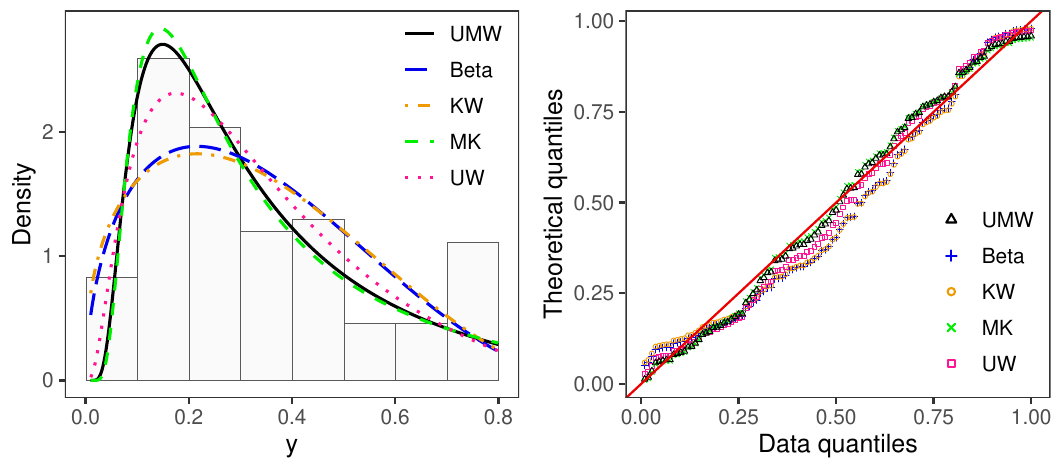}}
\subfigure[Tucuruí]{\includegraphics[width=0.45\textwidth]{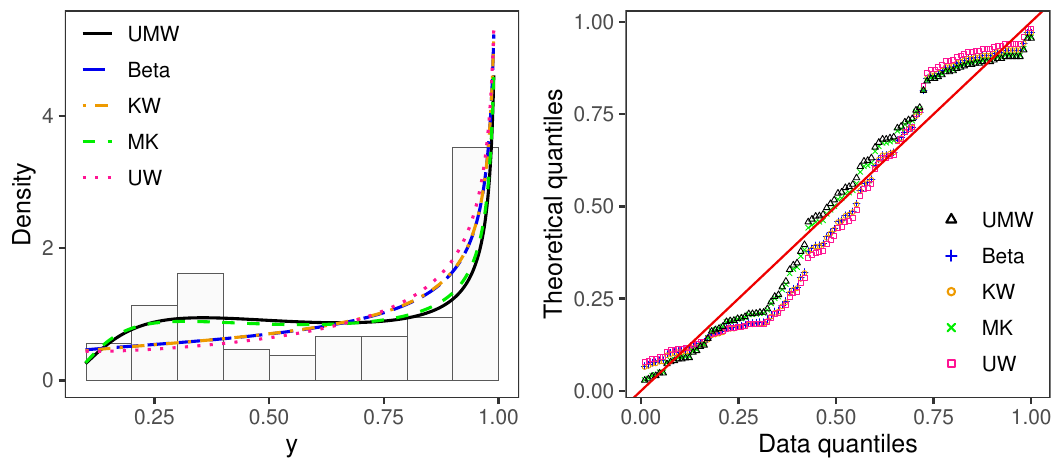}}
\caption{Histograms, density plots, and QQ plots for the relative useful volume of the reservoirs in the NE and N regions.}\label{fig:g_Ne}
\scriptsize{Source: Authors.}%
\end{figure}

\begin{figure}[H]
\centering
\subfigure[A. Vermelha]{\includegraphics[width=0.45\textwidth]{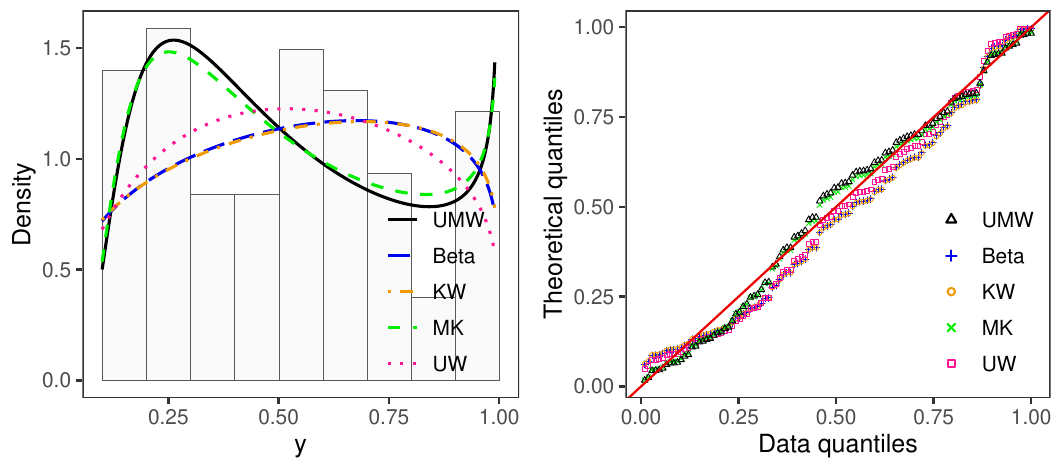}}
\subfigure[B. Bonita]{\includegraphics[width=0.45\textwidth]{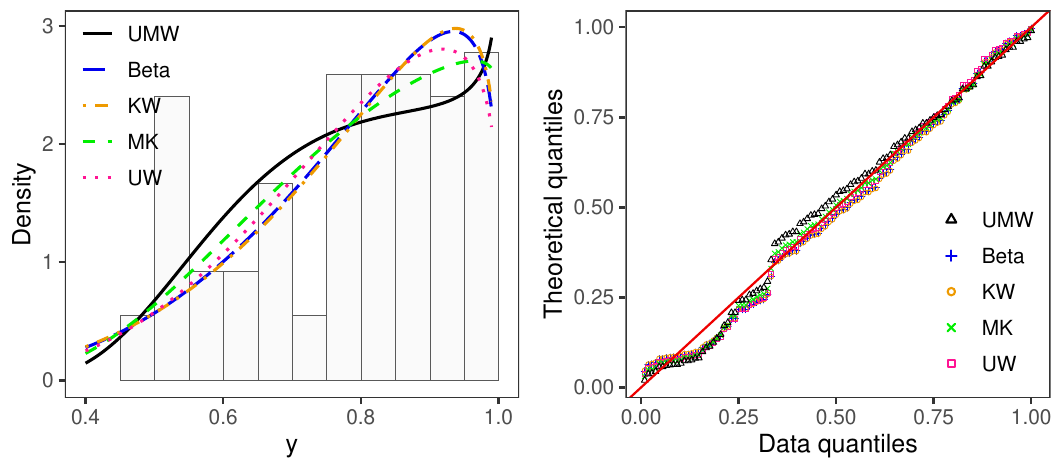}}
\subfigure[Capivara]{\includegraphics[width=0.45\textwidth]{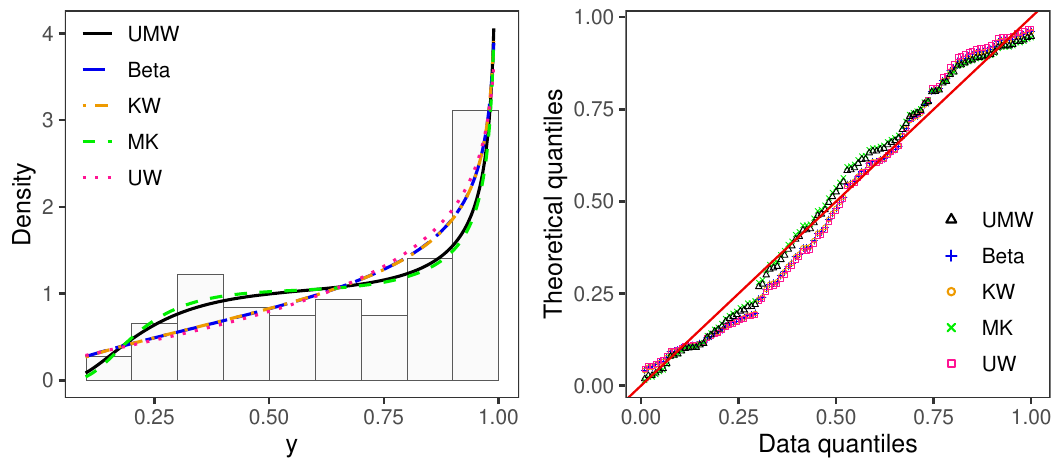}}
\subfigure[Chavantes]{\includegraphics[width=0.45\textwidth]{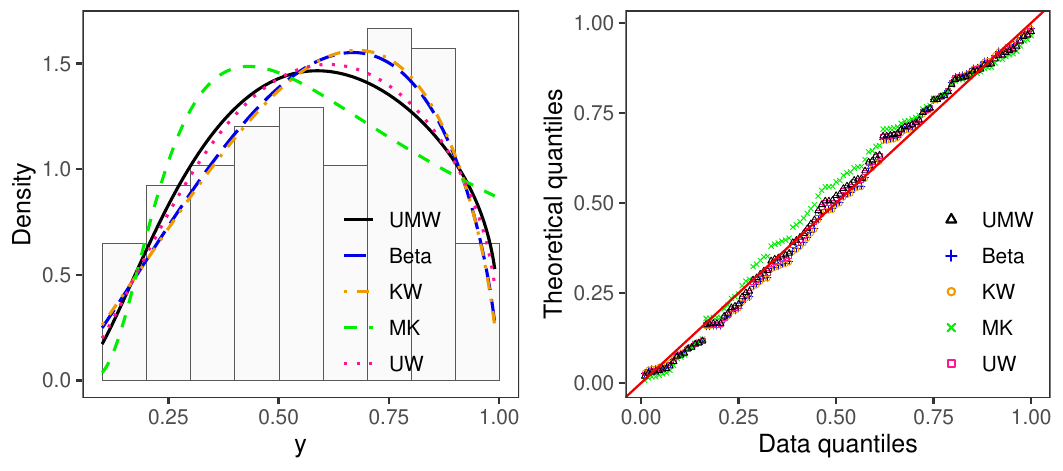}}
\subfigure[Emborcação]{\includegraphics[width=0.45\textwidth]{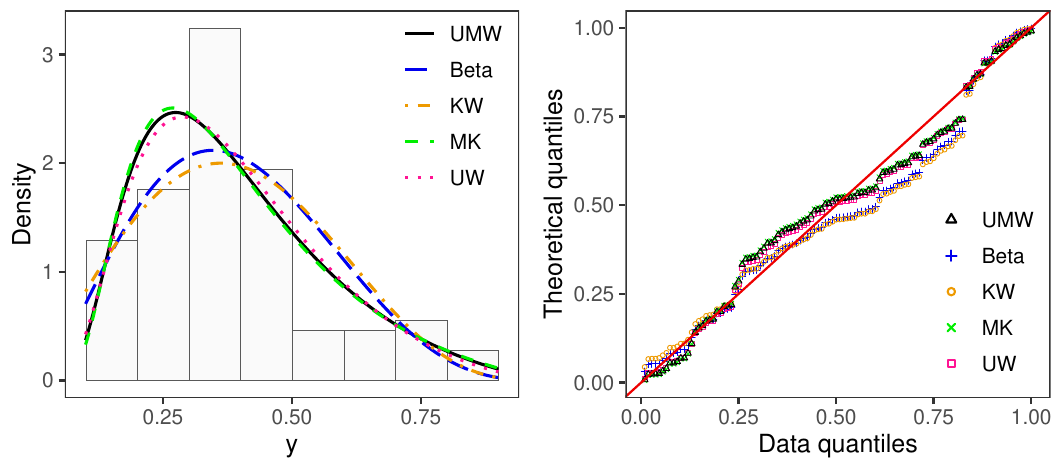}}
\subfigure[Furnas]{\includegraphics[width=0.45\textwidth]{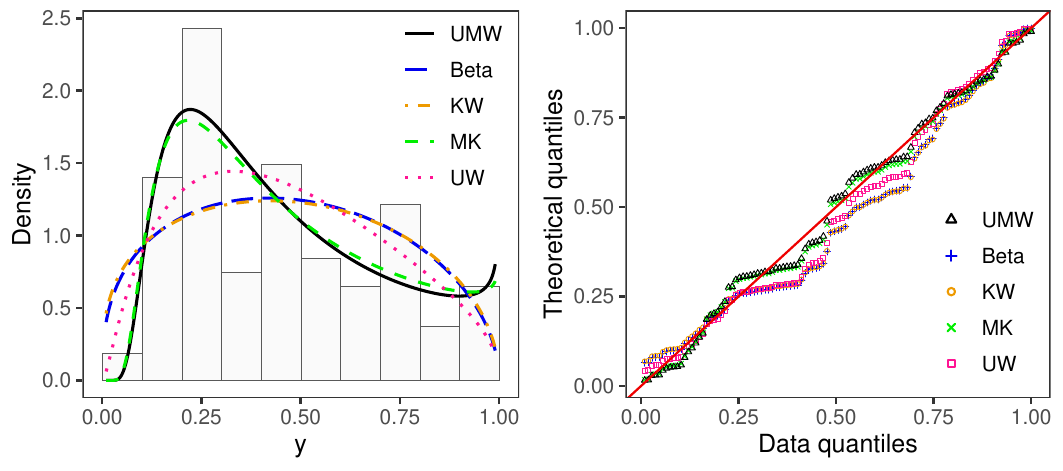}}
\subfigure[Itumbiara]{\includegraphics[width=0.45\textwidth]{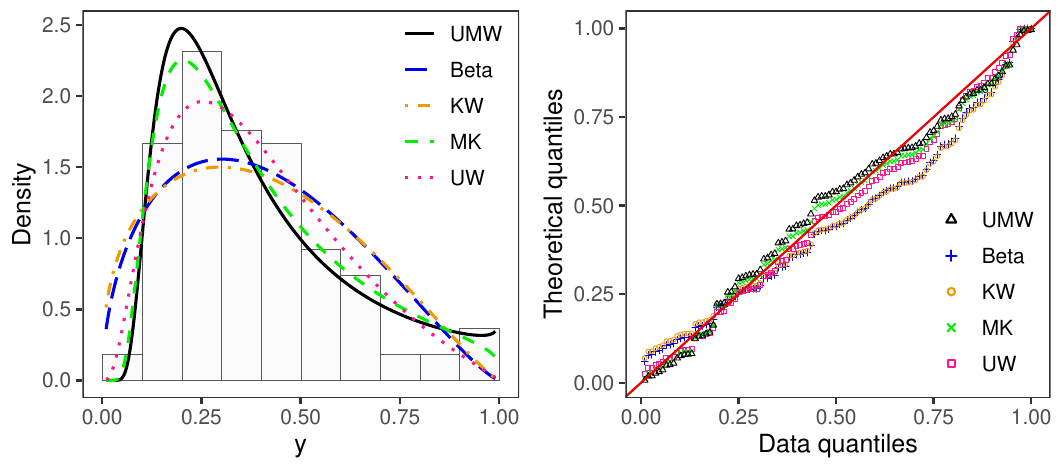}}
\subfigure[Nova Ponte]{\includegraphics[width=0.45\textwidth]{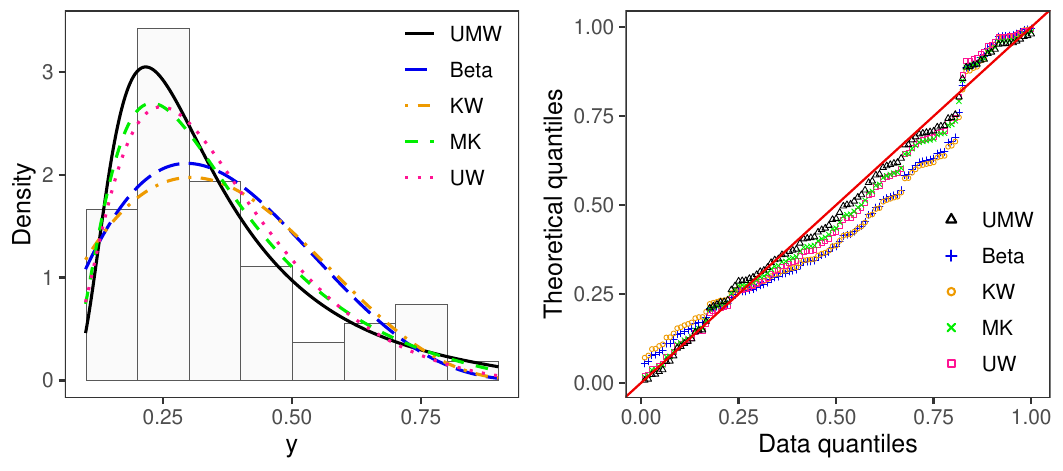}}
\subfigure[São Simão]{\includegraphics[width=0.45\textwidth]{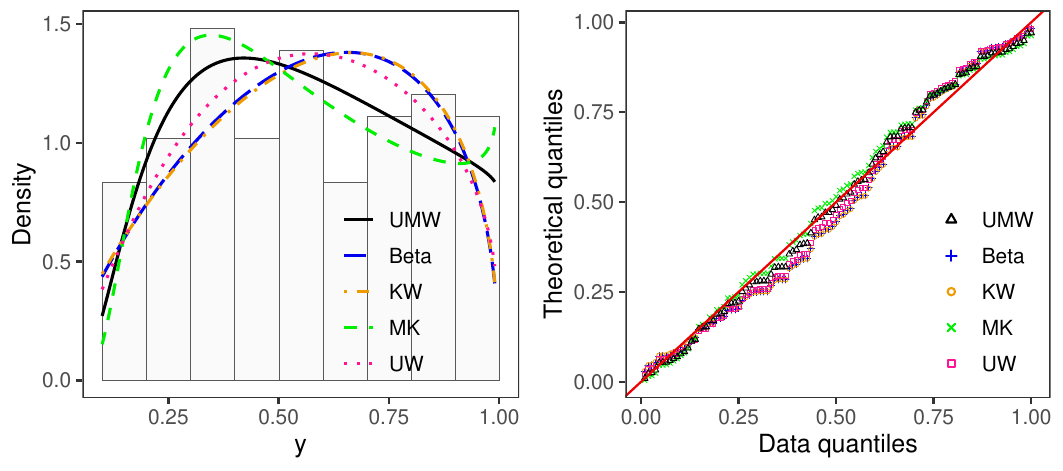}}
\subfigure[Serra do Facão]{\includegraphics[width=0.45\textwidth]{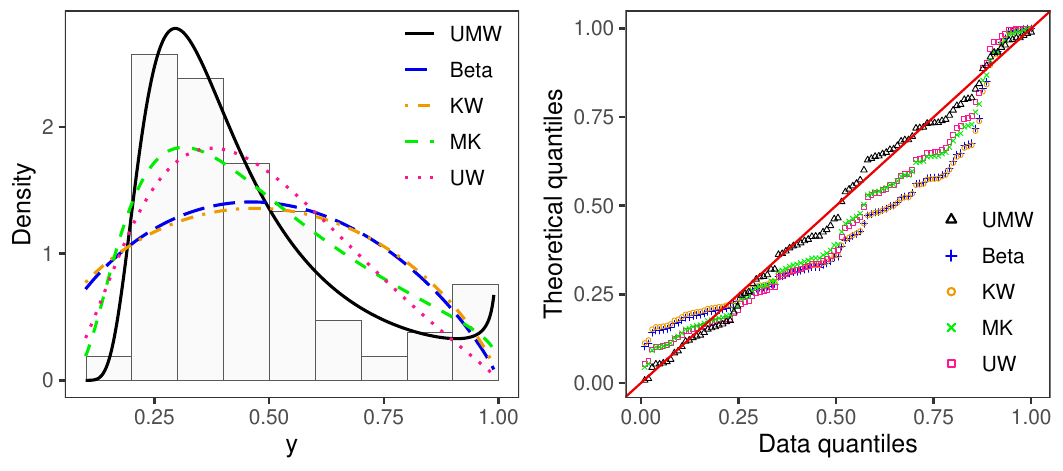}}
\caption{Histograms, density plots, and QQ plots for the relative useful volume of the reservoirs in the SE and CO regions.}\label{fig:g_SeMw}
\scriptsize{Source: Authors.}%
\end{figure}

\end{document}